\documentclass[12pt,a4paper]{article} 

\usepackage[top=3cm, bottom=3cm, left=2.5cm, right=2.5cm]{geometry}

\usepackage[utf8]{inputenc} 
\usepackage[english]{babel}
\usepackage{authblk} 

 \usepackage{newtxtext}

\makeatletter
\renewcommand{\section}{\@startsection{section}{1}{\z@}%
                       {-3.5ex \@plus -1ex \@minus -.2ex}%
                       {2.3ex \@plus.2ex}%
                       {\normalfont\large\bfseries}}
\renewcommand{\subsection}{\@startsection{subsection}{2}{\z@}%
                       {-3.25ex \@plus -1ex \@minus -.2ex}%
                       {1.5ex \@plus .2ex}%
                       {\normalfont\normalsize\bfseries}}
\makeatother

 \usepackage{latexsym}
 \usepackage{amsmath}
 \usepackage{amsfonts}
 \usepackage{amssymb}
 \usepackage{cases}
 \usepackage{pdflscape}  
\usepackage{graphicx}
\usepackage[normalem]{ulem}

\usepackage[T1]{fontenc}
\usepackage[utf8]{inputenc}

\usepackage[unicode=true]{hyperref}
\hypersetup{
  pdfstartview={FitH},
  pdffitwindow=false,
  pdfborder={0 0 0},
  colorlinks=true,
  linktoc=page,
  linkcolor=blue,
  urlcolor=blue,
  hypertexnames=false,
  pdfdisplaydoctitle=true,
  pdfduplex=DuplexFlipLongEdge,
}

\input{amssym.def}
\input{amssym.tex}

\newtheorem{theoreme}{Theorem }[section]
\newtheorem{proposition}[theoreme]{Proposition}
\newtheorem{lemma}[theoreme]{Lemma}
\newtheorem{definition}[theoreme]{Definition}

\newtheorem{remark}[theoreme]{Remark}

\newcommand{\beq}{\begin{equation}}
\newcommand{\eeq}{\end{equation}}
\newcommand{\bes}{\begin{subequations}}
  \newcommand{\ees}{\end{subequations}}
\newcommand{\gen}{\mathrm{gen}}

\def\bel{\begin{lemma}}
\def\eel{\end{lemma}}
\def\bet{\begin{theoreme}}
\def\eet{\end{theoreme}}
\def\bed{\begin{definition}}
\def\eed{\end{definition}}
\def\bep{\begin{proposition}}
\def\eep{\end{proposition}}
\def\ber{\begin{remark}}
\def\eer{\end{remark}}

\numberwithin{equation}{section}

\newcounter{smallarabics}
\newenvironment{arabicenumerate}
{\begin{list}{{\normalfont\textrm{(\arabic{smallarabics})}}}
  {\usecounter{smallarabics}\setlength{\itemindent}{0cm}
   \setlength{\leftmargin}{5ex}\setlength{\labelwidth}{4ex}
   \setlength{\topsep}{0.75\parsep}\setlength{\partopsep}{0ex}
   \setlength{\itemsep}{0ex}}}
{\end{list}}

\newcounter{smallroman}

\newcommand{\ben}{\begin{arabicenumerate}}
\newcommand{\een}{\end{arabicenumerate}}

\def\rr{{\mathbb R}}
\def\hh{{\mathbb H}}
\def\SS{{\mathbb S}}
\def\zz{{\mathbb Z}}
\def\cc{{\mathbb C}}

\def\nn{{\mathbb N}}
\def\dd{{\mathrm d}}
\def\ii{{\rm i}}

\def\bbbone{{\mathchoice {\rm 1\mskip-4mu l} {\rm 1\mskip-4mu l}
{\rm 1\mskip-4.5mu l} {\rm 1\mskip-5mu l}}}
\def\one{\bbbone}

\let\Re\relax
\DeclareMathOperator{\Re}{Re}

\def\0{{\rm\scriptscriptstyle 0}}

\def\pp{{\mathbb P}}

\def\ge{\mathrm{gen}\!\!}

\def\i{{\rm i}}

\def\sgn{{\rm sgn}}

\def\e{{\rm e}}
\def\d{{\rm d}}

\newcommand{\RR}{\mathrm{R}}
\newcommand{\II}{\mathrm{I}}

\def\slim{{\rm s-}\lim}

\def\12{\frac{1}{2}}
\def\32{\frac{3}{2}}
\def\52{\frac{5}{2}}

\def\qed{\hfill $\Box$\medskip}
\def\proof{{\bf Proof.}\ \ }

\newcommand\s{\mathrm{s}}
  \newcommand\h{\mathrm{h}}

\newcommand\ee{\mathrm{e}}

\begin{document}


\title {Point potentials  
on  Euclidean space,\\ hyperbolic space
  and sphere 
in any dimension
}

\author[1]{Jan Dereziński}
\author[1]{Christian Ga\ss}
\author[2]{Błażej Ruba}

\affil[1]{Department of Mathematical Methods in Physics, Faculty of Physics, \protect\\
University of Warsaw, Pasteura 5, 02-093 Warszawa, Poland, \protect\\ 
email: jan.derezinski@fuw.edu.pl, christian.gass@fuw.edu.pl}
\affil[2]{Department of Mathematics, University of Copenhagen, \protect\\ 
Universitetsparken 5, DK-2100 Copenhagen Ø, Denmark, \protect\\
email: btr@math.ku.dk}
\date{\today}
\maketitle


  \begin{abstract}
In dimensions $d=1,2,3$ the Laplacian can be
perturbed by a point potential. In higher dimensions the Laplacian
with a point potential cannot be defined as a~self-adjoint operator. 
However, for any dimension there exists a~natural family of functions that can be
 interpreted as Green's functions of the Laplacian with a spherically symmetric point potential.
In~dimensions $1,2,3$ they are the integral kernels of the resolvent of well-defined
self-adjoint operators. In higher dimensions they are not even
integral kernels of
bounded operators. Their construction uses the so-called generalized integral,
a concept going back to Riesz and Hadamard.

We consider the Laplace(-Beltrami) operator on the Euclidean space,
the hyperbolic space and the sphere in any dimension. We describe 
the corresponding  Green's functions, also perturbed by  a point 
potential. We describe their limit as the  scaled hyperbolic space 
and the scaled sphere approach the Euclidean space. Especially 
interesting is the behavior of positive eigenvalues of the spherical 
Laplacian, which undergo a shift
proportional to a negative power of the radius of the sphere.

 We expect that in any dimension our constructions yield possible  behaviors
  of the integral kernel of the resolvent of a perturbed Laplacian
  far from the  support of the perturbation. Besides,
   they can be viewed as toy models illustrating various aspects of renormalization in Quantum Field Theory, especially  the
 point-splitting method and  dimensional regularization.
\end{abstract}

\section{Introduction}
\subsection{Euclidean, hyperbolic and spherical Laplacian}
Let $\hh^d$ and $\SS^d$ denote the {\em hyperbolic space},
resp. the {\em unit sphere}, both $d$-dimensional. Let $\Delta_d$, $\Delta_d^\h$,
$\Delta_d^\s$ denote the {\em Laplace(-Beltrami) operators} on
$\rr^d$, $\hh^d$, resp. $\SS^d$.
It is convenient to shift the {\em hyperbolic Laplacian} by
$-\frac{(d-1)^2}{4}$ and the {\em spherical Laplacian} by $\frac{(d-1)^2}{4}$.
Our paper is devoted to the operators
\begin{align}
\label{eq:def_Hd}
  H_d:=-\Delta_d,\quad H_d^\h:=-\Delta_d^\h-\frac{(d-1)^2}{4},\quad H_d^\s:=-\Delta_d^\s+\frac{(d-1)^2}{4},
\end{align}
possibly perturbed by a {\em point potential}.

The operators $H_d$, $H_d^\h$, $H_d^\s$ can be viewed as self-adjoint
operators on $L^2(\rr^d)$, $L^2(\hh^d)$, resp. $L^2(\SS^d)$.
For $z\in\cc$ outside of the spectrum of $H_d$, $H_d^\h$,
resp. $H_d^\s$
one can define their {\em resolvent} ({\em Green's operator})
\begin{align}
 G_d(z):=(-z+H_d)^{-1},\quad G_d^\h(z):=(-z+H_d^\h)^{-1},\quad 
  G_d^\s(z):=(-z+H_d^\s)^{-1}.
\end{align}

The spectrum of $H_d$ and  $H_d^\h$ is continuous and 
coincides with $[0,\infty[$. 
The spectrum of $H_d^\s$ is discrete and equals $\left\{  \left(l+\frac{d-1}{2}\right)^2  \
  |\ l=0,1,\dots\right\}\subset [0,\infty[$. Therefore, it is
often   convenient to represent the spectral parameter $z\in\cc\backslash [0,\infty[$ as $z=-\beta^2$ with
    $\Re\beta>0$, so that
\begin{align}    
\label{eq:resolvents_beta}
        G_d(-\beta^2)=
        (\beta^2+H_d)^{-1},\quad
          G_d^\h(-\beta^2)=
    (\beta^2+H_d^\h)^{-1},\quad     G_d^\s(-\beta^2)=
    (\beta^2+H_d^\s)^{-1}.
\end{align}
Sometimes we will also write $\zeta^2$ for $z$.

For $0\leq a<b$ one can define the {\em spectral projections onto $[a,b]$}:
\begin{align}
  \pp_d(a,b):=\one_{[a,b]}(H_d),\quad \pp_d^\h(a,b):=\one_{[a,b]}(H_d^\h),\quad\pp_d^\s(a,b):=\one_{[a,b]}(H_d^\s).\label{specfunc}
\end{align}
We can also introduce the  {\em spectral
projections
onto eigenvalues} of $H_d^\s$:
\begin{align} \label{specfunc1}
\pp_{d,l}^\s:=\one_{ \left(l+\frac{d-1}{2}\right)^2 }(H_d^\s).
\end{align}

The integral kernels of the resolvents \eqref{eq:resolvents_beta}, 
denoted by $G_d(-\beta^2;x,x')$, $G_d^{\h}(-\beta^2;x,x')$ and 
 $G_d^{\s}(-\beta^2;x,x')$, are often called {\em Green's
   functions}. The integral kernels of the spectral projections
 \eqref{specfunc} are denoted
 $\pp_d(a,b;x,x')$, $\pp_d^\h(a,b;x,x')$, $\pp_d^\s(a,b;x,x')$. 
  The integral kernel of  \eqref{specfunc1} is denoted
 $\pp_{d,l}^\s(x,x')$. Explicit formulas for these in terms of special functions are known, and for convenience of the reader we provide them in our paper.

The integral kernels  related to $H_d$ are expressed in terms of
functions from  the Bessel  family. The integral kernels related to $H_d^\h$ and $H_d^\s$ are 
expressed in terms of Gegenbauer functions. Here are, for instance,
the formulas for Green's functions:
\begin{align}
       G_d\big(-\beta^2;x,x'\big)&
     =\frac{1}{(2\pi)^{\frac{d}{2}}}\Big(\frac\beta{r}\Big)^{\frac{d}{2}-1}K_{\frac{d}{2}-1}
         \big(\beta r\big), \label{green1}\\
     G_{d}^\h\big(-\beta^2;x,x'\big)&
     =\frac{\sqrt\pi\Gamma(\frac{d-1}{2}+\beta)}{\sqrt2(2\pi)^{\frac{d}{2}}2^{\beta}
         }{\bf Z}_{\frac{d}{2}-1,\beta }
  \big(\cosh(r)\big), \label{green2}\\
  G_{d}^\s\big(-\beta^2;x,x'\big)&
     =\frac{\Gamma(\frac{d-1}{2}+\i\beta )
       \Gamma(\frac{d-1}{2}-\i\beta )
     }{(4\pi)^{\frac{d}{2}}
  }{\bf S}_{\frac{d}{2}-1,\i\beta }         \big(-\cos(r)\big). \label{green3}
  \end{align}
Above, $r$ denotes the Euclidean, hyperbolic, resp. spherical distance
between $x$ and $x'$. $K_\alpha$ is the Macdonald function (one of
functions from the Bessel family). ${\bf 
  S}_{\alpha,\lambda}$ and ${\bf Z}_{\alpha,\lambda}$ are two kinds of
Gegenbauer
functions, see Appendix \ref{app:gegenbauer}.

One should 
note that Bessel and Gegenbauer functions have special properties  when their parameter $\alpha$
is  half-integer or integer.
For half-integer $\alpha$  Bessel and Gegenbauer functions can
be expressed as elementary functions. For integer $\alpha$ Bessel and
Gegenbauer functions have a~logarithmic singularity.
From the point of view of Green's operators, 
these values are important:
half-integer $\alpha$ is used in odd dimensions and integer $\alpha$
in  even dimensions.

All Green's 
functions 
\eqref{green1}, \eqref{green2} and \eqref{green3} behave similarly for 
$x,x'$ close to one another, which follows from well-known expansions 
of $K_\alpha$, ${\bf 
  S}_{\alpha,\lambda}$ and ${\bf Z}_{\alpha,\lambda}$. 
However, for large 
distances they are rather different.  This can be seen by comparing
the expansions of \eqref{green1}, \eqref{green2} and \eqref{green3}
for large distances, which we describe in 
\eqref{explicit4}, \eqref{integg4} and \eqref{evo1}.

\subsection{Point potentials}

The main goal of this paper  is  to extend the 
above theory to the operators $H_d$, $H_d^\h$ and $H_d^\s$ perturbed 
by a {\em point potential} (also called a {\em contact} or {\em delta potential}). 
It is a well known fact that the 1-dimensional Laplacian can be
perturbed by a delta potential in the form sense \cite{RSII}. In~dimensions 2 and
3 the Laplacian can also be perturbed by a
point-like perturbation, however one cannot use the naive form
formalism anymore \cite{AGHH,AK,BF}.  Thus in dimensions $d=1,2,3$ 
we obtain 1-parameter families of self-adjoint operators $H_d^\gamma$, 
$H_d^{\h,\gamma}$, $H_d^{\s,\gamma}$. We denote their 
resolvents by $G_d^\gamma(z)$, $G_d^{\h,\gamma}(z)$, 
$G_d^{\s,\gamma}(z)$. Their integral kernels have the form
\begin{align}
 \label{eq:Gdgamma+}
  G_d^\gamma( z ;x,x')&=   G_d( z ;x,x')+ 
                          \frac{ G_d( z;x,x_0)G_d( z;x_0,x')}{\gamma+\Sigma_d(z)},\\ \label{eq:Gdgamma+h}
    G_d^{\h,\gamma}( z;x,x')&=   G_d^\h( z;x,x')+ 
  \frac{
                                 G_d^\h( z;x,x_0)G_d^\h( z;x_0,x')}{\gamma+\Sigma_d^\h(z)},\\ \label{eq:Gdgamma+s}
      G_d^{\s,\gamma}( z;x,x')&=   G_d^\s( z;x,x')+ 
  \frac{ G_d^\s( z;x,x_0)G_d^\s( z;x_0,x')}{\gamma+\Sigma_d^\s(z)},
\end{align}
where $x_0$ is the position of the point potential
(e.g. the origin of coordinates of $\rr^d$ or the north pole of $\SS^d$).
 Here, the functions $\Sigma_d$, $\Sigma_d^\h$ and $\Sigma_d^\s$  satisfy
\begin{align}
  \partial_z\Sigma_d(z)=
&-\int_{\rr^d}G_d(z;x_0,x)^2\d x,\label{self1}\\
\partial_z\Sigma_d^\h(z)=
  &-\int_{\hh^d}G_d^\h(z;x_0,x)^2\d x,\label{self2}\\
  \partial_z\Sigma_d^\s(z)=
  &-\int_{\SS^d}G_d^\s(z;x_0,x)^2\d x.\label{self3}
\end{align}
 The parameter $\gamma\in\rr\cup\{\infty\}$  is a real  integration
constant  and  describes the strength of the 
perturbation.
The function $\gamma +\Sigma_d^\bullet(z)$, where $\bullet$ is empty,  $\h$
or $\s$, will be called the {\em full self-energy}. $\Sigma_d^\bullet(z)$ is the  {\em reference
  self-energy}, fixed by  imposing some additional conditions.

In dimensions $1$ and $3$ there exists a natural condition that
allows us to fix the reference self-energy:
$\lim\limits_{z\to-\infty}\Sigma_1^\bullet(z)=0$ and
$\lim\limits_{z\to-\infty}(\Sigma_3^\bullet(z)-\frac{\sqrt{-z}}{4\pi})=0$.

For $d=2$, one possible choice for the reference self-energy 
is to demand $\Sigma_2^\bullet(-\beta^2)\sim\frac{\ln\beta}{2\pi}$ 
for $\beta\to\infty$, or equivalently, $\Sigma_2^\bullet(-1)=0$. 
This, however, distinguishes a certain length scale corresponding to 
$\beta=1$. In order to avoid such an \emph{a priori} unphysical 
distinction, we treat all possible full self-energies on an equal 
footing as members of a family of reference self-energies parametrized 
by a real parameter $\varepsilon=-2\pi\gamma$:
\begin{align}
\gamma+\Sigma_2^{\bullet}(z)
=:\Sigma_2^{\bullet,\varepsilon}(z).
\end{align}
$\gamma$ (and $\varepsilon$ in $d=2$) are
closely related to the so-called {\em scattering
  length} $a$ used in the physical literature.
Here are the relations between these two parameters:
\begin{align}d=1,&\qquad a=-2\gamma;\\
  d=2,& \qquad a=\e^{2 \pi \gamma}=\e^{-\varepsilon} ;\\
  d=3,&\qquad  a=-\frac{1}{4\pi\gamma}.\end{align}

It is  well-known that the Laplacian is essentially self-adjoint 
on $C_\mathrm{c}^\infty(\rr^d\backslash\{0\})$ in dimensions $d\geq4$ \cite{RSII}.
In other words, there are no 
point-like perturbations of the Laplacian in dimensions $d\geq4$, if we
stick to the usual Hilbert space setting.
This corresponds to the divergence of the integrals in \eqref{self1},
\eqref{self2} and \eqref{self3} defining the self-energies.

The description of Green's functions for the Laplacians with a 
point potential in dimensions $d\geq4$ is probably the
  main novelty of our paper. Our starting points are equations
  \eqref{eq:Gdgamma+}, \eqref{eq:Gdgamma+h} and
  \eqref{eq:Gdgamma+s}. Hence, we need to give meaning to divergent
  self-energies. We will  
 consider two different but consistent methods to do this.
  The first will be called the {\em point splitting
  method}, and the second the {\em minimal subtraction method}.

In the first method we start with replacing the integrals 
\eqref{self1}, \eqref{self2},
\eqref{self2} by their ``point-split versions'', which are then
repeatedly differentiated in $z$ (the ``energy'') until convergent integrals are obtained.
Then we repeatedly integrate them to
get the self-energy. Integration constants from multiple integrations
can be gathered in a polynomial $\gamma(z)$, which replaces the
integration constant $\gamma$ used in lower dimensions. 
$\gamma(z)$ is 
a polynomial of degree $\leq n=\big\lfloor \frac{d-2}{2}\big\rfloor $, 
i.e. $n=\frac{d-3}{2}$ if $d$ is odd and degree $n=\frac{d-2}{2}$ if $d$ 
is even.

The second approach to define self-energies is to replace \eqref{self1},
\eqref{self2}, \eqref{self3} with the corresponding {\em generalized 
  integrals}. Then the self-energies  $\Sigma_d^\bullet(z)$ are well 
  defined in all dimensions up to only  one integration constant.

  As we explain in  Appendix \ref{app:genInts}, the 
  generalized integral is a natural  extension 
  of the classical integration to a certain class of not
  necessarily integrable functions. It resembles the minimal 
  subtraction scheme in QFT. Clearly, it is only one of many
  linear extensions of the integration functional. Other extensions
  differ by
  an
  additional polynomial of degree $\leq \big\lfloor \frac{d-2}{2}\big
  \rfloor $,  whose parameters 
  can be viewed as arbitrary ``renormalization constants''. Thus both approaches to 
  defining self-energies agree.

A generalized integral is said to have a {\em scaling anomaly} if it transforms inhomogeneously upon a rescaling of the integration variable. There is a big difference between
non-anomalous and anomalous generalized integrals. In the non-anomalous
case, the computation of a generalized integral essentially reduces to the
analytic continuation of the usual integral in a
certain parameter, which often (in particular, in our case) can be
interpreted as the dimension. In the anomalous case in addition to
analytic continuation one has to perform an appropriate subtraction.

One could ask whether it is natural to fix a certain full self-energy,
and to call it the {\em reference
  self-energy}.
We would like our reference self-energies to be algebraically as simple as possible, in particular,
they should be factorized in simple factors.

The generalized integral suggests a certain
expression, which we denote $\Sigma_d^{\bullet,\mathrm{ms}}$, (where
we fix a single  integration constant in some natural way and
$\mathrm{ms}$ stands for ``minimal subtraction'').
In odd dimensions $d\geq5$ there is an obvious choice of reference
self-energy which is given by an
algebraically simple expression. 
 This reference self energy is equal to $\Sigma_d^{\bullet,\mathrm{ms}}$, and can also be obtained by formally extending
$\Sigma_d^\bullet(z)$ to complex $d$ in the region 
 $|\Re d-2|<2$ ($d \neq 2$), and
then by using analytic continuation. 

In even dimensions $d\geq4$ we are in the
anomalous case, which is much more complicated. The generalized
integrals on the right-hand side of \eqref{self1}, \eqref{self2} and \eqref{self3}
involve non-elementary functions: the logarithm or the digamma
function $\psi(z):=\frac{\Gamma'(z)}{\Gamma(z)}$.

The anomalous generalized integral is not invariant under a change 
of variables. In the Euclidean case, the natural variable is 
$r$, the distance from the origin in some fixed units.
Since the generalized integral is invariant under a change of variable $r\to r^\alpha$ for any $\alpha>0$, one can equivalently use the coordinate $r^2$.

In the hyperbolic 
and spherical cases the variables $r$ (or $r^2$), now denoting the hyperbolic 
resp. spherical distance, seem not convenient to compute
self-energies. Instead, in \cite{DGR23a}, to this end
we used the variables $w=2(\cosh(r)-1)$ resp. $w=2(1-\cos(r))$. These variables are convenient in calculations involving resolvents of the Laplacian, and they seem to be a natural choice. Note that $w=r^2 + \mathcal{O}(r^4)$ is a function of $r^2$ in both cases.
Anyway, if we change the variable in the generalized integral 
according to \eqref{eq:change_of_var}, the resulting change in 
the self-energy is a polynomial of degree $\leq \big\lfloor 
\frac{d-2}{2}\big\rfloor $, which is consistent with the ambiguity 
in the point splitting approach.\footnote{We remark that it 
makes a difference whether we allow for coordinate changes that are functions 
of $r$ or only of $r^2$. In the latter case, the generalized integral transforms anomalously only in even dimensions. In the 
former case, it transforms anomalously also in odd dimensions, 
giving a polynomial freedom in any dimension. Since 
$w$ is a function of $r^2$, the change of variables 
$r^2\to w$ does not affect the generalized integral in odd dimensions.}

Thus, for even $d\geq4$,
selecting in some way the integration constant, we can introduce
the self-energy given by the generalized integral
$\Sigma_d^{\bullet,\mathrm{ms}}$. All self-energies are given by
 $\gamma(z)+\Sigma_d^{\bullet,\mathrm{ms}}(z)$, where
$\gamma$ is of
degree  $\leq \frac{d-2}{2}$.
In the hyperbolic and
spherical cases $\Sigma_d^{\bullet,\mathrm{ms}}$ is rather complicated and
has no obvious factorization. There exists, however,
a 1-parameter family of factorized expressions 
  $\Sigma_d^{\bullet,\varepsilon}(z)$, $\varepsilon \in \rr$, which
one can use as reference self-energies. We
absorb the highest
term of the polynomial $\gamma$ in  $\varepsilon$, so that now the remaining freedom
consists of a polynomial $\eta(z)$ of degree only $\leq\frac{d-4}{2}$. Thus
the general form of a full self-energy in even dimensions is now given by 
$\eta(z)+\Sigma_d^{\bullet,\varepsilon}(z)$.

Summarizing, for odd $d$ we obtained the families of functions
\begin{align}
  G_d^\gamma( z ;x,x')&=   G_d( z ;x,x')+ 
    \frac{ G_d( z;x,x_0)G_d( z;x_0,x')}{\gamma(z)+\Sigma_d(z)},\label{uyu1}\\
    G_d^{\h,\gamma}( z;x,x')&=   G_d^\h( z;x,x')+ 
  \frac{G_d^\h( z;x,x_0)G_d^\h( z;x_0,x')}{\gamma(z)+\Sigma_d^\h(z)},\label{uyu2}\\
      G_d^{\s,\gamma}( z;x,x')&=   G_d^\s( z;x,x')+ 
  \frac{ G_d^\s( z;x,x_0)G_d^\s( z;x_0,x')}{\gamma(z)+\Sigma_d^\s(z)},\label{uyu3}
\end{align}
parametrized by an arbitrary
polynomial $\gamma$ of degree $\leq\frac{d-3}{2}$.

For even $d$ we need to slightly modify \eqref{uyu1}, \eqref{uyu2} and
\eqref{uyu3}: we replace the superscript $\gamma$ with
$\varepsilon,\eta$ and $\gamma(z)+\Sigma_d^\bullet(z)$ with 
$\eta(z)+\Sigma_d^{\bullet,\varepsilon}(z)$. Here $\varepsilon$ is 
a real number and $\eta$ is an arbitrary polynomial of degree 
$\leq\frac{d-4}{2}$.

In what follows, abusing the notation, we will sometimes write 
$\gamma$ for a pair $\varepsilon,\eta$. $G_d^{\bullet,\gamma}(z)$ 
will be called {\em Green's functions}.  For  $d\geq4$ they are not 
integral kernels of bounded operators. Hence for such $d$,  they are 
not resolvents of  well-defined self-adjoint operators.

Here is the list of the reference self-energies in various dimensions:
 \begin{subequations}
\begin{align}\label{d1}
  d=1:\qquad&\Sigma_1(-\beta^2) =-\frac{1}{2\beta},\\\notag
            &\Sigma_1^\h(-\beta^2) =-\frac{1}{2\beta},\\\notag
            &\Sigma_1^\s(-\beta^2)=-\frac{\coth\pi\beta}{2\beta};
\\  
\label{d2}
  d=2:\qquad&\Sigma_2^\varepsilon(-\beta^2 )=\frac{1}{2\pi}\big(\ln\beta-\varepsilon\big),\\\notag
            &\Sigma_2^{\h,\varepsilon}(-\beta^2) =\frac{1 }{2\pi}\Big(\psi(\tfrac12+\beta)-\varepsilon\Big),\\\notag
            &\Sigma_2^{\s,\varepsilon}(-\beta^2)=\frac{1 }{4\pi}\Big(\psi(\tfrac12+\i\beta)+\psi(\tfrac12-\i\beta)-2\varepsilon\Big);
  \\           
             \label{d3}
  d=3:\qquad&\Sigma_3(-\beta^2 )=\frac{\beta}{4\pi},\\\notag
            &\Sigma_3^\h(-\beta^2) =\frac{\beta}{4\pi},\\\notag
            &\Sigma_3^\s(-\beta^2)=\frac{\beta\coth\pi\beta}{4\pi};
 \\
    \label{d4}
  \text{even }d\geq4:\qquad&\Sigma_d^\varepsilon(-\beta^2) = \frac{1
    }{(4
  \pi)^{\frac{d}{2}} \Gamma\big(\tfrac{d}{2}\big)} \big( \ln(\beta^2)-2\varepsilon  \big)
           (-\beta^2)^{\frac{d-2}{2}}         ,\\\notag
  &\Sigma_d^{\h,\varepsilon}(-\beta^2)=
    \frac{
  \psi\big(\tfrac{3-d}{2}+\beta\big)
  +\psi\big(\tfrac{d-1}{2}+\beta\big)
  -2\varepsilon
    }{ (4\pi)^{\tfrac{d}{2}} \Gamma\big(\tfrac{d}{2}\big)}
                       \notag
\prod_{j=0}^{\tfrac{d-4}{2}} 
 \Big(-\beta^2+\big(\tfrac12 + j\big)^2\Big),\\\notag
  &\Sigma_d^{\s,\varepsilon}(-\beta^2)=
    \frac{\psi\big(\tfrac{d-1}{2}+\i\beta\big)+\psi\big(\tfrac{d-1}{2}-\i\beta\big)-2\varepsilon
    }
{(4\pi)^{\frac{d}{2}}\Gamma(\frac{d}{2})}
 \prod_{j=0}^{\frac{d-4}{2}}\big(-\beta^2-(\tfrac12+j)^2\big) ;
 \\             
              \label{d5}
\text{odd }d\geq5:\qquad&
                          \Sigma_d(-\beta^2)=\frac{
                         \pi}{(4\pi)^{\frac{d}{2}}\Gamma(\frac{d}{2})}\beta (-\beta^2)^{\frac{d-3}{2}},\\\notag
& \Sigma_d^\h(-\beta^2) 
  =
                                                                                                                  \frac{
                                                                                   \pi }{ (4\pi)^{\tfrac{d}{2}} \, 
 \Gamma\big(\tfrac{d}{2}\big)}   \beta 
      \prod_{k=1}^{\frac{d-3}{2}}
    \big( -\beta^2 +k^2 \big),\\\notag
 &\Sigma_{d}^\s(-\beta^2) 
   = \frac{
   \pi  \coth(\pi\beta ) }
   { (4\pi)^{\frac{d}{2}}\Gamma\big(\tfrac{d}{2}\big)}  \beta 
      \prod_{k=1}^{\frac{d-3}{2}}
    \big( -\beta^2 -k^2 \big)
.\end{align}
\end{subequations}

Of course, some items of the above list are well-known. The 
self-energy in the Euclidean case for $d=1,2,3$ belongs to standard 
knowledge of contemporary quantum physics. The Euclidean self-energy for 
$d\geq4$ obtained with help of the generalized integral is partially covered in the literature, see e.g. \cite{KuPa} for odd 
dimensions. 
The self-energies for the hyperbolic and spherical Laplacian appear to 
be new.

As we stressed above, for $d\geq4$, the functions  \eqref{eq:Gdgamma+},
 \eqref{eq:Gdgamma+h}  and  \eqref{eq:Gdgamma+s} do not define bounded operators 
and their
inverses do not define self-adjoint operators. It is natural to ask
what is their meaning.

One approach that can be found in the literature is to extend the
Hilbert space, typically, to a Pontryagin space (with an indefinite
metric product). This approach is described e.g. in \cite{Ku}.

One can also consider a different interpretation. Fix a point $x_0$ in $\rr^d$, $\hh^d$ or $\SS^d$. Suppose that $H_d^\bullet+V$ is a self-adjoint operator obtained by perturbing $H_d$ in a ball around $x_0$ of small radius $r$. We expect that far away from that ball, the integral kernel of $(H_d^\bullet+V-z)^{-1}$ is well approximated by $G_{d}^{\bullet, \gamma}$ for some $\gamma$ determined by $V$. Thus coefficients of $\gamma$ summarize universal long distance properties of $V$. We will discuss this idea further in a separate paper, which is in preparation.

The reference self-energies (corresponding to $\gamma=0$ in 
odd dimensions and $\eta=0$ in even dimensions) are in some sense 
distinguished---the poles of the corresponding Green's functions can be 
easily computed. In the Euclidean case they are also distiguished by their 
scaling property (they are ``fixed points of the renormalization group'').

 Our analysis of point interactions in dimensions
$d\geq4$ resembles
renormalization in Quantum Field Theory. In QFT,
especially in the Wilsonian approach, one does not worry too much
whether the quantities computed by renormalization techiques
correspond to a well-defined Hamiltonian. They should reproduce the
``infrared behavior of correlation functions''. We apply a similar
philosophy
to Green's functions. Note in particular that the borderline case when
the perturbed
Green's functions 
do not correspond to self-adjoint operators is $d=4$---the physical
dimension of our spacetime. (Our spacetime has a Lorentzian signature,
however using the Wick rotation it can often be replaced by
the Euclidean $\rr^4$.)

Our analysis can be viewed as a toy model illustrating various
aspects of renormalization in QFT. As explained above, we use two
methods to define
self-energies. The first  applies the so-called {\em point-splitting}
and then  regularization by
{\em differentiation in the energy}. The second method, using 
generalized integrals, resembles the {\em minimal subtraction}
method. To compute them we use
{\em dimensional regularization}.
Both methods have their widely used counterparts in QFT.

\subsection{Flat limit}

Let $R>0$ and let $\hh_R^d$, $\SS_R^d$ denote the rescaled hyperbolic space of 
curvature $-\frac{1}{R^2}$, resp. the rescaled sphere of curvature 
$\frac1{R^2}$ (that means, of radius $R$). 
Intuitively it is clear that in some sense  $\hh_R^d$, $\SS_R^d$
converge to $\rr^d$ as $R\to\infty$.

Green's functions on the rescaled spaces are
   \begin{align}
         G_{d,R}^\h(-\beta^2;x,x')&=R^{-d+2}G_d^\h\Big(-(\beta
                                    R)^2,\frac{x}{R},\frac{x'}{R}\Big),\\
              G_{d,R}^\s(-\beta^2;x,x')&=R^{-d+2}G_d^\s\Big(-(\beta
                                         R)^2,\frac{x}{R},\frac{x'}{R}\Big).
                                         \end{align}
We describe the convergence of these Green's functions  to the Euclidean ones
$G_d(-\beta^2;x,x')$. This is of course well-known, see e.g. \cite{CDT}.

On the rescaled spaces  the reference self-energies are defined as follows:
\begin{align}
 \Sigma_{d,R}^\h(-\beta^2) &:
 =R^{2-d} 
    \Sigma_{d}^\h\big(-(\beta R)^2\big), &\text{ odd }d, \vspace{12pt} \\[1ex]
    \notag
   \Sigma_{d,R}^{\h,\varepsilon}(-\beta^2) & :=
  R^{2-d}\Sigma_{d}^{\h,\varepsilon+\ln R}\big(-(\beta R)^2\big) 
, & \text{ even }d, 
\\[2ex] \notag
   \Sigma_{d,R}^\s(-\beta^2) &:
 =R^{2-d}
    \Sigma_{d}^\s\big(-(\beta R)^2\big), &\text{ odd }d, \vspace{12pt} \\[1ex]
    \notag
   \Sigma_{d,R}^{\s,\varepsilon}(-\beta^2)
                           &:=R^{2-d}\Sigma_{d}^{\s,\varepsilon+\ln R}\big(-(\beta R)^2\big)
, &\text{ even }d. 
\end{align}
Note that in even dimensions we need an additional additive
renormalization, which can be traced back to rescaling of the variable in a generalized integral.

Using the above self-energies
we define the corresponding Green's functions.
We prove that they converge to 
the Euclidean Green's function with a point potential and
the same parameters. That is, in odd dimensions
$G_{d,R}^{\h\gamma}(-\beta^2;x,x')$ and
$G_{d,R}^{\s,\gamma}(-\beta^2;x,x')$ converge to
$G_d^\gamma(-\beta^2;x,x')$, and in even dimensions
$G_{d,R}^{\h,\varepsilon,\eta}(-\beta^2;x,x')$ and
$G_{d,R}^{\s,\varepsilon,\eta}(-\beta^2;x,x')$ converge to
$G_d^{\varepsilon,\eta}(-\beta^2;x,x')$.
This convergence is, perhaps, not
very surprising. However, it~requires a rather careful treatment of the self-energy
(including the choice of renormalization), especially
for even $d\geq4$, when there is the  scaling anomaly.

\subsection{Poles of Green's functions}

In dimensions 1,2,3, the singularities of Green's functions
$G_d^{\bullet,\gamma}( z )$, $G_d^{\bullet,\varepsilon,\eta}( z )$
are located at the spectrum
$H_d^{\bullet,\gamma}( z )$, $H_d^{\bullet,\varepsilon,\eta}( z )$.
In the Euclidean and hyperbolic case the continuous spectrum remains
$[0,\infty[$,
but the point potential may introduce an additional eigenvalue. In the
spherical case, the point potential shifts the old eigenvalues, and
may introduce a new one.
For example, in the Euclidean case we have
the following new eigenvalues:
\begin{align}H_1^\gamma:\quad&-\frac{1}{a^2},\quad \text{if } a<0,\\
  H_2^\varepsilon:\quad&-\frac{1}{a^2}, \\
    H_3^\gamma:\quad&-\frac{1}{a^2},\quad \text{if } a>0,\end{align}
             where we use the scattering length $a$  (see
             Subsect. \ref{ssc:Euclidean_ptpot}
              for its relation to $\gamma$ and $\varepsilon$).

              For dimensions $d\geq4$ the point potential
              may introduce additional poles of Green's functions located at 
$z$ satisfying
\begin{align}
  \gamma(z)+ \Sigma_d^\bullet(z)&=0, \qquad \text{odd }d;\label{uyu3_doubled}\\
  \eta(z)+ \Sigma_d^{\bullet,\varepsilon}(z)&=0, \qquad \text{even }d. \label{uyu4}
\end{align} 
The interpretation of these singularities 
is less clear. We may call them {\em eigenvalues} of
$H_d^{\bullet,\gamma}$, $H_d^{\bullet,\varepsilon,\eta}$, even though
strictly speaking these Hamiltonians do not exist in the Hilbert space sense.
For $d\geq4$ these poles may appear outside of the real line (after all, they
are not eigenvalues of a~true self-adjoint operator).

For  pure reference self-energies, the additional singularities are easy
to determine. If $d\geq3$ is odd, the singularities originating from 
\eqref{uyu3_doubled} with $\gamma=0$ are as follows:
\begin{align}
  H_d^{0}:\qquad&z=0; \label{uyu5} \\ \notag
  H_d^{\h,0}:\qquad&z=-k^2,\quad k=0,1,\dots,\frac{d-3}{2};\\ \notag
  H_d^{\s,0}:\qquad& 
  z=\big(k+\tfrac12)^2,\quad k\in \nn_0. 
\end{align}
If $d\geq2$ is even, the singularities originating from 
\eqref{uyu4} with $\eta=0$ are 
\begin{align}\label{uyu5_doubled}
  H_d^{0}:\qquad& 
  z=-\e^{2\varepsilon},\quad\text{and if } d\geq4 \text{ also } z=0;\\\notag
  H_d^{\h,0}:\qquad&z \text{ solving } \psi\big(\tfrac{3-d}{2}+\sqrt{-z}\big)
  +\psi\big(\tfrac{d-1}{2}+\sqrt{-z}\big)
  =2\varepsilon,\\\notag&
  \text{and}\quad z=-(k+\tfrac12)^2,\quad k=0,1,\dots,\frac{d-4}{2};\\  
  H_d^{\s,0}:\qquad&
z\text{ solving } \psi\big(\tfrac{d-1}{2}+\i\sqrt{-z}\big)+\psi\big(\tfrac{d-1}{2}-\i\sqrt{-z}\big)=2\varepsilon\notag\\ \notag
&\text{and}\quad         z=(k+\tfrac12)^2,\quad k=0,1,\dots,\frac{d-4}{2}. 
                     \end{align}

What is especially interesting
are eigenvalues (or poles of Green's
functions) in the spherical case inside $[0,\infty[$
for a general point potential, which we discuss
in Subsect. \ref{sec:eigs}.
In the
unperturbed case for the sphere of radius $R$ they are located at
\beq
\frac{(l+\frac{d-1}{2})^2}{R^2},\quad\text{with multiplicity
}\frac{(2l+d-1)(d+l-2)!}{(d-1)!l!},\quad l=0,1,\dots.\label{qrw}\eeq
One effect of the perturbation is that the multiplicity of each of
these eigenvalues is decreased by one (in particular $\left( \frac{d-1}{2R} \right)^2$ is not an eigenvalue) and a shifted eigenvalue appears.

Let $E_{d,l,R}^{\gamma}$ be the $l$th shifted eigenvalue in the odd 
case.
Below we give formulas for $E_{d,l,R}^{\gamma}$ in the generic case
$\gamma(z) \neq 0$. 
If $\nu$ denotes the order of vanishing of 
$\gamma(z)$ at $z=0$, we find:
\begin{align}
E_{1,l,R}^{ \gamma} & =   \frac{(l + \frac12)^2}{R^2}  + \frac{4 \gamma (l+\frac12)^2}{\pi R^3} + \mathcal{O} \left( \frac{\gamma^2}{R^4} \right) , \\
E_{d,l,R}^{\gamma}&=\frac{(l+\frac{d-1}{2})^2}{R^2}\\&
-\frac{2 (l+\frac{d-1}{2})^2
  \prod_{k=0}^{\frac{d-3}{2}}\big((l+\frac{d-1}{2})^2-k^2\big)}{(4\pi)^{\frac{d}{2}}\Gamma(\frac{d}{2})R^d \gamma \left( \frac{(l + \frac{d-1}{2})^2}{R^2} \right)}
  +\mathcal{O}(R^{2-2d+4\nu}),\quad \text{odd }d\geq3
;\notag \end{align}
 Note that $d=1$ is special.

Now consider the $l$th shifted eigenvalue $E_{d,l,R}^{\varepsilon,\eta}$ in the even-dimensional case. We have 
\begin{align}
E_{2,l,R}^{\varepsilon}&=\frac{(l+\frac12)^2}{R^2}+\frac{l+\frac12}{R^2 
\ln(R\e^{\varepsilon})}
+\mathcal{O}\Big(\frac{1}{R^2\ln^2 (R\e^{\varepsilon})}\Big),
\\E_{d,l,R}^{\varepsilon,\eta}&= \frac{(l+\frac{d-1}{2})^2}{R^2}\\
 & -\frac{2
  (l+\frac{d-1}{2})\prod_{j=0}^{\frac{d-4}{2}}\big((l+\frac{d-1}{2})^2-(j+\frac12)^2\big)}{(4\pi)^{\frac{d}{2}}\Gamma(\frac{d}{2})R^d \eta \left( \frac{(l + \frac{d-1}{2})^2}{R^2} \right)}
  + \mathcal{O}\Big( \ln(\ee^{\varepsilon} R) R^{-2d+2+4\nu}\Big),\quad
   \text{even }d\geq4.
   \notag
 \end{align}

Note that for $d\geq3$ we have a
systematic shift of the $l$th eigenvalue 
asymptotically proportional
to $\gamma(0)^{-1}$ resp. $\eta(0)^{-1}$ 
and inversely proportional to the volume of
$\SS^d$. The scaling of the shift with $R$ is changed if $\gamma(0)
=0$ resp. $\eta(0)=0$.

  In particular, for $d=3$ the $l=0$ eigenvalue moves up by
  \beq
  \approx-\frac{1}{|\SS^3|\gamma}=\frac{4\pi a}{|\SS^3|},\label{plpl}\eeq
  where $|\SS^3|$ is the volume of
  $\SS^3$. It is interesting to ask whether a similar formula is true for
  other compact manifolds.

\subsection{Comparison with the literature}

Explicit formulas for the Euclidean, hyperbolic and spherical Green's
functions in any dimension are known in the literature,  see e.g. 
\cite{CDT,Durand23}. In our presentation we made an effort to describe 
various facets of these Green's functions in a 
(hopefully) complete and transparent way. 
In particular, we use the Gegenbauer 
functions with the conventions of Appendix \ref{app:gegenbauer}, 
because they yield much simpler expressions than
the so-called associated Legendre functions, which are commonly found
in the literature
 \cite{CDT,Olver} in this context.

Point potentials in dimension $d=3$ go back to Fermi \cite{fermi}, and
since then have been  often used in the physics literature. Berezin and
Faddeev \cite{BF} seem to have been the first who interpreted them in a rigorous way.
They are the subject of an extensive mathematical literature, confer for example  
\cite{AGHH,AK}. Point potentials in dimension $d=1,2,3$ are special
cases of {\em singular perturbations}, that is, perturbations which
cannot be interpreted as operators. As we mentioned above, the case
$d=1$ can be interpreted as a form perturbation, so that one can use
the so-called KLMN Theorem \cite{RSII}. For $d=2,3$ the form technique
is not applicable, therefore in these dimensions point potentials
belong to the class of {\em form singular perturbations}.

The formula for the resolvent  of the form \eqref{eq:Gdgamma+} is
often called the {\em Krein formula} \cite{krein}. One can also find
the name {\em Aronszajn-Donoghue theory} for this kind of treatment of singular
rank one perturbations, see e.g. 
\cite{De17}.
\eqref{eq:Gdgamma+} is essentially the singular version of
the formula for the resolvent of an operator with a rank one  
perturbation, sometimes called
 the {\em 
  Sherman-Morrison formula}.

There exist a large literature about point potentials on $\rr^d$ for
$d\geq4$. These potentials are examples of {\em supersingular
  perturbations}. In order to interpret them as true linear operators
one needs to extend the Hilbert space by adding additional dimensions,
see e.g. \cite{KuPa}. This is reviewed in \cite{Ku}.
Note that our
approach is different: we do not look for the perturbed operator, we
try to compute Green's function associated with
point potentials in all dimensions.
It seems that the formulas
\eqref{d5}, \eqref{d4}  are new, at least in the hyperbolic and spherical case.

It is clear that point potentials can be
defined on a manifold of dimension 1,2,3, and have then similar
properties as on the Euclidean space. Nevertheless, we have never seen
their analysis on hyperbolic and spherical spaces including an
explicit formula for their resolvent. So we think that 
also the  identities \eqref{d2} and \eqref{d3} 
in the hyperbolic and spherical case are new.

The concept 
of a  generalized integral goes back to independent considerations of  
Hadamard \cite{Hadamard23, Hadamard32} and Riesz \cite{Riesz}.  The generalized integral is a linear extension of 
the integration functional to not necessarily integrable functions. 
It is closely related to the extension of homogeneous distributions \cite{Hoermander90}.
More recent accounts are given in \cite{Lesch97,Paycha}. 
In a parallel work \cite{DGR23a}, we revisited this concept in a~manner that is well-suited for  our applications.

The flat limit of hyperbolic and spherical Green's functions (without
point potentials) is discussed in \cite{CDT}. 
Note that the latter reference uses the \emph{associated Legendre 
equation} instead of the Gegenbauer equation. The two equations 
are equivalent. The relation between Gegenbauer functions in our 
convention and associated Legendre functions can be found in 
\cite{DGR23a}.

 Our 
results about the energy shift of eigenvalues of  the spherical Laplacian 
in dimensions $d\geq3$  seem to be new.  They are consistent with
the following known fact,   implicit   e.g.  in \cite{LSSY}:
in dimension $3$ in a large box the ground state energy of the 
Schr\"odinger Hamiltonian with a short range potential characterized 
by scattering length $a$ has the ground state energy 
$\approx\frac{4\pi a}{\mathrm{Vol}}$, where $\mathrm{Vol}$ is 
the volume of the box (compare with \eqref{plpl}).
This fact plays an important role in the well-known asymptotics of the
bound state energy of the $3$-dimensional $N$-body Bose gas.

The ground state energy of a dilute Bose gas as in dimensions $d\geq4$
was studied in \cite{Aaen}, where a similar asymptotics as for $d=3$ was obtained and analogs of the concept of the scattering length for higher  
dimensions proved useful. Our interest in point potentials in higher 
dimension was partially sparked by this paper.

Our paper extensively uses  results of the companion paper
\cite{DGR23a} by the same authors. In particular, we use the
conventions for special functions described in
\cite{DGR23a}. These conventions are also explained in
\cite{De1,De2}.

An important source of inspiration  for our paper is renormalization
in   Quantum Field
Theory. In fact, our paper applies the ideas of two distinct methods:
 the point splitting followed by differentiation in the energy, as
 well as dimensional regularization followed by  the minimal
subtraction. This is described e.g. in \cite{collins}.

\subsection{Strategy of the paper}
The study of the operators
$H_d$, $H_d^\h$ and $H_d^\s$ and their
point-like perturbations is carried out in Sections 
\ref{sec:Euclidean}, \ref{sec:hyperbolic} and \ref{sec:spherical}, 
respectively. In all three cases ($\rr^d, \hh^d$ and $\SS^d$), 
the general strategy to determine the family of renormalized Green's 
function is the same:
\begin{enumerate}
 \item We first describe (the well-known) 
Green's functions and spectral projections of the \emph{unperturbed} 
operators.
\item The subsequent analysis of point potentials greatly varies with the 
dimension. The computations in dimensions $d=1,2,3$ are straightforward and 
only employ the standard integral. These three dimensions are spelled out 
separately. Odd and even dimensions 
$d\geq 4$ need a more careful treatment.
We present two methods: the point-splitting method and the
minimal subtraction method. In particular, following the latter
approach we compute the self-energy derived from the generalized integral.
In the case of odd dimensions  
$d\geq5$, this integral is \emph{non-anomalous} and yields easily the
reference self-energy. In the case of even 
dimensions $d\geq4$, the generalized integral
is \emph{anomalous} and was computed in \cite{DGR23a}
via the method of dimensional regularization. We describe how to pass
from the minimal subtraction self-energy to the family of reference
self-energies that seem to be preferable.
\item In the hyperbolic and spherical cases, we describe the 
flat limit of the free and perturbed Green's functions. 
\item In the spherical case we discuss the poles of the perturbed Green's functions
  $G_d^{\s,\gamma}$  resp.   $G_d^{\s,\varepsilon,\eta}$.
\end{enumerate}

This paper has four appendices.
In Appendix \ref{app:genInts}, we briefly introduce the 
generalized integral and list its relevant properties.
Appendices \ref{app:bessel} 
and \ref{app:gegenbauer} contain basic information on Bessel resp. 
Gegenbauer functions as well as a collection of their bilinear generalized integrals, which are needed in our analysis of point-like perturbations. 
Moreover, Appendix \ref{app:gegenbauer} contains a list of the 
asymptotic behaviors of Gegenbauer functions, relevant 
in our description of the flat limit.
Both appendices
are based on our
parallel work \cite{DGR23a}, were the focus lies on  properties of the generalized integral and of Bessel and
Gegenbauer functions.
Finally, Appendix \ref{poch} contains some useful formulas related to
Pochhammer symbols and harmonic numbers.

\subsection{Notation for operators on Hilbert spaces}

The \emph{integral kernel} $A(x,y)$ of an operator $A$ on $L^2(\rr^d)$
is the function, or sometimes distribution, 
on $\rr^d\times\rr^d$, such that 
\begin{align}
(f|Ag)=\int\d x\int \d y \overline{f(x)}A(x,y)g(y).
\end{align}

More generally, one can also define the integral kernel of an operator
on $L^2(M,\d\mu)$, where $M$ is a manifold  with a measure $\d\mu$, and the scalar product is given by $\int \overline{f(x)}g(x)\d\mu(x)$.
Then we   define the integral kernel of an operator
$A$, also denoted $A(x,y)$,  by
\[(f|Ag)=\int\d \mu(x)\int \d\mu( y) \overline{f(x)}A(x,y)g(y).\]

Let $H$ be a self-adjoint operator. $\sigma(H)$ will denote its {\em spectrum}.
For $z\in\cc\backslash\sigma(H)$ we can define its {\em resolvent} $(-z+H)^{-1}$.
We will denote the {\em spectral projection} of $H$ corresponding to $a\in\sigma(H)$ by $\one_a(H)$. We will denote by $\one_{[a,b]}(H)$ the spectral projection of
$H$ corresponding to the closed interval $[a,b]$.

The spectral projections can be computed with help of the resolvent:
\begin{align}
  &\one_a(H)=-\slim_{\epsilon\searrow0}\i\epsilon(-a-\i\epsilon+H)^{-1}
,\label{stone0}\\
  &\one_{[a,b]}(H)-\frac12\big(\one_{a}(H)+\one_b(H)\big)
  \label{stone} \\
  \notag
  =& \slim_{\epsilon\searrow0} \int_a^b\frac{\d s}{2\pi\i} \big((-s-\i\epsilon+H)^{-1}-(-s+\i\epsilon+H)^{-1}\big).
\end{align}
$\slim$ denotes the limit in the strong operator topology. \eqref{stone} is called the {\em Stone formula}.

If $z\in\sigma(H)$, then $(-z+H)^{-1}$ is not well-defined. However,
in an appropriate topology  $(-z+H)^{-1}$ may have well-defined limits on the spectrum. Usually, the limit is different when we approach the spectrum from above and from below.
These limits are denoted by $(-z\mp \i0+H)^{-1}$.

Let $G(z):=(-z+H)^{-1}$ be the resolvent of a self-adjoint
operator $H$. Then 
\begin{align}\label{resol-}
  (H-z)G(z)=\one,\qquad
  G(z)^*=G(\bar z),\qquad
  \frac{\d}{\d z}G(z)=G(z)^2.
\end{align}

\section{Green's operators on Euclidean space}\label{sec:Euclidean}

\subsection{The Euclidean Laplacian}

As explained in introduction, in this subsection
we consider $H_d:=-\Delta_d$, where $\Delta_d$ is
the Laplacian on $L^2(\rr^d)$.
 We first recall the well-known   formulas for the integral kernel
of its resolvent
\beq G_d(z):=
(-z+H_d)^{-1}=(\beta^2+H_d)^{-1},\eeq
where we indicate two notations for the spectral parameter that we 
will use, $z=-\beta^2$. 
We also describe the integral kernel of $\pp_d(a,b)$, the spectral projection of
 $H_d$ onto $[a,b]$.
We express them in terms of various functions from the 
     Bessel family: $K_\alpha$, $J_\alpha$ and $H_\alpha^\pm$.   
     Relevant properties  of these functions are listed in Appendix \ref{app:bessel}. For
     completeness, we sketch a  proof of this theorem.

     \bet
     \begin{enumerate}\item
     For $\Re\beta>0$ we have
     \beq
     G_d(-\beta^2;x,x')
     =\frac{1}{(2\pi)^{\frac{d}{2}}}\Big(\frac\beta{|x-x'|}\Big)^{\frac{d}{2}-1}K_{\frac{d}{2}-1}
         \big(\beta|x-x'|\big). \label{eq:flat_resolvent}
         \eeq
\item      Green's function
         possesses limits on $]0,\infty[$ from above and below.
         For $\zeta\in\rr$, $\zeta>0$,  these limits are
     \beq
     G_d(\zeta^2\pm\i0;x,x')
     =\pm\frac{\i}{4}\Big(\frac{\zeta}{2\pi|x-x'|}\Big)^{\frac{d}{2}-1}
       H_{\frac{d}{2}-1}^\pm\big(\zeta|x-x'|\big).
       \label{besse1}  \eeq
     \item  For $d>2$, there exists also a limit at $z=0$:
       \beq
       G_d(0;x,x')=\frac{\Gamma(\frac{d}{2}-1)}{4\pi^\frac{d}{2}|x-x'|^{d-2}.}
       \eeq
\item       
         Finally,
         the integral  kernels of the spectral projections are
         \begin{align}
           \pp_d(a,b;x,x')
&=\int_{\sqrt{a}}^{\sqrt{b}}\Big(\frac{\zeta}{2\pi}\Big)^{\frac{d}{2}}       \frac{J_{\frac{d}{2}-1}\big(\zeta|x-x'|\big)}{|x-x'|^{\frac{d}{2}-1}}\d\zeta.          
 \label{besse2}      \end{align}
\end{enumerate}
\eet

\begin{proof}
Let $r:=|x-x'|$.
By Euclidean invariance, there exists a function $G_d(z,r)$ such that
$G_d(z;x,x')=G_d(z,r)$.  Away from $r=0$, we can write 
\begin{align}
0=  &\Big(-\partial_r^2-\frac{d-1}{r}\partial_r+\beta^2\Big)G_d(-\beta^2,r)\\
  =&r^{1-\frac{d}{2}}\Big(
     -\partial_r^2-\frac{1}{r}\partial_r+\frac{(\frac{d}{2}-1)^2}{r^2}+\beta^2\Big)r^{-1+\frac{d}{2}}G_d(-\beta^2,r). \nonumber
     \end{align}
     Then we find the solution vanishing at infinity and behaving for $r \to 0$ as
     \beq G_d(-\beta^2,r)\sim
     \begin{cases}
     \frac{\Gamma(\frac{d}{2}-1)}{4\pi^\frac{d}{2}}r^{2-d}, & d \neq 2, \\
      - \frac{1}{2 \pi} \ln r, & d = 2,
     \end{cases}  
     \eeq
     which implies the distributional differential equation 
     \begin{equation}
         (-\Delta_d - z) G_d(z;x,x')=\delta(x,x').
         \label{eq:Green_distr_eq}
     \end{equation}  
     Since $G_d(-z;x,x')$ is an integrable function of $x-x'$, i.e. 
     \begin{align}
      \int\Big|  G_d(-z;x,x') \Big| \dd x 
      =  \int \Big| G_d(-z;x,0) \Big| \dd x =:C,
     \end{align}
     Young's inequality for convolutions \cite[Theorem 4.2]{LiebLoss} implies 
     \begin{align}
      \Bigg|\!\Bigg| \int  G_d(-z;\cdot,x) f(x) \dd x\Bigg|\!\Bigg|_2 
      \leq C |\!| f |\!|_2, \quad f\in L^2(\rr^d).
     \end{align}
     Thus, $G_d(-z;x,x')$  is the integral kernel of a bounded operator, 
     and hence by \eqref{eq:Green_distr_eq} the integral kernel of the resolvent of $H_d$. This shows \eqref{eq:flat_resolvent}.

To derive the limits \eqref{besse1}, write $\beta=\beta_\RR+\i\beta_\II$ with
$\beta_\RR>0$. Then
\beq z=-(\beta_\RR+\i\beta_\II)^2=\big(|\beta_\II|-\i\beta_\RR\sgn(\beta_\II)\big)^2
\underset{\beta_\RR\searrow0}{\to}
|\beta_\II|^2-\i 0 \, \sgn(\beta_\II).\eeq
Hence, to get $z=\zeta^2\pm\i0$ we need to insert
$\beta=\mp\i (\zeta\pm\i0)$  with $\zeta>0$.

                 (\ref{besse2}) follows from   (\ref{besse1})
          by (\ref{stone}):
\beq\text{   lhs of      (\ref{besse2})}=\frac{1}{2\pi\i}\int_{\sqrt{a}}^{\sqrt{b}}2\zeta 
           \Big(G_d(\zeta^2+\i0;x,x')-G_d(\zeta^2-\i0;x,x')\Big) \d\zeta.\eeq 
\qed
\end{proof}

We remark that the function 
     \beq
   G_d(-\beta^2,r):
     =\frac{1}{(2\pi)^{\frac{d}{2}}}\Big(\frac\beta r\Big)^{\frac{d}{2}-1}K_{\frac{d}{2}-1}
         \big(\beta r\big). \label{eq:flat_resolvent.}
         \eeq
 is well-defined for
all  $d\in\cc$, and not only for positive integers. Note its symmetry:
\beq
G_{4-d}(-\beta^2,r)=\Big(\frac{\beta}{2\pi r} \Big)^{2-d}G_d(-\beta^2,r),
\eeq
coming from the symmetry of the Macdonald function $K_{\alpha}(z) = 
K_{-\alpha}(z)$ \cite{NIST}. We have also the homogenity relation
\begin{equation}
\label{eq:resolv_homog}
    G_d(-(\lambda \beta)^2, \lambda^{-1} r) = \lambda^{d-2} G_d(- \beta^2 , r),
\end{equation}
which is equivalent to the fact that $H_d$ is an operator homogeneous of degree $-2$.

Let us sum up the properties of Green's function for various $d$ (not
only positive integers). The behavior near zero is described by the 
power series (cf. \eqref{eq:BesselI_series}, \eqref{macdo1}): 
  \begin{subnumcases}{G_d(-\beta^2;r) =}
    \frac{1}{4\pi^{\frac{d}{2}}r^{d-2}}\sum_{k=0}^\infty\frac{(-1)^k\Gamma(\frac{d-2}{2}-k)}{k!}\big(\tfrac{\beta 
        r}{2}\big)^{2k}
\label{explicit1} \\ 
+\frac{\beta^{d-2}}{(4\pi)^{\frac{d}{2}}}\sum_{j=0}^\infty\frac{(-1)^j\Gamma(\frac{2-d}{2}-j)}{j!}\big(\tfrac{\beta 
                            r}{2}\big)^{2j}, 
                             &\quad $d\not\in 2\zz$; \notag

                             \vspace{18pt}
                            \\ \label{explicit2}
\frac{1}{4\pi^{\frac{d}{2}}r^{d-2}}
\sum_{k=0}^{\frac{d-4}{2}}\frac{(-1)^k(\frac{d-4}{2}-k)!}{k!}\big(\tfrac{\beta r}{2}\big)^{2k}
\\ 
  +\frac{\beta^{d-2}(-1)^{\frac{d}{2}}}{(4\pi)^{\frac{d}{2}}}
  \sum_{j=0}^\infty\frac{2\ln(\frac{\beta
    r}{2})+2\gamma_\mathrm{E}-H_j-H_{\frac{d-2}{2}+j}}{j!(\frac{d-2}2+j)!}
\big(\tfrac{\beta
                            r}{2}\big)^{2j}, & $d=2,4,6,\dots$.
                            \notag
   \end{subnumcases}

Note that the singular part of the first line of
\eqref{explicit1} has the same form as the first line of 
\eqref{explicit2}, and that for  $\Re d< 2+2n$ there exists the limit
\begin{equation}
\label{eq:limit_derivative_G}
    \lim_{x \to 0} \frac{\partial^n}{\partial z^n} G_d(z;x,0) = \frac{\Gamma(\frac{2-d}{2} +n)}{(4 \pi)^{\frac{d}{2}}}\beta^{d - 2 -2 n} . 
\end{equation}
The latter equation can be derived from the power series 
\eqref{explicit1} and \eqref{explicit2}. Note that there is no logarithmic  
singularity as $r\to0$ for $d=2,4,6,\dots$ because $2n>d-2$ and hence, 
at least one derivative must act on $\ln\big(\tfrac{\beta r}{2}\big)$.

There is an elementary formula for odd dimensions (cf. \eqref{eq:MacDonald_halfinteger}):
\begin{align}\label{explicit3}
  G_d(-\beta^2;x,x')&=\frac{1}{2(2\pi)^{\frac{d-1}{2}}}\Big(-\frac{1}{r}\partial_r\Big)^{\frac{d-3}{2}}\frac{\e^{-\beta
                   r}}{r},\quad d=3,5,\dots.
  \end{align}
For $\beta>0$ the behavior for large $r$ can be described by the
following asymptotic (divergent) series (cf. \eqref{asy}):
\begin{align}\label{explicit4}
  G_d(-\beta^2;x,x')&\simeq\frac{1}{2\beta}\Big(\frac{\beta}{2\pi r}\Big)^{\frac{d-1}{2}}
                      \e^{-\beta r}\sum_{j=0}^\infty
                      \frac{(\frac{d-1}{2}-j)_{2j}}
                      {j!(2\beta r)^j}.  \end{align}

\subsection{Point potentials on Euclidean space}
\label{ssc:Euclidean_ptpot}
Suppose that $H_d^\gamma$ is a self-adjoint extension of the 
restriction of $H_d$ to $C_\mathrm{c}^\infty(\rr^d\backslash\{0\})$. 
Consider its resolvent
\begin{align}
G_d^\gamma(z)=(H_d^\gamma+\beta^2)^{-1}=(H_d^\gamma-z)^{-1}.
\end{align}
By \eqref{resol-}, the integral kernel of $G_d^\gamma(z)$, denoted
$G_d^\gamma(z;x,x')$, satisfies
\begin{align}(-\Delta_x + \beta^2)  G_d^\gamma(z;x,x')&=\delta(x-x'), 
\label{eq:green_eqn} \qquad 
x,x'\neq0,         \\
   G_d^\gamma(z;x,x')&=   G_d^\gamma(z;x',x), \label{eq:green_sym} \\
\partial_z   G_d^\gamma(z;x,x')&= \int  G_d^\gamma(z;x,y)
 G_d^\gamma(z;y,x')\d y. \label{eq:green_res_id}
\end{align}
To solve these equations we make an ansatz 
\begin{align}
 \label{eq:ansatz_green_eqn}
  G_d^\gamma(z;x,x')&=   G_d(z;x,x')+ \frac{1}{\gamma(z)+\Sigma_d(z)} 
  G_d(z;x,0)G_d(z;0,x'),
\end{align}
which already incorporates the conditions \eqref{eq:green_eqn} and 
\eqref{eq:green_sym}. The denominator of the second term is split as 
$\gamma(z)+ \Sigma_d(z)$ because this expression will depend on some 
number of free parameters. We will fix $\Sigma_d(z)$ for every dimension 
and collect all free parameters in $\gamma(z)$.

\begin{remark}
Note that \eqref{eq:ansatz_green_eqn} describes a~spherically 
symmetric perturbation. In particular it excludes the $\delta'$ 
potential for $d=1$.
\end{remark}

Let us insert \eqref{eq:ansatz_green_eqn} into 
\eqref{eq:green_res_id} to determine $\Sigma_d(z)$. 
$G_d^\gamma(z;x,x')$ satisfies 
\eqref{eq:green_res_id} if 
\begin{equation}
  \frac{\d }{\d z} (\gamma(z)+\Sigma_d(z)) = -\sigma_d(z),
  \label{eq:Sigma_ODE}
\end{equation}
where
\beq
\sigma_d(z): =\int_{\rr^d} G_d(z;0,y) G_d(z;y,0)\d 
y =\frac{(\beta^2)^{\frac{d}{2}-1}2\pi^{\frac{d}{2}}}{(2\pi)^d
\Gamma(\frac{d}{2})}\int_0^\infty
     K_{\frac{d}{2}-1}(\beta r)^2 r\d r.
    \label{eq:sigma_def}
    \eeq

Note that
the rightmost integral in \eqref{eq:sigma_def} makes sense for complex 
$d$. It converges only for $|\Re(d-2)|<2$, which includes the dimensions
$d=1,2,3$ \cite{GR}: 
\begin{equation}
    \label{sigma0}
 \sigma_d(-\beta^2) = 
 \frac{\Gamma \left( \frac{4-d}{2} \right)}{(4\pi)^{\frac{d}{2}}}
 \beta^{d-4},\quad |\Re(d-2)|<2.
\end{equation} 
For these dimensions we take $\Sigma_d(z)$ to be a fixed 
anti-derivative of $\sigma_d(z)$. Then \eqref{eq:Sigma_ODE} 
says that $\gamma$ is a constant. $G_d^\gamma$ is the integral 
kernel of the resolvent of a closed operator $H_d^\gamma$, which 
is self-adjoint if $\gamma$ is real. One has $H_d=H_d^\infty$.

Below we propose how to define $\Sigma_d(z)$ in all dimensions. 
It will be seen that in contrast to $d=1,2,3$, it is natural to take 
$\gamma(z)$ to be a polynomial in $z$ of degree depending on $d$. 
Hence $G_d^\gamma$ depends on several parameters. For every choice of 
$\gamma(z)$, $G_d^\gamma(z;x,x')$ is a well-defined locally integrable 
function, but for $d\geq4$ it is not the integral kernel of a~bounded operator. It 
describes the asymptotic behavior of Green's function of a Laplacian 
with a perturbation of a very small support, as explained in a 
separate paper.  Let us discuss various dimensions separately. 

\paragraph{Dimension  $1$.}
We have $\sigma_1(-\beta^2)=\frac{1}{4\beta^3}$, so we define 
$\Sigma_1(-\beta^2)=-\frac{1}{2\beta}$, homogeneous and vanishing 
at infinity. We have 
\beq\label{d=1}
G_1^\gamma(-\beta^2;x,x')=\frac{\e^{-\beta|x-x'|}}{2\beta}
+\frac{\e^{-\beta|x|}\e^{-\beta|x'|}}{(2\beta)^2\big(\gamma-\frac{1}{2\beta}
\big)}
.\eeq 
Sometimes $a:=-2\gamma$ is called the {\em scattering length}. The operator
$H_1^\gamma$ is the perturbation of $H_1$ by the
quadratic form $\frac2{a}\delta(x)$. If $\gamma=0$, it is $-\Delta_d$ with 
Dirichlet boundary condition at $0$, and it is homogeneous of degree $-2$. 
Functions in the domain of $H_1^\gamma$ with $\gamma \neq 0$ have the leading 
singularity near zero proportional to $\frac{|x|}{a}-1$. For $a<0$ there exists 
a bound state $\e^{\frac{|x|}{a}}$ with eigenvalue $-\frac{1}{a^2}$.

\paragraph{Dimension  $2$.}
 The full self-energy is now
\begin{align} \label{eq:gamSig_d2}
\gamma+\Sigma_2(-\beta^2) = 
\gamma+\frac{\ln\beta}{2\pi}.\end{align}
 $\Sigma_2(-\beta^2)$ diverges both if $\beta\to 0$ and 
$\beta\to\infty$.  $\gamma$ is an arbitrary constant of integration.
 It is convenient to replace \eqref{eq:gamSig_d2} by a family of 
 reference self-energies introducing $\varepsilon:=-2\pi\gamma$ and
\beq
\Sigma_{2}^{\varepsilon}(-\beta^2):=\frac{1}{2\pi}(\ln\beta-\varepsilon).\eeq
Then $a:=\exp(2 \pi \gamma)=\e^{-\varepsilon}$ specifies a length scale, 
in the physics literature again  called the {\em  scattering length}.
We find
\beq
 G_2^\varepsilon(-\beta^2;x,x')
=\frac{K_0(\beta |x-x'|)}{2\pi} 
+\frac{K_0(\beta |x|)K_0(\beta |x'|)}{2\pi\ln(\beta \e^{-\varepsilon})}
.\label{resol1}
\eeq
 In contrast to $d=1$, the scattering 
length cannot be negative. The denominator of the second term of 
\eqref{resol1} can be rewritten as  $2 \pi \ln(\beta a)$. Functions 
in the domain of $H_2^\varepsilon$ behave near zero as
$\ln(\frac{|x|}{2a})+\gamma_\mathrm{E}$, where $\gamma_\mathrm{E}$ is 
the Euler-Mascheroni constant.
For all $a$ there is a~bound state $K_0 \big( \frac{|x|}{a} \big)$ with 
eigenvalue $-\frac{1}{a^{2}}$.

\paragraph{Dimension  $3$.} In this case we take 
$\Sigma_3(-\beta^2)=\frac{\beta}{4\pi}$, homogeneous and vanishing at $0$. 
We have
\beq\label{resol3}
G_3^\gamma(-\beta^2;x,x')=\frac{\e^{-\beta |x-x'|}}{4\pi|x-x'|}
+\frac{\e^{-\beta |x|}\e^{-\beta |x'|}}{(4\pi)^2 |x||x'| 
(\gamma+\frac{ \beta }{4\pi})}.\eeq

As in lower dimensions, one usually introduces 
the {\em scattering length}, now given by $a=-\frac{1}{4\pi\gamma}$.
The denominator of the second term of
\eqref{resol3} can be rewritten as $4\pi(\beta-\frac{1}{a})|x||x'|$. 
Functions in the domain of $H_3^\gamma$ behave as $1-\frac{a}{|x|}$ near $0$. Operator $H_3^0$ is homogeneous of degree $-2$. For $a>0$ there is a bound state $\frac{1}{|x|}\exp(-\frac{|x|}{a})$ with
eigenvalue $-\frac{1}{a^{2}}$.

\paragraph{Higher dimensions.}
  We will describe two methods of introducing self-energy in higher
dimensions. The first, in the physical terminology, is based on the
differentiation with respect to the energy and point splitting. First we rewrite
the definition of $\sigma_d$ from \eqref{eq:sigma_def} as
\begin{align}
  \sigma_d(z)&=\lim_{x\to0}\int G_d(z;x,y) G_d(z;y,0)\d 
y =\lim_{x\to0}\frac{\partial}{\partial z} G_d(z;x,0).\end{align}
             The limit on the right in general does not
             exist. However, if  $d<2+2n$ and we differentiate both
             sides $n$ times in $z$, then the limit becomes finite:
\begin{align}
    \sigma_d^{(n-1)}(z) 
& =  \lim_{x \to 0} \frac{\partial^n}{\partial z^n} G_d(z;x,0) = \frac{\Gamma(\frac{2-d}{2} +n)}{(4 \pi)^{\frac{d}{2}}}\beta^{d - 2 -2 n}, \label{eq:Sigma_nth_der_sep} 
\end{align}
where we inserted \eqref{eq:limit_derivative_G} in the last 
step.

We take $n$ to be the smallest integer greater than $\frac{d-2}{2}$ and choose some $\Sigma_d(z)$ satisfying
\begin{equation}
    \Sigma_d^{(n)}(z) = - \frac{\Gamma(\frac{2-d}{2} +n)}{(4 \pi)^{\frac{d}{2}}}\beta^{d - 2 -2 n}.
    \label{eq:Sigma_nth_der_expl}
\end{equation}
Then differentiating \eqref{eq:Sigma_ODE} $n-1$ times we find that $\gamma(z)$ is a polynomial of degree $n-1$.

The second method of defining the self-energy
yields  a concrete $\Sigma_d$ satisfying 
\eqref{eq:Sigma_nth_der_expl} for every dimension. To this end 
we define $\sigma_d(z)$  by replacing the Lebesgue integral in \eqref{eq:sigma_def} with the generalized integral $\gen \int_0^\infty$, 
which is defined in \eqref{gener}.
Then we can choose $\Sigma_d$ to be an  antiderivative 
of $-\sigma_d(z)$.

We will check that thus defined
$\Sigma_d$ satisfies \eqref{eq:Sigma_nth_der_expl} by explicit
computation, separately for odd and even dimensions.
 One can also see this by the following general argument.

 Consider
 \begin{equation}
\sigma_d(z) :
    = \frac{2\pi^{\frac{d}{2}}}{(2\pi)^d\Gamma(\frac{d}{2})} 
    \ge \int_0^\infty 
     \beta^{d-2}K_{\frac{d}{2}-1}(\beta r)^2 r\d r. 
\end{equation}
 Since the exponents of terms non-integrable near $0$ do not depend 
on $z$, one can check that
\begin{equation}
    \frac{\partial^{n-1}}{\partial z^{n-1}} \sigma_d(z) 
    = \frac{2\pi^{\frac{d}{2}}}{(2\pi)^d\Gamma(\frac{d}{2})} 
    \ge \int_0^\infty \frac{\partial^{n-1}}{\partial z^{n-1}} \Big(
     \beta^{d-2}K_{\frac{d}{2}-1}(\beta r)^2 \Big) r\d r.  \label{erer}
\end{equation} 
If $d < 2 +2n$, the generalized integral on the 
right-hand side converges in the classical sense, and we can write
\begin{align}
   \frac{\partial^{n-1}}{\partial z^{n-1}} \sigma_d(z) 
&    = \frac{2\pi^{\frac{d}{2}}}{(2\pi)^d\Gamma(\frac{d}{2})} 
\int_0^\infty \frac{\partial^{n-1}}{\partial z^{n-1}} \Big(
\beta^{d-2}K_{\frac{d}{2}-1}(\beta r)^2 \Big) r\d r\\
&    =  
\int \frac{\partial^{n-1}}{\partial z^{n-1}} \lim_{x\to0} \Big(
 G_d(z;x,y) G_d(z;y,0)\Big)\d y \nonumber \\
 &    =  
\int  \lim_{x\to0} \frac{\partial^{n-1}}{\partial z^{n-1}} \Big(
 G_d(z;x,y) G_d(z;y,0)\Big)\d y . \nonumber
\end{align}
Next, we rename $y=y_{n}$ and express the derivative in the last line as an $(n-1)$-fold integral using the resolvent identity repeatedly:
 \begin{align}
 \frac{\partial^{n-1}}{\partial z^{n-1}} \sigma_d(z)  &=\int \lim_{x \to 0} \Big( \int G_d(z;x,y_1) \cdots G_d(z;y_n,0) \d y_1 \cdots \d y_{n-1}  \Big) \d y_n.
 \end{align}
Since $G_d(z;x,y_1) \cdots G_d(z;y_n,0)$ is an integrable function of $y_1, \dots, y_n$, we may take the limit out of the integral.\footnote{If $f \in L^1(\rr^N)$ and $f_a(x):= f(x-a)$ for $a \in \rr^N$, then $f_a \to f$ strongly in $L^1(\rr^N)$ as $a \to 0$.} Therefore,
 \begin{align}
 \frac{\partial^{n-1}}{\partial z^{n-1}} \sigma_d(z)  &=\lim_{x \to 0} \int  G_d(z;x,y_1) \cdots G_d(z;y_n,0) \d y_1 \cdots  \d y_n.
 \end{align}
The integral can now be computed using the resolvent identity again. This shows that \eqref{eq:Sigma_nth_der_sep}, and hence \eqref{eq:Sigma_nth_der_expl}, is satisfied.

One could use other definitions of generalized integration, for 
example in other coordinates (cf.~\eqref{eq:change_of_var}). An 
inspection of formulas \eqref{explicit1}, \eqref{explicit2}, and 
\eqref{eq:gen_int_prime} shows that the resulting $\sigma_d(z)$ 
differs only by a polynomial of degree $n-2$, where $n$ is the 
smallest integer greater than $\frac{d-2}{2}$. Integrating to 
find $\Sigma_d(z)$ leads to another integration constant, and 
this accounts for the same freedom in the choice of $\Sigma_d(z)$ 
as suggested by \eqref{eq:Sigma_nth_der_expl}.

We note that the leading term of $\Sigma_d(z)$ for large $z$ is 
uniquely determined by either \eqref{eq:Sigma_nth_der_expl} or 
by calculation of generalized integrals. The term $\gamma(z)$ 
containing free parameters is of lower order for large $z$.

\paragraph{Odd dimensions $d \geq 5$.} If \eqref{eq:sigma_def} 
is understood as a generalized integral and $d$ is not an even 
integer, then expression \eqref{sigma0} remains valid. It is 
convenient to rewrite \eqref{sigma0} as 
\begin{align}
 \sigma_d(-\beta^2) =  \beta^{d-4}
 \frac{\pi}{(4\pi)^{\frac{d}{2}} \Gamma \left( \frac{d-2}{2} \right)
 \cos\big(\pi\tfrac{d-3}{2}\big)}
,\quad d\in\cc\setminus 2\zz.
\end{align}
Therefore we take
\beq
\Sigma_d(-\beta^2)= \beta^{d-2} \frac{\pi}{(4 \pi)^{\frac{d}{2}} 
\Gamma \left( \frac{d}{2} \right)  \cos\big(\pi\tfrac{d-3}{2}\big)},
\quad d\in\cc\setminus 2\zz,
\label{eq:Sigma_generic}
\eeq
homogeneous and vanishing at $0$.  Since 
the generalized integral is non-anomalous for $d\in\cc\setminus2\zz$, 
the same result is obtained using the generalized integral with the 
integration variable $\lambda r$ for any constant $\lambda >0$. 
Specifying $d$ to be an odd integer we obtain
\beq\label{oddo}
\Sigma_d(-\beta^2)
=(-\beta^2)^{\frac{d-3}{2}}\beta
\frac{
\pi
}{(4\pi)^{\frac{d}{2}}\Gamma(\tfrac{d}{2})}.
\eeq
The polynomial $\gamma(z)$ is of degree $\frac{d-3}{2}$, so it depends 
on $2$ parameters for $d=5$, on $3$ parameters for $d=7$ etc. For large 
$z$ it is subleading with respect to $\Sigma_d(z)$ by at least one power 
of $\beta$, and in contrast to $\Sigma_d$ it 
may contain even powers of $\beta$.
The function $G_d^0$ (i.e.~with $\gamma(z)=0$) satisfies the same 
homogeneity relation \eqref{eq:resolv_homog} as $G_d$.

\paragraph{Even dimensions $d \geq 4$.} 
The expression \eqref{sigma0} for $\sigma_d(-\beta^2)$
is not valid, even in the generalized sense because of the scaling anomaly. 
For $d=4,6,\dots$, we find (cf. App. \ref{app:bessel})
\begin{equation}
    \sigma_d(-\beta^2) 
    = - \frac{ (- \beta^2)^{\frac{d-4}{2}}}{(4 \pi)^{\frac{d}{2}}\Gamma 
    \left( \frac{d}{2} \right)} 
    \Big( 1 + (d-2)\left( \ln \tfrac{\beta}{2} +1 - \psi \left( \tfrac{d}{2}
    \right) \right)  \Big) ,
\end{equation}
and therefore the  self-energy given by the minimal
  subtraction method is
\begin{align}\label{even.}
\Sigma_d^\mathrm{ms}(-\beta^2) &=\frac{ (- \beta^2)^{\frac{d-2}{2}}}{(4
  \pi)^{\frac{d}{2}} \Gamma \left( \frac{d}{2} \right)} \left( 2 -
2  \psi \left( \tfrac{d}{2} \right) + \ln \tfrac{\beta^2}{4}  \right)  ,
\end{align}
to which  we can add a polynomial $\gamma$ of degree $\leq\frac{d-2}{2}$.
If a rescaled radial coordinate is used in the generalized integral, the 
result is shifted by a multiple of $(- \beta^2)^{\frac{d-2}{2}}$. Note that 
this power of $- \beta^2$ is also the leading term of $\gamma(- \beta^2)$,
but it is still subleading in the self-energy due to the 
presence of the logarithm in \eqref{even.}. Therefore, there are $2$ parameters for $d=4$, $3$ parameters for $d=6$ 
etc.

It is convenient to introduce a  scale-dependent
reference self-energy
\begin{align}
\Sigma_d^\varepsilon(-\beta^2):=
                                \frac{1}{(4 \pi)^{\frac{d}{2}}\Gamma(\tfrac{d}{2})}
                                (- \beta^2)^{\frac{d-2}{2}} (\ln\beta^2-2\varepsilon) ,
\end{align}
where $\varepsilon$ is used to absorb the highest term 
$\gamma_{\tfrac{d-2}{2}}\; (- \beta^2)^{\frac{d-2}{2}}$ in $\gamma$:
\begin{align}
-2 \varepsilon
=(4  \pi)^{\frac{d}{2}} \Gamma \left( \tfrac{d}{2} \right)
\gamma_{\frac{d-2}{2}}
+ 2 - 2 \psi \left( \tfrac{d}{2}\right)-\ln 4.
\end{align}
Thus we obtain a family of Green's functions
\begin{align}
 \label{eq:Gdgamma.}
  G_d^{\varepsilon,\eta}(z;x,x')&=   G_d(z;x,x')+ 
  \frac{1}{\eta(z)+\Sigma_d^{\varepsilon}(z)} G_d(z;x,x_0)G_d(z;x_0,x'), 
  \end{align}
where $\eta$ is an arbitrary polynomial of degree $\leq\frac{d-4}{2}$.

\begin{remark}[Scattering length in higher dimensions] Let $d \geq 3$ be odd.  If $x$ is small, 
$y$ large and $z=0$, then 
\begin{align}
  G_d^\gamma(0;x,y)\approx G_d(0;0,y)\Big(1+\frac{G_d(0;x,0)}{\gamma(0)}\Big)
  =G_d(0;0,y)\Big(1-\frac{a}{|x|^{d-2}}\Big),\label{krow}
  \end{align}
  where, following \cite[Appendix C]{LSSY}, we introduced the 
  scattering length
  \beq \label{length}
  a  :=-\frac{\Gamma(\frac{d}{2}-1)}{4\pi^{\frac{d}{2}}\gamma(0)}.\eeq
  If $d\geq4$ is even we can do the same, replacing $G_d^\gamma$ with
  $G_d^{\varepsilon,\eta}$ and $\gamma(0)$ with $\eta(0)$.

  Calling  $a$ a {\em length} is actually a misnomer, since now \eqref{length} does not have 
  the dimension of length, unlike for $d=1,2,3$. Moreover, \eqref{length} 
  is not consistent with the definition of $a$ for $d=1,2$.
\end{remark}

\section{Green's operators on hyperbolic space}
\label{sec:hyperbolic}
\subsection{Hyperbolic Laplacian}   
       The space $\rr^{1+d}$ equipped with the bilinear form
         \[[x|y]         =x^0y^{0}-x^1y^{1}-\dots-x^dy^{d}
         \]will be denoted $\rr^{1,d}$.
         The set
         \[\hh^d:=\{x\in\rr^{1,d}
         \mid[x|x]=1, \ x^0 > 0
                  \}\]
         equipped with the Riemannian metric inherited from $\rr^{1,d}$ is called
         the {\em hyperbolic space}.
         The geodesic distance between $x,x'\in\hh^d$ is given by 
         \beq
          d^{\h}(x,x')
         = \cosh^{-1}\big([x|x']\big),\qquad \cosh  d^{\h}(x,x')=[x|x'].
         \eeq
         $\hh^d$  has also a measure induced by the metric. 

In this section we study 
         \beq H_d^\h:=-\Delta_{d}^\h-\frac{(d-1)^2}{4},\eeq
         where
 $\Delta_{d}^\h$ is    the Laplacian on $L^2(\hh^d)$ induced by the metric. $H_d^\h$ is a self-adjoint operator. 
 For  $z\in \cc\backslash\sigma(H_d^\h)=\cc\backslash[0,\infty[$ we define the 
   {\em hyperbolic Green's operator} $ G_{d}^\h(z):= (-z+H_{d}^\h)^{-1}$. The spectral
   projection  of $H_d^\h$ onto $[a,b[\subset[0,\infty[$ is denoted 
   $\pp_{d}^\h(a,b)$.
   In the following theorem,  we express the integral kernels of
$G_{d}^\h(z)$ and    $\pp_{d}^\h(a,b)$ in terms of two 
kinds of Gegenbauer functions, $\bf{S}_{\alpha,\beta}$ and 
$\bf{Z}_{\alpha,\beta}$, which are defined in Appendix 
\ref{app:gegenbauer}. 
\bet
    \begin{enumerate}\item
For
  $  \Re\beta>0$ the integral kernel of $G_{d}^\h(-\beta^2)$ is
     \beq
     G_{d}^\h\Big(-\beta^2;x,x'\Big)
     =\frac{\sqrt\pi\Gamma(\frac{d-1}{2}+\beta)}{\sqrt2(2\pi)^{\frac{d}{2}}2^{\beta}
         }{\bf Z}_{\frac{d}{2}-1,\beta }                        \big([x|x']\big).\label{fad1}\eeq
\item
For $\zeta>0$, it has the following limits 
     \beq 
     G_{d}^\h(\zeta^2\pm\i0;x,x')
     =\frac{\sqrt\pi\Gamma(\frac{d-1}{2}\mp\i\zeta)}{\sqrt2(2\pi)^{\frac{d}{2}}2^{\mp\i\zeta }
         }{\bf Z}_{\frac{d}{2}-1,\mp\i\zeta }                        \big([x|x']\big).\label{fad2}\eeq
 \item 
 The integral  kernel of $\pp_{d}^\h(a,b)$ is
  \begin{align}
  \label{leg1}
  \pp_{d}^\h(a,b;x,x') 
  = \int_{\sqrt{a}}^{\sqrt{b}}
  \frac{2 \zeta \sinh(\pi\zeta)\,
            \Gamma(\tfrac{d-1}{2}+\i\zeta)
             \Gamma(\tfrac{d-1}{2}-\i\zeta)
           }
               {      \pi    (4\pi)^{\frac{d}{2}}}{\bf
           S}_{\frac{d}{2}-1,\i\zeta }
                        \big([x|x']\big)\d\zeta.
 \end{align}
 
\end{enumerate}
\eet

\begin{proof}
The isometry group of $\hh^d$ acts transitively on pairs $(x,x')$ with fixed $[x|x']$, so there exists a function $G_d^\h(z;w)$
such that
$G_{d}^\h(z;x,x')=G_d^\h(z;w)$, $w:=[x|x']$. We have
\begin{align}
0=& \left((1-w^2)\partial_w^2-dw\partial_w+\beta^2
-\big(\tfrac{d-1}{2}\big)^2\right)G_d^\h(-\beta^2,w)
  \end{align}
Near the diagonal the hyperbolic Green's function should have the same
asymptotics as Green's function of the Laplacian:
     \beq G_d^\h\big(-\beta^2,\cosh(r)\big)\sim 
     \begin{cases}
     \frac{\Gamma(\frac{d}{2}-1)}{4\pi^\frac{d}{2}}r^{2-d}, & d \neq 2, \\
     - \frac{1}{2 \pi} \ln r, & d = 2.
     \end{cases}  
     \eeq
     Besides,
it should vanish for $w\to\infty$. This fixes uniquely
$G_d^\h(-\beta^2,w)$  to be \eqref{fad1}. \eqref{fad2} follows immediately.
     
         To derive (\ref{leg1}) we note that by (\ref{stone})
\beq\text{lhs of (\ref{leg1})}
=           \frac{1}{2\pi\i}\int_{\sqrt{a}}^{\sqrt{b}}2\zeta
           \Big(G_{d}^\h\Big(\zeta^2+\i0;x,x'\Big)
           -G_{d}^\h\Big(\zeta^2-\i0;x,x' \Big) \Big) \d\zeta .
           \eeq 
           Then we use
           the following identity, which is a consequence of
           (\ref{formu1}):
\begin{align}
\label{formu}
& \frac{\sqrt\pi}{\sqrt2}
\Big(2^{\i\zeta}{\bf Z}_{\alpha,-\i\zeta}(z)\Gamma\Big(\frac12+\alpha-\i\zeta\Big)
- 2^{-\i\zeta}{\bf Z}_{\alpha,\i\zeta}(z)\Gamma\Big(\frac12+\alpha+\i\zeta\Big)\Big) \\
=&\i2^{-\alpha}  \sinh(\pi\zeta) 
\Gamma\Big(\frac12+\alpha+\i\zeta\Big)\Gamma\Big(\frac12+\alpha-\i\zeta\Big)
{\bf S}_{\alpha,\i\zeta}(z).
\notag
\end{align}
\qed
\end{proof}

We can view $G_d^\h(-\beta^2,w)$ as defined for all $d\in\cc$. It
satisfies the symmetry
\beq\label{hypersym}
G_{4-d}^\h(-\beta^2,\cosh r)=\frac{1}
{(\frac{3-d}{2}+\beta)_{d-2}(2\pi\sinh r)^{2-d}}G_d^\h(-\beta^2,\cosh r).\eeq
Here are explicit formulas for the hyperbolic Green's function
useful for small $r$. They follow from \eqref{fad1},
the connection formula \eqref{formu1}, \eqref{solu1}, \eqref{solu1_form2} and
$\frac{\cosh(r)-1}{2}=\sinh^2\frac{r}{2}$: 
\begin{subnumcases}{ G_d^\h(-\beta^2;r)=}
 \frac{1}{(4\pi)^{\frac{d}{2}}\sinh(\frac{ r}{2})^{d-2}}
 \sum_{k=0}^\infty 
   \frac{(-1)^k\Gamma(\frac{d-2}{2}-k)(\frac12+\beta-k)_{2k}
}{k!}\sinh^{2k}(\tfrac{r}{2}) \label{integg1}
\\
  +\frac{1}{(4\pi)^{\frac{d}{2}}}
 \sum_{j=0}^\infty\frac{(-1)^j(\frac{3-d}{2}+\beta-j)_{d-2+2j}
 \Gamma(\frac{2-d}{2}-j)}{j!}\sinh^{2j}(\tfrac{r}{2}), 
&            \hspace{-8ex}          $d\not\in 2\zz$; \notag 
 \vspace{18pt} \\ 
 \frac{1}{(4\pi)^{\frac{d}{2}}\sinh(\frac{r}{2})^{d-2}}
\sum_{k=0}^{\frac{d-4}{2}}
\frac{(\frac12+\beta-k)_{2k}(\frac{d-4}{2}-k)!  (-1)^k }{k!}
                      \sinh^{2k}\big(\tfrac{r}{2} \big)
\label{integ-} \\ \notag 
+\frac{ (-1)^{\frac{d-2}{2}}} {(4\pi)^{\frac{d}{2}}} 
\sum_{j=0}^\infty \frac{(\frac{3-d}{2}
+\beta-j)_{d-2+2j}}{j!(j+\frac{d-2}{2})!}\sinh^{2j}\big(\tfrac{r}{2}\big)\Big(
 H_{\tfrac{d-2}{2}+j}+H_j-2\gamma_\mathrm{E}
 \\ \notag 
 - \psi(\tfrac{d-1}{2} + \beta+j) 
- \psi(\tfrac{3-d}{2} +\beta - j) 
- \ln\big( \sinh^{2}(\tfrac{r}{2}) \big) \Big), 
&\hspace{-8ex} $d=2,4,\dots$.
\end{subnumcases}
 From this one easily gets for  $2+2n >\Re d $
\begin{align} \label{eq:Gh_diagonal}
    \lim_{x \to x'} &\frac{\partial^n}{\partial z^n} G_d^{\h}(z ; x,x')  = \frac{\partial^n}{\partial z^n} 
    \frac{\Gamma(\frac{2-d}{2}) (\frac{3-d}{2}+\sqrt{-z})_{d-2}}{(4 \pi)^{\frac{d}{2}}}
                      \quad  d \not \in 2 \zz, \\
     = & \frac{\partial^n}{\partial z^n} \frac{(-1)^{\frac{d}{2}}(\frac{3-d}{2}+\sqrt{-z})_{d-2}}{(4
         \pi)^{\frac{d}{2}}\Gamma (\tfrac{d}{2})} (\psi
         (\tfrac{d-1}{2} + \sqrt{-z}) + \psi(\tfrac{3-d}{2} +
         \sqrt{-z}))
        , \quad d= 2,4,6,\dots  \nonumber
\end{align}

For odd dimensions we have an expression in terms of elementary functions:
\begin{align}
 G_d^\h(-\beta^2;x,x') &= 
 \frac{1}{2(2\pi)^{\tfrac{d-1}{2}}}
 \Bigg(-\frac{1}{\sinh {r}}\partial_r
 \Bigg)^{\tfrac{d-3}{2}}
 \frac{\e^{-\beta {r}}}{\sinh {r}} ,\qquad d=3,5,\dots\label{forodd}
\end{align}
To describe the behavior near infinity we use the expansions
\eqref{solu3} respectively \eqref{solu3_form2}  as well as 
$\frac{\cosh(r)+1}{2}=\cosh^2(\frac{r}{2})$ and 
$\frac{\cosh(r)-1}{2}=\sinh^2(\frac{r}{2})$: 
\begin{align}\label{integg4}
 & G_d^\h(-\beta^2;x,x')=\frac{1}{(4\pi)^{\frac{d}{2}}}\sum_{j=0}^\infty
                      \frac{\Gamma(\frac12+\beta+j)\Gamma(\frac{d-1}{2}+\beta+j)}{ j!\Gamma(1+2\beta+j)
                      \big( \cosh(\frac{r}{2})\big)^{2j+d-1+2\beta}}
 \\ \notag
  &=\frac{\Gamma\big(\tfrac{d-1}{2}+\beta\big) 
  }{ (4\pi)^{\frac{d}{2}}
 \big( \sinh(\frac{r}{2})\big)^{d-2}}
\sum_{j=0}^\infty\frac{
\Gamma\big(\frac12+\beta+j\big)
\big(-\tfrac{d+1}{2}+\beta\big)_j
}{j! \Gamma(1+2\beta+j)\big( \cosh(\frac{r}{2})\big)^{2j+1+2\beta}} .
                      \end{align}
Note that \eqref{integg1}, \eqref{integ-}, \eqref{forodd} and
\eqref{integg4} are the analogs of \eqref{explicit1}, \eqref{explicit2},
\eqref{explicit3} and \eqref{explicit4}.

\subsection{Point potentials on hyperbolic space}

We fix the point $x_0:=(1,0,\dots,0)$ in $\hh^d$.
Green's function of the hyperbolic Laplacian with a~point-like
potential located at $x_0$ has the form
\begin{align}
 \label{eq:Gdgamma-h}
  G_d^{\h,\gamma}(z;x,x')&=   G_d^\h(z;x,x')+ 
  \frac{1}{\gamma(z)+\Sigma_d^{\h}(z)} G_d^\h(z;x,x_0)G_d^\h(z;x_0,x'),
  \end{align}
where\footnote{ Note that we choose the integration 
variable $2w$ in \eqref{integral1}. This choice is only important when 
replacing the standard integral by the anomalous generalized integral, 
which is needed in even $d\geq4$ and has the scaling anomaly 
\eqref{eq:genInt_scaling}. Using $2w$ instead of $w$ ensures that 
the reference self-energy given by the generalized integral 
is asymptotic to the reference self-energy given by the flat generalized 
integral in any dimension -- c.f.~Subsection \ref{ssc:flat_limit_hyp}.}
\begin{align}
-\frac{\dd}{\dd z}
(\gamma(z)+\Sigma_d^\h(z)) =&
 \label{eq:sigma_hyper_generic-h}
  \sigma_d^\h(z):=\int_{\hh^d}G_d^\h(z;x_0,x)^2\d x
  \\ \notag
  =&\int_{1}^\infty
    G_d^\h(z,w)^2|\SS^{d-1}|(w^2-1)^{\frac{d}{2}-1}\d w
  \\ 
 =& \frac{\pi \Gamma\big(\tfrac{d-1}{2}+\beta\big)^2
    }{ 2^{2\beta+1}(4\pi)^{\tfrac{d}{2}} \Gamma\big(\tfrac{d}{2}\big)}
   \int_{2}^\infty
    {\bf Z}_{\tfrac{d}{2}-1,\beta}(w)^2(w^2-1)^{\frac{d}{2}-1}\d2 w.\label{integral1}
\end{align}
The integrand of \eqref{integral1} is well-defined
for any complex $d$.  As in the flat case, the integral is convergent 
 only for $|\Re(d-2)|<2$, which includes the dimensions
$d=1,2,3$: 
\begin{align} \label{integral}
  \sigma_d^\h(-\beta^2)
  =\begin{cases} 
    \dfrac{\pi (\frac{3-d}{2}+\beta)_{d-2}H_{d-2}(\frac{3-d}{2}+\beta)
    }{ (4\pi)^{\tfrac{d}{2}} \,2\beta\Gamma(\frac{d}{2})
    \,\sin\big(\pi\tfrac{d}{2}\big) },
    &\quad    |\Re d-2|<2,\;d\neq2;  \vspace{18pt} \\
   \dfrac{\psi'(\frac12+\beta)}{4\pi\beta}, &\quad d=2.
   \end{cases}
\end{align}
In terms of $z=-\beta^2$ these formulas can be conveniently
rewritten as 
\begin{align}
  \sigma_d^\h(z)
  =\begin{cases}
    \dfrac{\pi\, \partial_z\, 
    (\frac{3-d}{2}+\sqrt{-z})_{d-2}}{ (4\pi)^{\tfrac{d}{2}}  \Gamma\big(\tfrac{d}{2}\big)
    \,\sin\big(\pi\tfrac{d}{2}\big)}    , &\quad 
 |\Re d-2|<2,\;d\neq2;\label{integral-1}  \vspace{18pt} \\
  -\dfrac{1}{2\pi}\partial_z \psi(\tfrac12+\sqrt{-z}), 
  &\quad d=2.
   \end{cases}
\end{align}

Based on the same arguments as in the Euclidean case, we define $\sigma_d^\h(z)$ for higher dimensions as a generalized integral, choose an anti-derivative $\Sigma_d^\h$ and let $\gamma(z)$ be a polynomial of degree the smallest integer greater than $\frac{d-4}{2}$.
As in the Euclidean case, this $\Sigma^\h_d$ satisfies for $n$ large enough
\begin{equation}
     \frac{\partial^n}{ \partial z^n} \Sigma_d^\h(z) = -  \lim_{x' \to x} \frac{\partial^n}{ \partial z^n}  G_d^\h(z;x,x'),
\end{equation}
and the right hand side was computed in \eqref{eq:Gh_diagonal}.
For the following discussion of special cases, we introduce the notation 
$[x|x']=\cosh {r}$, $[x|x_0]=\cosh \theta $ and $[x'|x_0]=\cosh \theta'$.

\paragraph{Dimension  $1$.}
We have 
\begin{align}
 {\bf Z}_{-\tfrac12,\beta}(\cosh {r})
 &= \frac{2^\beta}{\Gamma(1+\beta)} \e^{-\beta {r}} 
 \quad\textup{and}\quad
 G_1^\h(-\beta^2;x,x') = \frac{\e^{-\beta {r}}}{2\beta}.
\end{align}
Moreover, \eqref{integral-1} gives 
\beq
\sigma_1^\h(z)
= \frac{1}{2} \partial_z
\frac{1}{\sqrt{-z}}. \eeq
Imposing $\Sigma_1^\h(-\infty)=0$, we find 
\beq \Sigma_1^\h(-\beta^2) = -\frac{1}{2\beta},\eeq
which coincides with the Euclidean $\Sigma_1(-\beta^2)$. Thus, 
\begin{align}
 G_1^{\h,\gamma}(-\beta^2;x,x') 
 = \frac{\e^{-\beta {r}}}{2\beta} 
 + \frac{\e^{-\beta \theta } \e^{-\beta \theta '}
 }{(2\beta)^2 \big(\gamma -\tfrac{1}{2\beta}\big)},
\end{align}
and we obtain the same expression as in the Euclidean case, in accord with the isometry $\hh^1 \cong \rr^1$.

\paragraph{Dimension $2$.}
We have
\begin{align}
G_2^\h(-\beta^2;x,x') 
  = \frac{\Gamma\big(\tfrac12+\beta\big)}{\sqrt{2\pi} 2^{\beta+1}}
    {\bf Z}_{0,\beta}(\cosh {r}).
\end{align}
From \eqref{integral-1} we obtain a family of self-energies depending
on a parameter $\varepsilon:=-2\pi\gamma$:
\begin{align}
\Sigma_2^{\h,\varepsilon}(-\beta^2)
=\frac{1}{2\pi} \big( \psi(\tfrac12+\beta)-\varepsilon\big) .
\end{align}  
Thus
\begin{align}
 G_2^{\h,\varepsilon}(-\beta^2;x,x') 
 =& \frac{\Gamma\big(\tfrac12+\beta\big)}{\sqrt{2\pi} 2^{\beta+1}}
    {\bf Z}_{0,\beta}(\cosh {r})\\
  &    +\frac{\Gamma\big(\tfrac12+\beta\big)^2}{2^{2\beta+2}}
    \;\frac{{\bf Z}_{0,\beta}(\cosh \theta ) {\bf Z}_{0,\beta}(\cosh \theta ')
    }{\psi(\tfrac12+\beta)
-\varepsilon    }. \nonumber
\end{align}

\paragraph{Dimension  $3$.} 
We have 
\begin{align}
 {\bf Z}_{\tfrac12,\beta}(\cosh {r})
 &=  \frac{2^\beta}{\Gamma(1+\beta)} \frac{\e^{-\beta {r}}}{\sinh {r}}
 \quad\textup{and}\quad
 G_3^\h(-\beta^2;x,x') = \frac{\e^{-\beta {r}}}{4\pi\,\sinh {r}}
 .
\end{align}
Moreover, \eqref{integral-1} 
gives \beq\sigma_3^\h(z) = -\tfrac{1}{4\pi} \partial_z \sqrt{-z},\eeq
so that imposing $\Sigma_3^\h(0)=0$ yields
\beq\Sigma_3^\h(-\beta^2) = \frac{\beta}{4\pi}.\eeq
In dimension 3, the hyperbolic self-energy equals the 
Euclidean self-energy: 
$\Sigma_3^\h(-\beta^2)=\Sigma_3(-\beta^2)$.
However, Green's function is different:
\begin{align}
 G_3^{\h,\gamma}(-\beta^2;x,x') 
 = \frac{\e^{-\beta {r}}}{4\pi\,\sinh {r}}
  + \frac{\e^{-\beta \theta }\e^{-\beta \theta '}}{(4\pi)^2\,\sinh \theta \,\sinh \theta ' 
  \big( \gamma +\tfrac{\beta}{4\pi} \big)}.
\end{align}

\paragraph{Odd dimensions $d\geq5$.}
The identities \eqref{integral} and \eqref{integral-1}
remain valid for $d\in\cc\backslash 2\zz$ if the integrals
are interpreted in the generalized sense. Therefore, we can set
\begin{align}
  \Sigma_d^\h(-\beta^2)
  =& -
    \frac{\pi\,
    (\frac{3-d}{2}+\beta)_{d-2} }{ (4\pi)^{\tfrac{d}{2}}  \Gamma\big(\tfrac{d}{2}\big)
     \,\sin\big(\pi\tfrac{d}{2}\big)}
     .\label{integral-1.}
\end{align}
If $d$ is an odd integer, we 
can rewrite \eqref{integral-1.} as
\begin{align}
 \Sigma_d^\h(-\beta^2) 
 &=
   \frac{\pi }{ (4\pi)^{\tfrac{d}{2}} \, 
 \Gamma\big(\tfrac{d}{2}\big)}
 \beta \prod_{j=1}^{\tfrac{d-3}{2}} \big(-\beta^2 + j^2\big). \label{forodd1}
\end{align}
Inserting \eqref{forodd} and \eqref{forodd1} 
into \eqref{eq:Gdgamma-h} we obtain an expression for
$  G_d^{\h,\gamma}(z;x,x')$.

\paragraph{Even dimensions  $d\geq4$.}

In this case the formula \eqref{integral-1.} is not applicable. 
Instead we introduce a family of reference self-energies 
parametrized by $\varepsilon\in\rr$: 
 \begin{align}\label{integ---} 
 \Sigma_{d}^{\h,\varepsilon}(-\beta^2):=   
 \frac{\psi\big(\tfrac{3-d}{2}+\beta\big) 
  +\psi\big(\tfrac{d-1}{2}+\beta\big) 
  -2\varepsilon
    }{ (4\pi)^{\tfrac{d}{2}} \Gamma\big(\tfrac{d}{2}\big)} \;
\prod_{j=0}^{\tfrac{d-4}{2}} 
 \Big(-\beta^2+\big(\tfrac12 + j\big)^2\Big).
   \end{align}
   We obtain a family of Green's function
\begin{align}
 \label{eq:Gdgamma-h.}
  G_d^{\h,\varepsilon,\eta}(z;x,x')&=   G_d^\h(z;x,x')+ 
  \frac{G_d^\h(z;x,x_0)G_d^\h(z;x_0,x')}{\eta(z)+\Sigma_d^{\h,\varepsilon}(z)} , 
  \end{align}
where $\deg\eta\leq\frac{d-4}{2}$.

Let us derive \eqref{integ---}   from 
the integral \eqref{integral1}.  For $d\in2\zz$, $d>2$, it has to 
be understood in the generalized anomalous sense, and is not equal 
to \eqref{integral}. Instead, the anomalous integral 
given by \eqref{eq:genintZZ_limit.} yields
\begin{align}\notag
  \sigma_d^\h(-\beta^2)=&\frac{(-1)^{\frac{d}{2}-1}(\frac{3-d}{2}+\beta)_{d-2}}
  {(4\pi)^{\frac{d}{2}}\Gamma(\frac{d}{2})
  }\Bigg(\tfrac{\psi'(\frac{3-d}{2}+\beta)+\psi'(\frac{d-1}{2}+\beta)}{2\beta}\\\notag
&\qquad
+\frac{H_{\frac{d-2}{2}}(\frac12-\beta)-H_{\frac{d-2}{2}}(\frac12+\beta)}{2\beta} 
\ln 4 \\&+\sum_{k=0}^{\frac{d-4}{2}}\frac{\psi(\frac32+k+\beta)+\psi(-\frac12-k+\beta)-\psi(\frac{d}{2}-1-k)-\psi(1+k)}{\beta^2-(\frac12+k)^2}\Bigg).
                 \label{qws1}
\end{align}
We notice that $\partial_z = -\frac{1}{2\beta}\partial_\beta$ and 
$\partial_z (z)_k = H_k(z)(z)_k$. Using this, the Leibniz 
rule for $\partial_z$ and identities satisfied by the Pochhammer symbol, harmonic numbers and the digamma function (see Appendix \ref{poch}) yields
\begin{align}\notag
  \sigma_d^\h(z)
  =& \frac{ (-1)^{\tfrac{d}{2}} 
    }{ (4\pi)^{\tfrac{d}{2}} \Gamma\big(\tfrac{d}{2}\big)}
     \partial_z \Bigg(
\Big(
  \psi\big(\tfrac{3-d}{2}+\sqrt{-z}\big)
  +\psi\big(\tfrac{d-1}{2}+\sqrt{-z}\big)
  -\ln 4 \Big)
\prod_{j=0}^{\tfrac{d-4}{2}} 
 \Big(-z-\big(\tfrac12 + j\big)^2\Big) \Bigg)
  \\\label{even11}
&+ \frac{ 1
    }{ (4\pi)^{\tfrac{d}{2}} \Gamma\big(\tfrac{d}{2}\big)}
\pi_d^\h(z),\end{align}
where $\pi_d^\h(z)$ is the polynomial of degree
$\frac{d-4}{2}$ defined by
\begin{align} \notag
  \pi_d^\h(z)&=
\Bigg(\sum_{k=0}^{\tfrac{d-4}{2}} 
\frac{\psi\big(\tfrac{d-2}{2}-k\big) +\psi(1+k)}{z+\big(\tfrac12+k\big)^2} 
  \\ \notag &
   + \sum_{k=0}^{\tfrac{d-4}{2}} \sum_{l=k+1}^{\tfrac{d-4}{2}}
   \frac { 2l+1 }{
     \Big(z+\big(\tfrac12+k\big)^2\Big)
              \Big(z+\big(\tfrac12+l\big)^2\Big)}  \Bigg)
              \prod_{j=0}^{\tfrac{d-4}{2}} 
 \Big(z+\big(\tfrac12 + j\big)^2\Big).
\end{align}

Therefore, the following function is an antiderivative of minus
\eqref{even11} and is a possible self-energy:
\begin{align}
\Sigma_d^{\h,\mathrm{ms}}(-\beta^2)&=
 \frac{  \psi\big(\tfrac{3-d}{2}+\beta\big)
  +\psi\big(\tfrac{d-1}{2}+\beta\big)
  -\ln 4 
    }{ (4\pi)^{\tfrac{d}{2}} \Gamma\big(\tfrac{d}{2}\big)}
                       \notag
\prod_{j=0}^{\tfrac{d-4}{2}} 
 \Big(-\beta^2+\big(\tfrac12 + j\big)^2\Big)\\&  +  \frac{ 1
    }{ (4\pi)^{\tfrac{d}{2}} \Gamma\big(\tfrac{d}{2}\big)}
     \Pi_d^\h(-\beta^2), \label{insero1}\end{align}
   where
\begin{align}
\label{polypi1}
\Pi_d^\h(z):=-\int_0^{z}\pi_d^\h(\tau)\d\tau,
\end{align} 
which is a polynomial of degree $\frac{d-2}2$  with 
$\Pi_d^\h(0)=0$.
 \eqref{insero1} will be called the {\em reference  
self-energy based on minimal subtraction}. (The superscript $\mathrm{ms}$ stands for the ``minimal subtraction''). 

 As in the Euclidean case, we can add to \eqref{insero1} an arbitrary 
 polynomial 
$\gamma(-\beta^2)$ of degree $\leq\frac{d-2}2$.
Let $\frac{\ln 4-2\varepsilon}{(4\pi)^\frac{d}{2}\Gamma(\frac{d}{2})}$ be the coefficient at the term
$z^{\frac{d-2}{2}}$ of $\frac{\Pi_d^\h(z)}{(4\pi)^\frac{d}{2}\Gamma(\frac{d}{2})}+\gamma(z)$. Then we can
write
\beq
\Sigma_d^{\h,\mathrm{ms}}(-\beta^2)+
\gamma(-\beta^2)=
\Sigma_d^{\h,\varepsilon}(-\beta^2)+
\eta(-\beta^2),\eeq
where $\eta$ is a polynomial of degree $\leq\frac{d-4}{2}$ and 
$\Sigma_d^{\h,\varepsilon}$ was introduced in
\eqref{integ---}. We prefer the latter as the family of reference
self-energies because of the factorized form.

\begin{remark}
We have two proposals
for the reference self-energy for even $d\geq4$:
$\Sigma_d^{\h,\varepsilon}$, $\varepsilon\in\rr$, in \eqref{integ---} and
$\Sigma_d^{\h,\mathrm{ms}}$ in \eqref{insero1}.
We choose the former as the standard
one, because of its simplicity. In particular, it is factorized, which allows to
determine easily its zeros responsible for singularities of the
corresponding Green's functions.

However $\Sigma_d^{\h,\mathrm{ms}}$ is also in some sense special. It is obtained
with the help of the generalized integral, a concept closely related
to the minimal subtraction method in QFT. One can criticize it saying
that because of anomaly it
depends on the choice of the integration variable, which in the
hyperbolic case is chosen to be $2(w-1)=2(\cosh r-1)$. However, this is actually a
natural variable. It is closely related to the family of
conformal transformations, which are best expressed in the variable  $w$:
         \begin{equation}
             \phi_\lambda(w) = \frac{\lambda^{-1}(w+1)+\lambda (w-1)}{\lambda^{-1} (w+1) - \lambda (w-1)}.
             \label{popi} \end{equation}
$\phi_\lambda$ form a 1-parameter group
 $\phi_\lambda \circ \phi_\mu = \phi_{\lambda \mu}$ and 
$\phi_1 = \mathrm{id}$.

As outlined in  the introduction, Green's functions that we introduce
most likely describe the asymptotics of the resolvent of the Laplacian
with a  perturbation supported in a shrinking region. We expect that
fine-tuning the perturbation we should be able to see Green's
functions corresponding to various $\varepsilon,\eta$. One can ask the
question whether the self-energy
$\Sigma_d^{\h,\mathrm{ms}}$ (based on the generalized integral) is
distinguished and obtained by a special way of shrinking the
perturbation? As of now, we do not know.
 \end{remark}

\paragraph{Spectral properties.} Green's functions $G_{d}^{\h,\gamma}(z)$ 
and $G_d^{\gamma}(z)$ have a cut at $z\in[0,\infty[$, which in dimensions $1,2,3$ corresponds
to the continuous spectrum of $H_{d}^{\h,\gamma}$ and
$H_{d}^{\gamma}$. There may be also some
singularities outside of $[0,\infty[$, which in the hyperbolic case, apart from dimension
$d=1,3$, have a more
complicated structure  than in the flat case, because 
 the logarithm is replaced by the digamma function. Note that
in dimensions $d=1,3$ the poles are exactly the same as in the flat
case, because
\begin{align}
 &\Sigma_{1}^\h(-\beta^2)=\Sigma_1(-\beta^2)
 \quad\text{and}\quad 
 \Sigma_{3}^\h(-\beta^2)
=\Sigma_3(-\beta^2).
\end{align}

\subsection{Flat limit of the hyperbolic Laplacian}
\label{ssc:flat_limit_hyp}
           Let $R>0$.
 Instead of the  hyperbolic space of curvature $-1$, we can use its
 scaled version of curvature $-\frac{1}{R^2}$:
 \[\hh_R^d:=\{x\in\rr^{1,d}
 \mid[x|x]=R^2
 \}.\]
 We can introduce various objects from the previous subsection
 corresponding to $\hh_R^d$, which will be distinguished by the
 subscript $R$.
         Clearly,  $\hh_1^d=\hh^d\ni x\mapsto Rx\in\hh_R^d$ is a bijection
         and
         \[ d^{\h}_{  R}(Rx,Rx')=R \, d^{\h}(x,x').\]
  The map $U_R:L^2(\hh^d)\to L^2(\hh_R^d)$ given by
         \[U_Rf(x):=R^{-\frac{d}{2}}f\Big(\frac{x}{R}\Big)\]
         is unitary.  If $K(x,x')$ is the integral kernel of $K$ on
         $L^2(\hh^d)$ and $K_R(x,x')$ 
         is the integral kernel of $U_RKU_R^{-1}$ on $L^2(\hh_R^d)$, then
         \beq
         K_R(x,x')=R^{-d}K\Big(\frac{x}{R},\frac{x'}{R}\Big).\eeq
       The   hyperbolic Laplacian on $L^2(\hh_R^d)$ is 
               \beq\Delta_{d,R}^\h=\frac{1}{R^2}U_R\Delta_d^\h U_R^{-1}.\eeq
We set
         \beq H_{d,R}^\h:=-\Delta_{d,R}^\h-\frac{(d-1)^2}{4R^2},\eeq
so that $\sigma(H_{d,R}^\h)=[0,\infty[$.
    For  $z\in\cc\backslash[0,\infty[$ we set
        \begin{align}
        G_{d,R}^\h(z)&:=
                       (-z+H_{d,R}^\h)^{-1}, \label{eq:free_green_scaled} \\
          \pp_{d,R}^\h(a,b)&:=           \one_{[a,b]}(H_{d,R}^\h).
    \end{align}
Note that 
\begin{align}
     G_{d,R}^\h(-\beta^2;x,x')&=R^{-d+2}G_d^\h\Big(-(\beta R)^2;\frac{x}{R},\frac{x'}{R}\Big), \label{eq:Green_scaling} \\    
           \pp_{d,R}^\h(
           a,b;x,x') &=R^{-d}
    \pp_d^\h\Big(aR^2,
bR^2;\frac{x}{R},\frac{x'}{R}\Big).
\end{align}
Proceeding as for $R =1$, we introduce
\begin{align}
 G_{d,R}^{\h,\gamma}(z;x,x')
=   
  G_{d,R}^\h(z;x,x') + 
  \frac{ G_{d,R}^\h\big(z;x,R x_{0}\big)G_{d,R}^\h\big(z;R x_{0},x'\big)  }{\gamma(z) +\Sigma_{d,R}^\h(z)}, \label{eq:scaled_perturbed_Green_def}
\end{align}
where $\gamma(z)$ is a polynomial (of degree as for $R=1$), and $\Sigma_{d,R}^{\h}$ is a particular solution of 
\begin{equation}
    - \frac{\dd}{\dd z}  \Sigma_{d,R}^{\h}(z) = \sigma_{d , R}^{\h}(z):= \int_{\hh^d_R} G_{d,R}^{\h}(z; x, Rx_0)^2 \dd x.
    \label{eq:Sigma_der_R}
\end{equation}
For dimensions $d$ for which the integral does not converge, we integrate over angles and then compute the generalized integral with respect to the radial coordinate $2Rw_R =R^2 [ \frac{x}{R} | x_0 ]$:
\begin{align}
\sigma^{\h}_{d,R}(z) &= 
\gen \int_{2R^2}^\infty |\SS^{d-1}| G_{d,R}^\h(z;x,R x_0)^2 (w_R^2 - R^2)^{\frac{d-2}{2}} R \frac{\dd (2Rw_R)}{2R}  \\ \notag
&= R^d \, \gen \int_{2R^2}^\infty |\SS^{d-1}| (R^{-d+2} G_{d}^\h(R^2z ; w))^2 (w^2 -1)^{\frac{d-2}{2}} \frac{d (2 R^2 w)}{2R^2}  \\ \notag
& = R^{-d +4 }  
\begin{cases}
\sigma_d^\h(R^2 z), & d \not \in \{ 4, 6, \dots \}, \vspace{18pt}\\
\sigma_d^\h(R^2 z) - \frac{2 (-1)^{\frac{d}{2}}}{(4 \pi)^{\frac{d}{2}} \Gamma (\frac{d}{2})} \dfrac{\partial}{\partial (R^2 z)} 
\dfrac{\Gamma \big(\tfrac{d-1}{2}+ \sqrt{-R^2 z} \big)}{\Gamma \big( \tfrac{3-d}{2}+ \sqrt{- R^2 z} \big)}  \ln R, & d \in \{ 2 , 4 , \dots \}.
\end{cases}
\end{align}
In the last equality we used \eqref{eq:genInt_scaling2} ; the scaling anomaly coefficient $f_{-1}$ was computed in \cite{DGR23a}.
  In even dimensions we prefer to use
\begin{align}
 G_{d,R}^{\h,\varepsilon,\eta}(z;x,x')
=   
  G_{d,R}^{\h}(z;x,x') + 
  \frac{ G_{d,R}^\h\big(z;x,R x_{0}\big)G_{d,R}^\h\big(z;R x_{0},x'\big)  }{\eta(z) +\Sigma_{d,R}^{\h,\varepsilon}(z)}.
\end{align}
We choose the reference self-energies to be
\begin{align*}
 \Sigma_{d,R}^\h(-\beta^2)&:=R^{2-d}
    \Sigma_{d}^\h\big(-(\beta R)^2\big), &\quad d\text{ odd};\\
 \Sigma_{d,R}^{\h,\varepsilon}(-\beta^2)&:=R^{2-d}
   \Sigma_d^{\h,\varepsilon+\ln R}\big(-(\beta R)^2\big),&\quad
                                                             d\text{ even}.
\end{align*}

\bet \label{thm:flat_limit}
Let $-\beta^2\in\cc\backslash[0,\infty[$. We have 
\begin{align}
\label{eq:limit_Gh_free}
    G_{d,R}^\h\big(-\beta^2; r\big)&=
    G_d\big(-\beta^2; r\big)
     \Big(1+\mathcal{O}\big(\tfrac{1}{\beta R}\big)
    +\mathcal{O}\big(\tfrac{r}{ R}\big)\Big)
\end{align}
and
\begin{align}
  \label{eq:asymp_R_Sigma_h}
 \Sigma_{d,R}^\h\big(-\beta^2\big)
& =\Sigma_{d}\big(-\beta^2\big) 
                                     \Big(1+\mathcal{O}\big(\tfrac{1}{\beta
                                     R}\big) \Big),\quad d\text{ odd};\\
 \Sigma_{d,R}^{\h,\varepsilon}\big(-\beta^2\big)  & =\Sigma_{d}^\varepsilon\big(-\beta^2\big) 
    \Big(1+\mathcal{O}\big(\tfrac{1}{\beta R}\big) \Big),\quad d\text{ even}.
\end{align}
Thus if we have a family $x_R,x_R'\in\hh_R^d$ and $x,x'\in\rr^d$ such that
\begin{align}& \label{dfdf1}
\lim_{R\to\infty}d^\h_R(x_R,x_R')=|x-x'|,\\ \notag &
\lim_{R\to\infty}d^\h_R(x_R,R x_{0})=|x|, \\ \notag & 
\lim_{R\to\infty}d^\h_R(x_R',Rx_{0})=|x'|, 
\end{align}
then
\begin{align}
  \lim_{R\to\infty} G_{d,R}^{\h,\gamma}\big(-\beta^2; x_R,x'_R\big)&=
    G_d^\gamma(-\beta^2;x,x'),\qquad d\text{ odd}; 
                                                                     \label{eq:limit_Gh_pert-}\\
  \lim_{R\to\infty} G_{d,R}^{\h,\varepsilon,\eta}\big(-\beta^2; x_R,x'_R\big)&=
    G_d^{\varepsilon,\eta}(-\beta^2;x,x'),\qquad d\text{ even}. 
                                                                     \label{eq:limit_Gh_pert+}
\end{align}
\eet

\begin{proof} 
Using the asymptotics of the Gegenbauer 
functions from Thm. \ref{thm:asymptotics}, we find 
\begin{align}
  G_{d,R}^\h(-\beta^2 ,r_R)
  =&R^{-d+2}G_d^\h\Big(-(\beta R)^2,\cosh\frac{r_R}{R}\Big)
  \\ \notag
  =&\frac{R^{-d+2}\Gamma(\frac{d-1}{2}+\beta
     R)}{\Gamma(\frac{3-d}{2}+\beta
     R)}\frac{\sqrt\pi\Gamma(\frac{3-d}{2}+\beta R)}{\sqrt2(2\pi)^{\frac{d}{2}}2^{\beta R}}{\bf
     Z}_{\frac{d-2}{2},\beta R}\Big(\cosh\frac{r_R}{R}\Big)
     \\ \notag
 =&\frac{(\frac{r_R}{R})^{\frac{d-1}{2}}}{(\sinh\frac{r_R}{R})^{\frac{d-1}{2}}(2\pi)^{\frac{d}{2}}}\Big(\frac{\beta}{r_R}\Big)^{\frac{d-2}{2}}
       K_{\frac{d-2}{2}}(\beta r_R) \Big(1+\mathcal{O}\big(\tfrac{1}{\beta R}\big)\Big)
  \\ \notag
 =&\frac{1}{(2\pi)^{\frac{d}{2}}}\Big(\frac{\beta}{r_R}\Big)^{\frac{d-2}{2}}
       K_{\frac{d-2}{2}}(\beta r_R) \Big(1+\mathcal{O}\big(\tfrac{1}{\beta R}\big)
       +\mathcal{O}\big(\tfrac{r_R}{ R}\big)\Big).
       \end{align}
This proves \eqref{eq:limit_Gh_free}. \eqref{eq:asymp_R_Sigma_h}  
   follows from \ref{thm:asymp_genintSZ} and
    \begin{align}
     \psi\big(\tfrac12\pm \alpha+ \beta R\big) -\ln(\beta R)
 &= \mathcal{O}\big(\tfrac{1}{\beta R}\big).
    \end{align}
    Now let
    \begin{align}
r_R:=d^\h_R(x_R,x_R'),&\quad 
\theta_R:=d^\h_R(x_R,R x_{0}),\quad 
                       \theta_R':=d^\h_R(x_R',R x_{0}),\\
      r:=|x-x'|,&\quad 
\theta:=|x-x_{0}|,\quad 
                       \theta':=|x'-x_{0}|.\end{align}      
    By \eqref{eq:limit_Gh_free} and \eqref{eq:asymp_R_Sigma_h},  we obtain
\begin{align}
 G_{d,R}^{\h,\gamma}(-\beta^2;x_R,x'_R)
  &=
 G_{d,R}^\h(-\beta^2,r_R) + 
  \frac{ G_{d,R}^\h\big(-\beta^2,\theta_R\big) 
  G_{d,R}^\h\big(-\beta^2,\theta'_R\big) 
  }{\gamma(- \beta^2) +\Sigma_{d,R}^\h(-\beta^2)}. 
\\ \notag 
&=  G_{d}(-\beta^2,r_R) 
  \Big(1+\mathcal{O}\big(\tfrac{1}{\beta R}\big) +\mathcal{O}\big(\tfrac{r_R}{R}\big)\Big) 
  \\ \notag &\quad+ 
  \frac{ G_{d}\big(-\beta^2,\theta_R\big) 
  G_{d}\big(-\beta^2,\theta'_R\big) 
  }{\gamma(- \beta^2) +  \Sigma_{d}\big(-\beta^2\big) 
    \big(1+\mathcal{O}\big(\tfrac{1}{\beta R}\big) \big)}
    \Big(1+\mathcal{O}\big(\tfrac{1}{\beta R}\big) 
    +\mathcal{O}\big(\tfrac{\theta_R}{R}\big) +\mathcal{O}\big(\tfrac{\theta'_R}{R}\big)\Big). 
\end{align}
Now \eqref{dfdf1} implies $\frac{r_R}{R}$, $\frac{\theta_R}{R}$ and
$\frac{\theta'}{R}$ to be $\mathcal{O}(\frac1R)$. Hence, 
the limit $R\to\infty$ of the right-hand side is
\begin{align}
   G_{d}(-\beta^2,r) 
 +
    \frac{ G_{d}\big(-\beta^2,\theta\big) 
  G_{d}\big(-\beta^2,\theta'\big) 
  }{\gamma(-\beta^2) +  \Sigma_{d}\big(-\beta^2\big) }=G_d^\gamma(-\beta^2;x,x'),
   \end{align}
which
proves  
\eqref{eq:limit_Gh_pert-}. \qed
\end{proof}

\section{Green's operators on the sphere}
\label{sec:spherical}
\subsection{Spherical Laplacian}

        Equip the space $\rr^{1+d}$ with the Euclidean bilinear form
         \[(x|y)         =x^0y^{0}+x^1y^{1}+\dots+x^dy^{d}.
         \]
         The set
         \[\SS^d:=\{x\in\rr^{1+d}
         \mid(x|x)=1
         \}\]
         equipped with the Riemannian metric inherited from $\rr^{1+d}$ is called
         the {\em (unit) sphere}.
         The geodesic distance between $x,x'\in\SS^d$ is given by
         \beq 
d_{\s} (x, x')= \cos^{-1}(x|x'),\qquad\cos d_{\s}(x,x')=(x|x'). 
         \eeq
$\SS^d$  has also a natural measure, so one can define
         $L^2(\SS^d)$.

In this section we study the  operator
         \beq H_d^\s:=-\Delta_{d}^\s+\frac{(d-1)^2}{4},\eeq
         where
 $\Delta_{d}^\s$ is    the      {\em
   spherical Laplacian}.
 For  $z\in
 \cc\backslash\sigma(H_d^\s)\supset\cc\backslash[0,\infty[$ we define the 
   {\em spherical Green's operator} $ G_{d}^\s(z):= (-z+H_{d}^\h)^{-1}$. 
 Eigenfunctions of $H_d^\s$ with eigenvalue   $(l+\frac{d-1}{2})^2$
 are called 
   {\em  spherical harmonics of degree $l$}. The corresponding 
   spectral projection will be denoted\beq 
     \pp_{d,l}^\s:=           \one_{l(l+d-1)}\big(-\Delta_{d}^\s\big) 
     =    \one_{(l+\frac{d-1}{2})^2}\big(H_{d}^\s\big) .
     \eeq

   In the following theorem,  we express the integral kernels of
 $ G_{d}^\s(z)$
   and $\pp_{d,l}^\s$
in terms of
Gegenbauer function $\bf{S}_{\alpha,\beta}$ and Gegenbauer 
polynomials $C_l^\alpha$, see Appendix \ref{app:gegenbauer}.

    \bet \begin{enumerate}\item
     For $\Re\beta>0$ the integral kernel of $G_d^\s(-\beta^2)$ is
     \beq\label{iao1}
     G_{d}^\s(-\beta^2;x,x')
     =\frac{\Gamma(\frac{d-1}{2}+\i\beta )
       \Gamma(\frac{d-1}{2}-\i\beta )
     }{(4\pi)^{\frac{d}{2}}
         }{\bf S}_{\frac{d}{2}-1,\i\beta }         \big(-(x|x')\big).\eeq
   \item\eqref{iao1} is true also on the real line away of the
     spectrum of $H_d^\s$.
It is convenient to rewrite it as follows: for $\zeta\in\rr$, $\zeta-\frac{d-1}{2}\not\in\{0,1,2,\dots\}$,
     \beq\label{iao2}
     G_{d}^\s(\zeta^2;x,x')
     =\frac{\Gamma(\frac{d-1}{2}+\zeta )
       \Gamma(\frac{d-1}{2}-\zeta )
     }{(4\pi)^{\frac{d}{2}}
         }{\bf S}_{\frac{d}{2}-1,\zeta }         \big(-(x|x')\big).\eeq  
\item     The integral kernel of $\pp_{d,l}^\s$ is
     \begin{align}
       \pp_{d,l}^\s(x,x')&=\frac{(2l+d-1)\Gamma(\frac{d-1}{2})}{4\pi^{\frac{d+1}{2}}}C_l^{\frac{d-1}{2}}\big((x|x')\big)
    \label{jaco1} \end{align}
\end{enumerate}\eet
\begin{proof} Let  $w:=-(x|x')$. By the spherical symmetry, there exists 
a function $G_d^\s(z,w)$ such that, $G_{d}^\s(z,x,x')=:G_d^\s(z,w)$. 
We obtain the differential equation 
\begin{align}
0=&\left((1-w^2)\partial_w^2-dw\partial_w-\beta^2
-\Big(\frac{d-1}{2}\Big)^2\right)G_d^\s(-\beta^2,w)
 . \end{align}
  Then we need to find the solution regular near $w=-1$ 
     and such that 
     \beq G_d^\s\big(-\beta^2,-\cos(r)\big)\sim 
     \begin{cases}
     \frac{\Gamma(\frac{d}{2}-1)}{4\pi^\frac{d}{2}}r^{2-d}, & d \neq 2, \\
     - \frac{1}{2 \pi} \ln r, & d = 2,
     \end{cases}  
     \eeq
This yields \eqref{iao1}, \eqref{iao2}.
To see (\ref{jaco1}) we use (\ref{stone0}): 
     \begin{align}
&\quad \text{lhs of (\ref{jaco1})}     
\\ \notag &=
 \lim_{\epsilon\searrow0}\Big(\Big(l+\i\epsilon+\frac{d-1}{2}\Big)^2
       -\Big(l+\frac{d-1}{2}\Big)^2\Big)
  G_{d}^\s\Big(\Big(l+\i\epsilon+\frac{d-1}{2}\Big)^2;x,x'\Big)
  \\ \notag
       &=
       \lim_{\epsilon\searrow0}\i\epsilon (2l+d-1)
       \frac{\Gamma(-l-\i\epsilon)\Gamma(d-1+l+\i\epsilon)}
            {2^d\pi^{\frac{d}{2}}}
         {\bf S}_{\frac{d}{2}-1,\frac{d-1}{2}+l+\i\epsilon}
                  \big(-(x|x')\big).\\ \notag
                &=
\frac{(-1)^l(2l+d-1)\Gamma(d-1+l)}
            {l!2^d\pi^{\frac{d}{2}}}
         {\bf S}_{\frac{d}{2}-1,\frac{d-1}{2}+l}
                  \big(-(x|x')\big)\\ \notag
   &=    \frac{(2l+d-1)\Gamma(\frac{d-1}{2})}{4\pi^{\frac{d+1}{2}}}C_l^{\frac{d-1}{2}}\big((x|x')\big).
     \end{align}
     At the end we used  \eqref{polynomial} and
     $\sqrt\pi\Gamma(d-1)=2^{d-2}\Gamma(\frac{d-1}{2})\Gamma(\frac{d}{2})$.  
   \qed      
         \end{proof}

The function $G_d^\s(-\beta^2,-\cos r)$ can be defined for any
$d\in\cc$. However, we do not have a~symmetry similar to
\eqref{hypersym}, except for even integers. To obtain
formulas describing the behavior close to $r=0$, we use
\eqref{formu2}, \eqref{solu1}, \eqref{solu1_form2} and
$\frac{1-\cos r}{2}=\sin^2(\frac{r}{2})$: 
\begin{subnumcases}{G_d^\s(-\beta^2,r)=}
 \frac{1}{\sin^{d-2}(\frac{ r}{2})}
                      \sum_{k=0}^\infty 
                      \frac{(\frac{1}{2}+\i\beta-k)_{2k}\Gamma(\frac{d-2}{2}-k)}
                      {(4\pi)^{\frac{d}{2}} k!}\sin^{2k}\big(\tfrac{r}{2}\big)
 \label{integ-s1} \\ \notag
 +\sum_{j=0}^\infty 
          \frac{(-1)^j (\frac{1}{2}+\i\beta) _{\frac{d-2}{2}+j} 
          (\frac{1}{2}-\i\beta) _{\frac{d-2}{2}+j} \Gamma(\frac{2-d}{2}-j)
          }{(4\pi)^{\frac{d}{2}} j!}
          \sin^{2j}\big(\tfrac{r}{2}\big), 
          & \hspace{-4ex}$d\not\in2\zz$;
 \vspace{18pt} \\ \label{integ-s} 
 \frac{1}{\sin^{d-2}(\frac{r}{2})}\sum_{k=0}^{\frac{d-4}{2}}
   \frac{(\frac12+\i\beta-k)_{2k} (\frac{d-4}{2}-k)!}{(4\pi)^{\frac{d}{2}} k!}
                         \sin^{2k}\big(\tfrac{r}{2} \big)
 \\ \notag 
+ \sum_{j=0}^\infty                                                                                                                                                         \frac{(-1)^{j}(\frac{3-d}{2} +\i \beta-j)_{d-2+2j}}{ (4\pi)^{\frac{d}{2}} 
j! (j+\frac{d-2}{2})! } \sin^{2j}\big(\tfrac{r}{2} \big)
\Big( H_{\tfrac{d-2}{2} + j} +H_j -2\gamma_\mathrm{E}
\\ \notag
- \psi(\tfrac{d-1}{2} + \i\beta+j) 
- \psi(\tfrac{d-1}{2} -\i\beta + j) 
- \ln\big( \sin^{2}(\tfrac{r}{2}) \big) \Big),  &\hspace{-6ex}$d=2,4,\dots$.
\end{subnumcases}
If $2+2n > \Re d$, then
\begin{align}
&\lim_{x \to x'} \partial_z^n G^\s_d(z;x,x') 
\\ \notag 
&=\begin{cases}
    \frac{\Gamma\big(\tfrac{2-d}{2}\big)
                }{(4 \pi)^{\frac{d}{2}}}\, \partial_z^n 
\Big(                (\tfrac12+\i\sqrt{-z})_{\frac{d-2}{2}}(\tfrac12-\i\sqrt{-z})_{\frac{d-2}{2}}
\Big)                            
                , & d \not \in 2 \zz;   \\
\frac{(-1)^{\frac{d}{2}}}{(4 \pi)^{\frac{d}{2}}\Gamma\big(\tfrac{d}{2}\big)} \partial_z^n \Big(    (\psi(\tfrac{d-1}{2}+ \i \sqrt{-z}) + \psi( \tfrac{d-1}{2} - \i \sqrt{-z} )) \prod_{j=0}^{\frac{d-4}{2}} ((\tfrac{1}{2}+j)^2-z) \Big) , & d =2,4, \dots.                
  \end{cases}
\end{align}
For odd integers we have  an expression in terms of elementary functions:
\begin{align}
  G_d^\s(-\beta^2;x,x') &= \frac{
  1}{2(2 \pi)^{\tfrac{d-1}{2}}\sinh\pi\beta}
  \Big(- \frac{1}{\sin r} \partial_r \Big)^{\tfrac{d-3}{2}}
  \frac{\sinh\big(\beta(\pi-r)\big)}{\sin r},\quad d=3,5,\dots
\label{greenodd}\end{align}
To describe the behavior close to pairs of antipodal points, that is
close to $r=\pi$, we use \eqref{solu1}, \eqref{solu1_form2} as well as  
 $\frac{1+\cos r}{2}=\cos^2(\frac{r}{2})$ and $\frac{1-\cos r}{2}=\sin^2(\frac{r}{2})$: 
\begin{align}\label{evo1}
  G_d^\s(-\beta^2;x,x')&=\frac{1}{(4\pi)^{\frac{d}{2}}}\sum_{j=0}^\infty
            \frac{\Gamma(\frac{d-1}{2}+\i\beta+j)
                      \Gamma(\frac{d-1}{2}-\i\beta+j)
                      }{\Gamma(\frac{d}{2}+j)j!}\cos^{2j}\big(\tfrac{r}{2}\big)
    \\ \notag &= \frac{\Gamma(\frac{d-1}{2}+\ii\beta )
       \Gamma(\frac{d-1}{2}-\ii\beta )}{(4\pi)^{\tfrac{d}{2}}\big(\sin^2\big(\tfrac{r}{2}\big)\big)^{\tfrac{d-2}{2}}}
\sum_{j=0}^\infty
    \frac{\big(\frac12+\ii\beta\big)_j\big(\frac12-\ii\beta\big)_j}
                                      {\Gamma\big(\tfrac{d}{2}+j\big)j!}
        \cos^{2j}\big(\tfrac{r}{2}\big).
\end{align}
For integer $d$, \eqref{evo1} can be simplified as follows: 
\begin{subnumcases}{ \hspace{-6ex}G_d^\s(-\beta^2,r)=}
 \tfrac{(-1)^{\frac{d-2}{2}}\pi}{(4\pi)^{\frac{d}{2}}\cosh\pi\beta}\sum_{j=0}^\infty
\frac{(-1)^j
     (\frac{3-d}{2}-j+\i\beta)_{d-2+2j}}{(\frac{d-2}{2}+j)!j!}\cos^{2j}\big(\tfrac{r}{2}\big),
  & \hspace{-1.5ex} $d=2,4,\dots$; \label{eveno1} \vspace{12pt}\\
 \tfrac{(-1)^{\frac{d-1}{2}}\pi}{(4\pi)^{\frac{d}{2}}\i\sinh\pi\beta}\sum_{j=0}^\infty
\frac{(-1)^j
     (\frac{3-d}{2}-j+\i\beta)_{d-2+2j}}{\Gamma(\frac{d}{2}+j)j!}\cos^{2j}\big(\tfrac{r}{2}\big), & \hspace{-1.5ex} $d=3,5,\dots$. \label{eveno2}
\end{subnumcases}

Note that \eqref{integ-s1}, \eqref{integ-s}, \eqref{greenodd} and
\eqref{evo1}
are the analogs of \eqref{explicit1}, \eqref{explicit2},
\eqref{explicit3} and \eqref{explicit4}.

\subsection{Point potentials on the sphere}

         Denote the {\em north pole} of the sphere by $x_0=(1,0,\dots,0)$.
Green's function of the spherical Laplacian with a point-like
potential located at $x_0\in\mathbb{S}^d$ has the form
\begin{align}
 \label{eq:Gdgamma-s}
  G_d^{\s,\gamma}(z;x,x')&=   G_d^\s(z;x,x')+ 
  \frac{ G_d^\s(z;x,x_0)G_d^\s(z;x_0,x')}{\gamma (z)+\Sigma_d^\s(z)},
\end{align}
where (with $z=-\beta^2$)
\begin{align}
-\partial_z\Sigma_d^\s(z) &= \label{eq:sigma_sphere_generic}
  \sigma_d^\s(z):=\int_{\SS^d}G_d^\s(z;x_0,x)^2\d x
  \\ \notag
  &=\int_{-1}^1
    G_d^\s(z;w)^2|\SS^{d-1}|(1-w^2)^{\frac{d}{2}-1}\d w
  \\ 
 &= \frac{\Gamma\big(\tfrac{d-1}{2}+\i\beta\big)^2
    \Gamma\big(\tfrac{d-1}{2}-\i\beta\big)^2
    }{2^{2d} \pi^{\frac{d}{2}} \Gamma\big(\tfrac{d}{2}\big)}
   \int_{-2}^2
    {\bf S}_{\tfrac{d}{2}-1,\i\beta}(w)^2(1-w^2)^{\frac{d}{2}-1}\d 2 w.\label{integralq}
\end{align}
The integrand of \eqref{integralq} is well-defined for any 
complex $d$. Again, the integral is convergent only for 
$|\Re(d-2)|<2$, which includes the dimensions $d=1,2,3$: 
\begin{align}\label{integral-s1}
   \sigma_d^\s(z)
   &= \begin{cases}
   \dfrac{ \Gamma(\tfrac{2-d}{2})}
    {(4\pi)^{\frac{d}{2}}}
     \partial_z
\, \Big(  (\tfrac12+\i\sqrt{-z})_{\frac{d-2}{2}}(\tfrac12-\i\sqrt{-z})_{\frac{d-2}{2}}\Big),
 &\quad |\Re d-2|<2,\;d\neq2; \vspace{12pt}\\
 - \dfrac{1}{4\pi} \partial_z
 \big( \psi\big(\tfrac12+\i\sqrt{-z}\big) 
 +\psi\big(\tfrac12-\i\sqrt{-z}\big) \big), &\quad d=2.
 \end{cases}
\end{align}

Let us discuss specific dimensions.
    Except for dimension 1, we will use the notation
 $\cos r=(x|x')$,  $\cos\theta=(x|x_0)$,  $\cos\theta'=(x'|x_0)$. 

\paragraph{Dimension  $1$.} 
This case can be solved in an elementary way. 
Solving $(-\partial_\theta^2+\beta^2)
g(\theta)=\delta(\theta)$ on $\SS^1=\rr/2\pi\zz$ we obtain
$g(\theta)=\frac{\cosh(\beta(\theta-\pi))}{2 \beta
  \sinh(\pi\beta)}$,  
  $|\theta-\pi|<\pi$. This yields
 the following description of the
 1-dimensional Green's function: 
\begin{align} 
G_1^\s(-\beta^2;\theta,\theta')=
\frac{\cosh(\beta(\theta-\theta'-\pi))}{2\beta\sinh(\pi\beta)},\label{qquu}\end{align}
where $\theta-\theta' \in ]0, 2 \pi [$. We
put the contact potential at $\theta=\pi$ and compute 
\begin{align}
  \sigma_1^\s(-\beta^2)=\int_{-\pi}^{ \pi}
  G_1^\s(-\beta^2,\pi,\theta)^2\d\theta&=       \frac{\pi}{4\beta^2\sinh^2(\pi\beta)}+\frac{\cosh(\pi\beta)}{4\beta^3\sinh(\pi\beta)},\\
  \Sigma_1^\s(-\beta^2)=-\int_{\beta^2}^{\infty}\sigma_1^\s(-\rho)\d\rho&=
-\frac{\coth(\pi\beta)}{2\beta}.
\end{align}
Thus
  \begin{align}   
  G_1^{\s,\gamma}(-\beta^2;\theta,\theta')=
  \frac{\cosh(\beta(\theta-\theta'-\pi))}{2\beta\sinh(\pi\beta)}+\frac{\cosh(\theta\beta)\cosh(\theta'\beta)}{(2\beta)^2\sinh^2(\pi\beta)\big(\gamma-\frac{1}{2\beta}\coth(\pi\beta)\big)}.
  \end{align}
  
 Let us show that it also follows from
the general theory.
We check that
\begin{align}{\bf
  S}_{-\frac12,\i\beta}\big(-\cos(\theta-\theta')\big)
  &=\frac{1}{\sqrt\pi}
\cosh(\beta(\theta-\theta'-\pi)),\\
\Gamma(\i\beta) 
  \Gamma(-\i\beta)&=\frac{\pi}{\beta\sinh(\pi\beta)}.
                    \end{align}
Thus 
\begin{align} 
G_1^\s(-\beta^2;\theta,\theta')=\frac{\Gamma(\i\beta) 
  \Gamma(-\i\beta)}{2\sqrt\pi}{\bf 
  S}_{-\frac12,\i\beta}\big(-\cos(\theta-\theta')\big)
\end{align} 
yields \eqref{qquu}. The formula \eqref{integral-s1} specified to $d=1$
gives
\begin{align} 
\sigma_1^\s(z)
=\partial_z \frac{\coth(\pi\sqrt{-z})}{2\sqrt{-z}}.
\end{align}
Imposing the condition
  $\Sigma_1^\s(-\infty)=0$ yields   
\begin{align}
\label{eq:Sigma_1_s}
 \Sigma_1^\s(-\beta^2) 
 = -\frac{\coth(\pi\beta)}{2\beta }.
\end{align}

\paragraph{Dimension  $2$.}
We have
\begin{align}
 G_2^\s(-\beta^2;x,x') 
=\frac{{\bf S}_{0,\i\beta}(-\cos r)}{4\cosh(\pi\beta)}.
\end{align}
From \eqref{integral-s1} we obtain a family of self-energies depending on
the parameter $\varepsilon:=-2\pi\gamma$:
\begin{align}
 \Sigma_2^{\s,\varepsilon}(-\beta^2) = \frac{1}{4\pi}
 \Big( \psi\big(\tfrac12+\i\beta\big) 
        +\psi\big(\tfrac12-\i\beta\big) -2\varepsilon\Big).
\end{align}
Thus, 
\begin{align}
 G_2^{\s,\varepsilon}(-\beta^2;x,x') 
 =& \frac{{\bf S}_{0,\i\beta}(-\cos r)}{4\cosh(\pi\beta)}
 + \frac{\pi}{4 \cosh^2(\pi\beta)} \;
 \frac{{\bf S}_{0,\i\beta}(-\cos \theta) {\bf S}_{0,\i\beta}(-\cos \theta')}{
  \psi(\tfrac12+\i\beta) 
        +\psi(\tfrac12-\i\beta) -2\varepsilon}.
\end{align}
In contrast to the Euclidean case, $\Sigma_d^{\s,\varepsilon}(-\beta^2)$ 
has a singularity at $-\infty$ but not at $0$.

\paragraph{Dimension  $3$.}  
 We have 
  \begin{align}
 {\bf S}_{\frac12,\i\beta}(-\cos r)
 &=\frac{2\sinh\big((\pi-r)\beta\big)}{\beta\sqrt\pi\sin r}, \\
 G_3^{\s}(-\beta^2,x,x')
 &=\frac{1}{4\pi}
 \frac{\sinh\big((\pi-r)\beta\big)}{\sinh(\pi\beta)\,\sin r}.
\end{align}
Using \eqref{integral-s1} and choosing the integration constant so that $\Sigma_3^\s(-\beta^2) =\Sigma_3(-\beta^2) + o(1)_{\beta\to\infty}$,
\begin{align}
\label{eq:Sigma_3_s}
 \Sigma_{3}^\s(-\beta^2) 
 = \frac{ \beta \coth (\pi \beta) }{ 4\pi } .
\end{align}
Therefore,
 \begin{align}
 \hspace{-2ex}
  G_3^{\s,\gamma}(-\beta^2,x,x')
  =\frac{\sinh\big(\beta(\pi-r)\big)}{4\pi\sinh(\beta\pi)\sin r}+
  \frac{\sinh\big(\beta(\pi-\theta)\big)
    \sinh\big(\beta(\pi-\theta')\big)}{(4\pi)^2 \sin \theta \sin \theta'\sinh^2(\beta\pi)\big(\gamma
  +\frac{\beta \coth(\pi\beta) }{4\pi}\big)}.
 \end{align}

\paragraph{Odd dimensions $d\geq5$.}  
Equation
\eqref{integral-s1} is
still valid if understood in the generalized sense for all
$d\in\cc\setminus2\zz$. Therefore we can set for such $d$ 
\begin{align}\label{integral-s1.}
   \Sigma_d^\s(- \beta^2)
   = & -
   \frac{ \Gamma(\frac{2-d}{2})}
    {(4\pi)^{\frac{d}{2}}}\,
                (\tfrac12+\i\beta)_{\frac{d-2}{2}}(\tfrac12-\i\beta)_{\frac{d-2}{2}}
       .\end{align}
       Due to \eqref{eq:Poch_id2}, 
\eqref{integral-s1.} specified to  odd integer values is
\begin{align}\label{greenodd1}
 \Sigma_{d}^\s(-\beta^2) 
 &= \frac{\pi \beta  \coth(\pi\beta ) }
   { (4\pi)^{\frac{d}{2}}\Gamma (\tfrac{d}{2})  }  
      \prod_{k=1}^{\frac{d-3}{2}}
    \Big(- k^2 -\beta^2 \Big).
\end{align}
Combining \eqref{greenodd}, \eqref{greenodd1} 
and \eqref{eq:Gdgamma-s} we obtain $G_d^{\s,\gamma}(- \beta^2; x,x')$.

\paragraph{Even dimensions $d\geq4$.}  
Similarly as in the flat and hyperbolic case, the formula
\eqref{integral-s1.} is not applicable. As we  will argue below,  for
$d\in2\zz$, $d\geq4$, we will  introduce a family
of reference self-energies parametrized by $\varepsilon\in\rr$:
   \beq\label{integ---s}
 \Sigma_{d}^{\s,\varepsilon} (-\beta^2):=   \frac{1
    }{ (4\pi)^{\tfrac{d}{2}} \Gamma\big(\tfrac{d}{2}\big)}
                       \Big(
  \psi\big(\tfrac{d-1}{2}+\i\beta\big) 
  +\psi\big(\tfrac{d-1}{2}-\i\beta\big) 
  -2\varepsilon \Big) 
\prod_{j=0}^{\tfrac{d-4}{2}} 
 \Big(-\beta^2-\big(\tfrac12 + j\big)^2\Big).
 \eeq 
 Thus we obtain a family of Green's functions
 \begin{align}
 \label{eq:Gdgamma-s.}
  G_d^{\s,\varepsilon,\eta}(z;x,x')&=   G_d^\s(z;x,x')+ 
  \frac{1}{\eta(z)+\Sigma_d^{\s,\varepsilon}(z)} G_d^\s(z;x,x_0)G_d^\s(z;x_0,x'), 
  \end{align}
parametrized by $\varepsilon\in\rr$ and a polynomial $\eta$ with $\deg\eta\leq\frac{d-4}{2}$.

 Let us derive \eqref{integ---s} using 
 the integral in \eqref{integralq} as the starting point.
 For
 $d\in2\zz$, $d\geq4$, unfortunately, \eqref{integralq} has 
 to be understood in the anomalous generalized sense and is not
equal to \eqref{integral-s1}. Instead, it is given by 
\eqref{eq:formo_integer1_limit1.} and we have
\begin{align} \label{eq:sigma_sphere_even_new}
  \sigma_d^\s(-\beta^2)= &
                          \frac{(-1)^{\frac{d-2}{2}}(\frac12+\i\beta)_{\frac{d-2}{2}}
                          (\frac12-\i\beta)_{\frac{d-2}{2}}}
{\Gamma(\frac{d}{2})(4\pi)^{\frac{d}{2}}}
      \Bigg( 
       \frac{\ii }{2 \beta } \Big(
       \psi'\big(\tfrac{d-1}{2} +\ii \beta\big) 
    - \psi'\big(\tfrac{d-1}{2} -\ii \beta\big)   
    \Big) \\\notag
   &\quad-\frac{\ii}{2\beta}\Big(H_{\frac{d-2}{2}}\big(\tfrac12+\ii\beta)-H_{\frac{d-2}{2}}\big(\tfrac12-\ii\beta)\Big)
   \ln 4
    \\   &\quad\notag+
  \sum_{k=0}^{\frac{d-4}{2}} 
 \frac{\psi \big(-\tfrac12-k+\ii \beta \big) 
    +  \psi\big(-\tfrac12-k-\ii \beta \big) 
    -\psi(\frac{d-2}{2}-k)-\psi(1+k) 
 }{\big(\tfrac12 +k\big)^2+\beta^2
  }
\Bigg).
\end{align}
As in the hyperbolic case, application of the 
Leibniz rule and identities satisfied by the Pochhammer symbol, harmonic numbers and the digamma function (see Appendix \ref{poch}) yields
\begin{align}\notag
\hspace{-1ex}
 \sigma_d^\s(z) 
  &= -
 \frac{1}
{\Gamma(\frac{d}{2})(4\pi)^{\frac{d}{2}}}
    \partial_z \Bigg(
\Big(\psi\big(\tfrac{d-1}{2}+\i\sqrt{-z}\big)+\psi\big(\tfrac{d-1}{2}-\i\sqrt{-z}\big)-\ln 4\Big)     \prod_{j=0}^{\frac{d-4}{2}}\big(z-(\tfrac12+j)^2\big)
\Bigg)\\\label{eq:sigma_sphere_even_dim}
  &\quad+\frac{1}
{\Gamma(\frac{d}{2})(4\pi)^{\frac{d}{2}}}\pi_d^\s(z),
\end{align}
where $\pi_d^\s$ is a polynomial of degree $\frac{d-4}{2}$ given by
\begin{align} \label{iuyt}
  \pi_d^\s(z)=&
 \prod_{j=0}^{\frac{d-4}{2}}\big(z-(\tfrac12+j)^2\big)
\Bigg(\sum_{k=0}^{\tfrac{d-4}{2}} 
\frac{\psi\big(\tfrac{d-2}{2}-k\big) +\psi(1+k)}{z-(\tfrac12+k)^2} 
  \\ \notag & \qquad
   - \sum_{k=0}^{\tfrac{d-4}{2}} \sum_{l=k+1}^{\tfrac{d-4}{2}}
   \frac { 2l+1 }{
     \Big(z-\big(\tfrac12+k\big)^2\Big) \Big(z-\big(\tfrac12+l\big)^2\Big)}  \Bigg).
\end{align}  
We define 
\beq\label{polypi2}\Pi_d^\s(z):=-\int_0^{z}\pi_d^\s(\tau)\d\tau,\eeq
 a polynomial of degree $\frac{d-2}{2}$ 
with $\Pi_d^\s(0)=0$ and
\begin{align}\notag
  \Sigma_d^{\s,\mathrm{ms}}(z)&=
 \frac{1}{\Gamma(\frac{d}{2})(4\pi)^{\frac{d}{2}}}
\Big(\psi\big(\tfrac{d-1}{2}+\i\sqrt{-z}\big)+\psi\big(\tfrac{d-1}{2}-\i\sqrt{-z}\big)-\ln 4\Big)     
\prod_{j=0}^{\frac{d-4}{2}}\big(z-(\tfrac12+j)^2\big) 
\\
  &+\frac{1}
{\Gamma(\frac{d}{2})(4\pi)^{\frac{d}{2}}}\Pi_d^\s(z).\label{insero2}\end{align} 
Then $  \Sigma_d^{\s,\mathrm{ms}}$ is an antiderivative of minus
\eqref{eq:sigma_sphere_even_dim}, and will be called the 
{\em reference self-energy based on the minimal
  subtraction}. From $  \Sigma_d^{\s,\mathrm{ms}}$
we pass to the family of reference self-energies
$ \Sigma_{d}^{\s,\varepsilon}$ by absorbing $\Pi_d^\s(z)$ 
into $\varepsilon$ and $\eta(z)$ as in the hyperbolic case.

\subsection{Flat limit of the spherical Laplacian}

Let $R>0$. Instead of the unit sphere we can consider the sphere of
radius $R$
         \[\SS_R^d:=\{x\in\rr^{1+d}\mid(x|x)=R^2\}.\]
Various objects defined using $\SS_R^d$ instead of $\SS^d$ will have
the subscript $R$. We have a bijection $\SS^d=\SS_1^d\ni x\mapsto Rx\in \SS_R^d$. The distance in $\SS_R^d$ satisfies
\beq d^{\s}_R(Rx,Rx')=R d^\s(x,x').\eeq
The Laplace-Beltrami operator on $\SS_R^d$, denoted $\Delta^\s_{d,R}$, is
a self-adjoint operator on $L^2(\SS_R^d)$ and
\begin{align}
\sigma(-\Delta_{d,R}^\s)&=
\Big\{\frac{l(l+d-1)}{R^2} \, \big | \, l=0,1,2,\dots\Big\}.
\end{align}        
We set
$ H_{d,R}^\s:=-\Delta_{d,R}^\s+\frac{(d-1)^2}{4R^2}$. 
For  $\Re\beta>0$ and  $ a<b,$ we set
  \begin{align}
        G_{d,R}^\s(z)&:=
        \big(-z+H_{d,R}^\s\big)^{-1},\\
    \pp_{d,l,R}^\s&:=    
    \one_{\frac{1}{R^2}(l+\frac{d-1}{2})^2}\big(H_{d,R}^\s\big),\\
\pp_{d,R}^\s(a,b)&:
=\one_{[a,b]}\big(H_{d,R}^\s\big).
      \end{align}
    
We have 
     \begin{align}
     G_{d,R}^\s(-\beta^2;x,x') &=  R^{-d+2}   G_{d}^\s\Big(-(\beta R)^2;\frac{x}{R},\frac{x'}{R}\Big),\\
       \pp_{d,l,R}^\s(x,x')&=R^{-d}
                           \pp_{d,l}^\s\Big(\frac{x}{R},\frac{x'}{R}\Big),\\
                  \pp_{d,R}^\s(
           a,b;x,x') &=R^{-d}
    \pp_d^\s\Big(R^2a,R^2 b;\frac{x}{R},\frac{x'}{R}\Big).
    \end{align}

The self-energy on 
$\SS_R^d$ is defined analogously to the hyperbolic case, and comes out to be
\begin{align*}
 \Sigma_{d,R}^\s(-\beta^2) 
  &:= R^{2-d}
    \Sigma_{d}^\s\big(-(\beta R)^2\big), &d\text{ odd};\\
   \Sigma_{d,R}^{\s,\varepsilon}(-\beta^2) 
  &:= R^{2-d}
    \Sigma_{d}^{\s,\varepsilon+\ln R}\big(-(\beta R)^2\big), &d\text{ even}.
\end{align*}

Let $x,x'\in\SS_R^d$. The perturbed Green's functions on $\SS_R^d$
in odd, resp. even dimensions are
\begin{align} \label{eq:Gs_scaled_odd}
  G_{d,R}^{\s,\gamma}(z;x,x')&
 =   
  G_{d,R}^\s(z;x,x') + 
  \frac{ G_{d,R}^\s\big(z;x,R x_0\big) 
  G_{d,R}^\s\big(z;R x_0,x'\big) 
  }{\gamma(z) +\Sigma_{d,R}^\s(z)};
\\ \label{eq:Gs_scaled_even}
 G_{d,R}^{\s,\varepsilon,\eta}(z;x,x')&
=   
  G_{d,R}^{\s}(z;x,x') + 
  \frac{ G_{d,R}^\s\big(z;x,R x_{0}\big)G_{d,R}^\s\big(z;R x_{0},x'\big)  }{\eta(z) +\Sigma_{d,R}^{\s,\varepsilon}(z)}.
\end{align}

Note that $\gamma(z)$ and $\eta(z)$ on the 
right-hand sides of \eqref{eq:Gs_scaled_odd} and 
\eqref{eq:Gs_scaled_even} do not depend on $R$. This choice 
of renormalization is analogous to the hyperbolic case.
Then all Green's functions  have the correct flat limit in the following sense:
\bet 
Let $-\beta^2\in\cc\backslash[0,\infty[$. Then
\begin{align}
\label{eq:limit_Gs_free}
    G_{d,R}^\s\big(-\beta^2,  r\big)&=
    G_d\big(-\beta^2 , r \big)
     \Big(1+\mathcal{O}\big(\tfrac{1}{\beta R}\big)+\mathcal{O}\big(\tfrac{r}{ R}\big)\Big),\\
  \Sigma_{d,R}^\s(-\beta^2)&=\Sigma_{d}\big(-\beta^2\big) 
                             \Big(1+\mathcal{O}\big(\tfrac{1}{\beta R}\big)
                             \Big) ,\quad d\text{ odd};\label{eq:lim}\\
   \Sigma_{d,R}^{\s,\varepsilon}(-\beta^2)&=\Sigma_{d}^\varepsilon\big(-\beta^2\big) 
    \Big(1+\mathcal{O}\big(\tfrac{1}{\beta R}\big) \Big),\label{eq:lim-}\quad 
                                            d\text{ even}.
\end{align}

Thus if we have a family $x_R,x_R'\in\SS_R^d$ and $x,x'\in \rr^d$ such that
\begin{align}\notag&
\lim_{R\to\infty}d_R^\s(x_R,x_R')=|x-x'|,\\ \notag &
\lim_{R\to\infty}d^\s_R(x_R,R x_{0})=|x|, \\
& \lim_{R\to\infty}d^\s(x_R,R x_{0})=|x'|, 
\label{dfdf}\end{align}
then
\begin{align}
  \lim_{R\to\infty} G_{d,R}^{\s,\gamma}\big(-\beta^2; x_R,x'_R\big)&=
G_d^\gamma(-\beta^2,x,x') ,\quad d\text{ odd};                                                                     \label{eq:limit_Gh_pert}\\
    \lim_{R\to\infty} G_{d,R}^{\s,\varepsilon,\eta}\big(-\beta^2; x_R,x'_R\big)&=
G_d^{\varepsilon,\eta}(-\beta^2,x,x'), \quad d\text{ even}.
\label{eq:limit_Gh_pert+_doubled}
\end{align}
\eet
\begin{proof} 
Using the asymptotics of the Gegenbauer 
functions from Thm. \ref{thm:asymptotics}, we find 
\begin{align}  
G_{d,R}^\s(-\beta^2 ;r_R)
  =&R^{-d+2}G_d^\s\Big(-(\beta R)^2;-\cos\frac{r_R}{R}\Big)
  \\\notag
  =&\frac{R^{-d+2}\Gamma(\frac{d-1}{2}+\i\beta
     R)\Gamma(\frac{d-1}{2}-\i\beta
     R)}{(4\pi)^{\frac{d}{2}}}{\bf
     S}_{\frac{d-2}{2},\i\beta R}\Big(-\cos\frac{r_R}{R}\Big)
     \\\notag
=&\frac{\beta^{d-2}\pi\e^{-\pi R\beta}}{2^{\frac{d-2}{2}}(2\pi)^{\frac{d}{2}}}
{\bf
   S}_{\frac{d-2}{2},\i\beta R}\Big(-\cos\frac{r_R}{R}\Big)
   \Big(1+\mathcal{O}\big(\tfrac{1}{\beta R}\big)\Big)
   \\ \notag
   =
&\frac{(\frac{r_R}{R})^{\frac{d-1}{2}}}{(\sin\frac{r_R}{R})^{\frac{d-1}{2}}(2\pi)^{\frac{d}{2}}}\Big(\frac{\beta}{r_R}\Big)^{\frac{d-2}{2}}
       K_{\frac{d-2}{2}}(\beta r_R)  \Big(1+\mathcal{O}\big(\tfrac{1}{\beta R}\big)\Big).
       \end{align}
This proves \eqref{eq:limit_Gs_free}. To prove \eqref{eq:lim}
    we use  Thm. \ref{thm:asymp_genintSZ} and
    \begin{align}
     \psi\big(\tfrac12+\alpha\pm\i\beta R\big) 
      -\ln(\beta R) \mp \i \tfrac{\pi}{2}
 &=\mathcal{O}\big(\tfrac{1}{\beta R}\big).
    \end{align}
Then we argue as in the hyperbolic case.  \qed
\end{proof}

\subsection{Poles of Green's functions and spectral properties} \label{sec:eigs}

All singularities of Green's functions $G_{d,R}^{\s,\gamma}(z)$ are isolated. In dimensions $d=1,2,3$ they correspond to the  point
spectrum of $H_{d,R}^{\s,\gamma}$. In this section we analyze the
location of these singularities. They come in two types: poles of 
$G_{d,R}^{\s}$ and zeros of $\gamma(z) + \Sigma_{d,R}^\s(z)$, resp. 
$\eta(z) + \Sigma_{d,R}^{\s,\varepsilon}(z)$. 
First we discuss the former.

Let $ l \in\nn_0$ and let $\omega_{d,l} =\frac{d-1}{2}+l$ parametrize 
the eigenvalues $R^{-2}\omega_{d,l}^2$ of the unperturbed operator. 
It is well-known that the multiplicity of $R^{-2}\omega_{d,l}^2$,
  or in other words the dimension of the range of $\pp_{d,l,R}^{\s}$,
  is $m_{d,l} = \binom{d+l}{d} - \binom{d+l-2}{d}$.
 Therefore,
  \begin{equation}
    \int \pp_{d,l,R}^\s(x,Rx_0)^2 \dd x =
 \pp^\s_{d,l,R}(Rx_0,Rx_0) = \frac{m_{d,l}}{|\SS^d| R^d}. \label{dime}\end{equation}
The right-hand side of \eqref{dime} can be verified explicitly using 
an appropriately rescaled version of \eqref{jaco1}. Moreover, the rank 
of the residue of the unperturbed Green's operator $G_{d,R}^\s(z)$ at 
$z=R^{-2}\omega_{d,l}^2$ is $m_{d,l}$. Let us show that after 
perturbation this rank drops by $1$.

\bet 
$G_{d,R}^{\s,\gamma}(z)$ has a pole of rank $m_{d,l}-1$ at 
$R^{-2}\omega_{d,l}^2$ . In particular, since $m_{d,0}=1$,  
the perturbed Green's function does not have a pole at 
$R^{-2} \omega_{d,0}^2$.
\eet

\proof
  We have
\begin{equation}
    G^{\s}_{d,R}(z;x,x') = \frac{\pp_{d,l,R}^\s(x,x')}{R^{-2}\omega_{d,l}^2-z} + R(z;x,x'),
\end{equation}
with a remainder $R(z;x,x')$ nonsingular at $z=R^{-2}\omega_{d,l}^2$ and satisfying
\begin{equation}
    \int_{\SS^d_R} \pp^\s_{d,l,R}(x,x') R(z;x',x'') \dd x' =0.
\end{equation}
From this and \eqref{dime} we can deduce that near $z=R^{-2}\omega_{d,l}^2$, function
$ \sigma^\s_{d,R}(z)$ is given by 
\begin{align}
  \sigma^\s_{d,R}(z) & =
\int G_{d,R}^\s(z;0,y)^2\dd y\\
&=                       \frac{1}{(R^{-2}\omega_{d,l}^2-z)^2} \int \pp_{d,l,R}^\s(x,Rx_0)^2 \dd x + \mathcal{O}(1) \\
    & = \frac{1}{(R^{-2}\omega_{d,l}^2-z)^2} \frac{m_{d,l}}{|\SS^d| R^d} + \mathcal{O}(1).
\end{align}
The self-energy thus satisfies
\begin{equation}
    \Sigma^s_{d,R}(z) =-  \frac{1}{R^{-2}\omega_{d,l}^2-z} \frac{m_{d,l}}{|\SS^d| R^d} + \mathcal{O}(1).
\end{equation}
In particular, $\gamma(z)+ \Sigma_{d,R}^{\s}(z) \neq 0$ at 
$z= R^{-2}\omega_{d,l}^2$ due to the singularity of $\Sigma_{d,R}^{\s}(z)$
at this point.
Hence 
\begin{equation}
    G^{\s , \gamma}_{d,R} (z; x,x') = \frac{\pp_{d,l,R}^\s(x,x') - \frac{|\SS^d| R^d}{m_{d,l}} \pp^\s_{d,l,R}(x,Rx_0) \pp^\s_{d,l,R}(Rx_0,x')}{R^{-2}\omega_{d,l}^2-z} +\mathcal{O}(1).
    \label{eq:green_residue}
\end{equation}
Now note that $\sqrt{\frac{|\SS^d| R^d}{m_{d,l}}} \pp^\s_{d,l,R}(\cdot
,Rx_0) = \sqrt{\frac{|\SS^d| R^d}{m_{d,l}}} \pp^\s_{d,l,R}(Rx_0,\cdot)
$ is a real-valued and $L^2$-normalized vector in the range of
$\pp^s_{d,l,R}$. We remark that $\sqrt{\frac{|\SS^d| R^d}{m_{d,l}}} \pp^\s_{d,l,R}(\cdot
,Rx_0)$ may be characterized as the unique normalized vector in the range of $\pp^s_{d,l,R}$ which is invariant under orthogonal transformations preserving $Rx_0$ and is non-negative at $R x_0$. Hence the numerator in \eqref{eq:green_residue} is (the
integral kernel of) the orthogonal projection onto the orthogonal
complement of $\sqrt{\frac{|\SS^d| R^d}{m_{d,l}}} \pp^\s_{d,l,R}(\cdot
,Rx_0)$ in the range of $\pp^\s_{d,l,R}$. In particular it is the
kernel of a projection of rank $m_{d,l}-1$. \qed

The rank of the residue at $R^{-2}\omega_{d,l}^2$ drops by one because
the pole corresponding to one ``eigenvector'' is shifted. We find the
shifted poles by solving the equation
\begin{align} 
\gamma(z)+ \Sigma_{d,R}^{\s}(z) =0 
\quad\text{resp.}\quad
\eta(z) + \Sigma_{d,R}^{\s,\varepsilon}(z)=0.
\end{align}

First let us consider dimension $1$. The unperturbed poles are at $z= R^{-2}l^2$, with multiplicity $1$ for $l=0$ and multiplicity $2$ for $l=1,2,\dots$. The perturbation cancels the pole at $0$ and decreases the multiplicity for $l \geq 1$ to $1$. Putting $\beta R = \i t$ and $\epsilon = \frac{\gamma}{R}$, the equation for shifted poles takes the form
\begin{equation}
    -\frac{1}{2 t} \cot(\pi t) =  \epsilon.
\end{equation}
If $\gamma =0$, then $\epsilon=0$ and the solutions are half-integers. 
Negative solutions correspond to the same $z$, so we focus at positive 
half-integers. If $\gamma \neq 0$, then $\epsilon $ is nonzero but small, 
so we can use the implicit function theorem to find a solution at
\begin{equation}
    t= l + \frac{1}{2} + \frac{2l+1}{\pi} \epsilon + \mathcal{O}(\epsilon^2).
\end{equation}
We remark that if $\gamma \in \rr$, this solution is in $]l,l+1[$ (no matter the size of $\epsilon$).

Unpacking the notation, the above calculation proves:

\begin{theoreme}
    $G_{1,R}^{s, \gamma}(z)$ has poles at $z=R^{-2}l^2$ with $l= 1,2,\dots$ and at 
    \begin{equation}
     E_{1,l,R}^{ \gamma} = \left( \frac{l + \frac12}{R} \right)^2 \left( 1 + \frac{4 \gamma}{\pi R} + \mathcal{O} \left( \frac{\gamma^2}{R^2} \right) \right), \qquad l = 0, 1 , \dots.   
    \end{equation}
All residues are rank one projections.
\end{theoreme}

We see that for $d=1$ and large $R$, the eigenvalues are approximated by these for $\gamma=0$. Starting from dimension $2$, the eigenvalues approach the unperturbed ones (infinite $\gamma$) instead.

\begin{theoreme}
For $d=2,3$, the $l$th ($l=0,1,\dots$) shifted pole of $G_{d,R}^{s,\gamma}$ is at
\begin{itemize}
    \item $d=2$: $E_{2,l,R}^{ \gamma} = \frac{1}{R^2} \left( (l+ \frac12)^2 + \frac{l+ \frac12}{\ln \frac{R}{a}} + \mathcal{O} \left(\frac{1}{\ln^2 \frac{R}{a}} \right) \right)$, 
    where $a = \ee^{2 \pi \gamma}=\ee^{-\varepsilon}$, 
    \item $d=3$: $E_{3,l,R}^{ \gamma} = \frac{(l+1)^2}{R^2} \left( 1 - \frac{1}{2 \pi^2 R \gamma } + \mathcal{O} \left( \frac{1}{R^2 \gamma^2} \right) \right)$, except for the case $\gamma=0$, in which the pole is at $E_{3,l,R}^{ 0} = \frac{(l + \frac12)^2}{R^2}$.
\end{itemize}
The residues are rank one projections.
\end{theoreme}
\begin{proof}
The claim about residues is obvious. We consider the equation for 
the shifted eigenvalue. Throughout the proof we set
$\beta R = \i t$. First consider $d=2$. We have the equation
\begin{equation}
    \psi (\tfrac{1}{2} + t) + \psi(\tfrac{1}{2} - t) = 2 \ln \frac{R}{a}. 
\end{equation}
The right-hand side blows up for $R \to \infty$, so we denote it $\frac{1}{\epsilon}$. We expand the left hand side around the unperturbed pole: writing $t = l+\frac{1}{2}+ \delta$ we obtain
\begin{equation}
    \frac{1}{\delta} + \mathcal{O}(1)_{\delta \to 0} = \frac{1}{\epsilon}.
\end{equation}
By the implicit function theorem there exists a solution $\delta = \epsilon + \mathcal{O}(\epsilon^2)$. 

Next consider $d=3$. We have the equation
\begin{equation}
    - t \cot(\pi t) = 4 \pi R \gamma.
\end{equation}
We denote the right hand side by $\frac{1}{\epsilon}$ and expand the left hand side around $l+1$, finding a solution of the form 
\begin{equation}
    t = (l+1) \left( 1 - \frac{\epsilon}{\pi} + \mathcal{O}(\epsilon^2) \right).
\end{equation}
Separate analysis of $\gamma =0$ (hence $\epsilon = \infty$) is elementary. 
\end{proof}

Due to the presence of the polynomials $\gamma(z)$ resp. $\eta(z)$, 
the situation is  more complicated in higher dimensions,
especially in even dimensions $d\geq4$. We distinguish three cases:
\begin{enumerate}
 \item $\gamma(z)\equiv0$ resp. $\eta(z)\equiv0$, 
where the poles correspond to the zeros of the reference self-energies. 
We might call this the \emph{unitary gas case}.
\item  $\gamma(z)$ resp. $\eta(z)$ is a non-constant polynomial. The shifted poles of the perturbed Green's functions are located near the poles of the unperturbed Green's function. As we will see, the rate of convergence of perturbed poles to unperturbed ones as $R \to \infty$ is modified if $\gamma(z)$ resp. $\eta(z)$ vanishes at zero, and depends on the degree of vanishing.
\item $\gamma(z)=\gamma_0$ resp. $\eta(z)=\eta_0$ are nonzero constants. 
This could be treated on the same footing as case 2. with $\gamma(z)$
resp. $\eta(z)$ not vanishing at zero, but since the conclusions are particularly simple we prefer to state them separately.
\end{enumerate}

\paragraph{Odd dimensions $d\geq5$.}
Let us first find the zeros and poles of the reference self energy:
\begin{lemma}
 \label{lem:poles_SigRef_odd} 
 Let $d\geq5$ be an odd integer. The zeros and poles 
 of the reference self-energy are located at $\zeta=\ii\beta\geq0$ such that 
 \begin{align}
    \Sigma_{d,R}^\s(\zeta^2) &=0 &&\quad\Leftrightarrow\quad 
  \zeta= \frac{l+\tfrac12}{R},\quad l\in\nn_0, \\ \notag 
  \Sigma_{d,R}^\s(\zeta^2) &\quad\text{has a pole} &&\quad\Leftrightarrow\quad 
  \zeta= \frac{k}{R},\quad k= l+\frac{d-1}{2},\quad l\in\nn_0.
 \end{align}
\end{lemma}
\begin{proof}
 The reference self-energy is
 \begin{align}
  \Sigma_{d,R}^\s(-\beta^2) 
    = \frac{\pi  \coth(\pi\beta R ) \beta}
   { (4\pi)^{\frac{d}{2}}\Gamma (\tfrac{d}{2})  }  
      \prod_{k=1}^{\frac{d-3}{2}}
    \Big(- \frac{k^2}{R^2} -\beta^2 \Big),\quad l\in\nn_0.
 \end{align}
 Writing $\beta=\ii \zeta$ with $\zeta\geq0$,  we find  
  \begin{align}
  \Sigma_{d,R}^\s(\zeta^2) 
    = \frac{\pi  } { (4\pi)^{\frac{d}{2}}\Gamma (\tfrac{d}{2})  }  
    \cos(\pi\zeta R)
      \frac{\zeta \prod_{k=1}^{\frac{d-3}{2}}
    \Big( \zeta^2- \frac{k^2}{R^2} \Big)}{\sin(\pi\zeta R)}.
 \end{align}
 The zeros of $\Sigma_{d,R}^\s$ are located at the zeros of the 
 cosine. The poles are located at the zeros of the sine, except 
 for the first few, which are canceled by the zeros of the numerator. 
 \qed
\end{proof}

Note that the location of the poles of $ \Sigma_{d,R}^\s$ 
 precisely corresponds to the unperturbed eigenvalues.

\begin{theoreme}
 \label{eq:poles_refSigma_odd}
 Let $d\geq5$ be an odd integer.
 \begin{enumerate}
  \item Suppose that $\gamma(z)\equiv0$ vanishes identically, then 
  $G_{d,R}^{s, \gamma}$ has a sequence of isolated poles located at 
 \begin{align}
   E_{d,l,R}^{0} = \frac{\big(l+\tfrac12\big)^2}{R^2},\quad l\in\nn_0.
 \end{align}
 \item Suppose that $\gamma(z)$ does not vanish identically and let $\nu$ be the order of vanishing of $\gamma(z)$ at $0$ ($\nu=0$ if $\gamma(0) \neq 0$). The $l$th shifted pole of $G_{d,R}^{s, \gamma}$ is located at
\begin{equation}
  E_{d,l,R}^{ \gamma} = \frac{\omega_{d,l}^2}{R^2} \left( 1 -  \frac{2 \prod_{k=1}^{\frac{d-3}{2}} (\omega_{d,l}^2 - k^2)}{ (4 \pi)^{\frac{d}{2}}\Gamma(\tfrac{d}{2}) R^{d-2} \gamma \left( R^{-2}\omega_{d,l}^2 \right)} + \mathcal{O}(R^{-2d+4+4 \nu}) \right).
\end{equation}
 \item In the special case where $\gamma(z)=\gamma_0\in\rr\setminus\{0\}$ is  
 constant, the $l$th shifted pole of $G_{d,R}^{s, \gamma}$ 
 is located at 
 \begin{align}
  E_{d,l,R}^{ \gamma} = \frac{\omega_{d,l}^2}{R^2} \left( 1 -  \frac{2 \prod_{k=1}^{\frac{d-3}{2}} (\omega_{d,l}^2 - k^2)}{ (4 \pi)^{\frac{d}{2}}\Gamma(\tfrac{d}{2}) R^{d-2} \gamma_0} + \mathcal{O}\big(R^{-2(d-2)}\big) \right).
 \end{align}
 In particular, the first correction to the 
 unperturbed eigenvalues is inversely proportional to the volume 
 of the sphere. 
 \end{enumerate}
\end{theoreme}
\begin{proof}
 The first case $\gamma(z)\equiv0$ follows directly from Lemma 
 \ref{lem:poles_SigRef_odd}.
 The third case is a special case of the second, and the latter 
 can be derived analogously to lower dimensions. The only 
 complication is that $\gamma(z)$ may vanish at zero, in which case the scaling of the pole shift with $R$ is modified. This is taken into 
 account by the introduction of $\nu$.  \qed
\end{proof}

We note that $\nu$ can be as large as $\tfrac{d-3}{2}$, in which 
case the shift of the eigenvalue is proportional to $R^{-3}$ and the 
first neglected term is proportional to $R^{-4}$. That is, the scaling 
of the unperturbed eigenvalue and the scaling of the shift with $R$ 
differ only by a single power.

\paragraph{Even dimensions $d\geq4$.}
In even dimensions, we considered a family of reference self-energies 
parametrized by $\varepsilon\in\rr$. We look for the zeros of the 
reference self-energies first. 
\begin{lemma}
 \label{lem:poles_refSigma_even} Let $d\geq4$ be an even integer. For 
 large $R$, the 
 zeros of the family of reference self-energies are located 
 at $\zeta=\ii\beta\geq0$ such that 
 \begin{align}
  \zeta^2=\begin{cases}
         \dfrac{1}{R^2}\big(\frac{1}{2}+ j\big)^2, &\quad j=0,\dots, \frac{d-4}{2},
         \vspace{12pt} \\ 
        \dfrac{1}{R^2}
        \Big( \omega_{d,l}^2 + \dfrac{\omega_{d,l}}{\ln(\ee^\varepsilon R)} + \mathcal{O}\Big(\dfrac{1}{\ln^2(\ee^\varepsilon R)}\Big) \Big), &\quad l\in\nn_0.
        \end{cases}.
 \end{align}
\end{lemma}
\begin{proof}
Let $\beta R=\ii t$. 
The reference self-energies are 
\begin{align*}
\Sigma_{d,R}^{\s,\varepsilon} \big(\tfrac{t^2}{R^2}\big)
 &=   \frac{1
    }{ (4\pi)^{\tfrac{d}{2}} \Gamma\big(\tfrac{d}{2}\big)}
                       \Big(
  \psi\big(\tfrac{d-1}{2}+t\big) 
  +\psi\big(\tfrac{d-1}{2}-t \big) 
  -2\varepsilon - 2\ln R\Big) \\ \notag &\quad \times
   R^{2-d}\prod_{j=0}^{\tfrac{d-4}{2}} 
 \Big(t^2-\big(\tfrac12 + j\big)^2\Big).
\end{align*}
The zeros of the second line are obvious. If $t=0$, the 
whole expression is $\sim R^{2-d} \ln R$, which neither 
corresponds to a pole nor a zero (but it is an approximate zero 
for large $R$).

We look for the zeros of the first line, which correspond to 
\begin{align}
 \psi\big(\tfrac{d-1}{2}+t\big) 
  +\psi\big(\tfrac{d-1}{2}-t \big) 
  =2\varepsilon + 2\ln R.
\end{align}
The right-hand side is large for large $R$, so as for $d=2$, we denote 
it $\frac{1}{\epsilon}$ (note the difference between $\epsilon$ and 
$\varepsilon$). We perturb the left-hand side around the 
unperturbed poles by setting $t = \frac{d-1}{2} + l + \delta$ for 
$l\in\nn_0$. We obtain 
\begin{align}
 \frac{1}{\delta} + \mathcal{O}(1)_{\delta\to0} = \frac{1}{\epsilon},
\end{align}
so by the implicit function theorem there exists a solution 
$\delta= \epsilon + \mathcal{O}(\epsilon^2)$.
 \qed
\end{proof}

Lemma \ref{lem:poles_refSigma_even} allows us to describe the 
poles of the perturbed Green's functions. 
\begin{theoreme}
 \label{thm:poles_even_d}
 Let $d\geq4$ be an even integer.
 \begin{enumerate}
  \item Suppose that $\eta(z)\equiv0$ vanishes identically, then 
  $G_{d,R}^{s, \varepsilon,0}$ has a sequence of isolated poles located at 
 \begin{align}
   E_{d,l,R}^{\varepsilon,0} 
   = \dfrac{1}{R^2}
        \Big( \omega_{d,l}^2 + \dfrac{\omega_{d,l}}{\ln(\ee^\epsilon R)} + \mathcal{O}\Big(\dfrac{1}{\ln^2(\ee^\epsilon R)}\Big) \Big), &\quad l\in\nn_0,
 \end{align}
 and a finite number of additional poles at 
 \begin{align}
  E_{d,j,R}^{\varepsilon,0,{\rm exceptional}} 
  = \frac{1}{R^2}\Big(j+\frac12\Big)^2,\quad j=0,\dots,\tfrac{d-4}{2}.
 \end{align}

 \item
 Suppose that $\eta(z)$ does not vanish identically and let $\nu$ be the order of vanishing of $\eta(z)$ at $0$ ($\nu=0$ if $\eta(0) \neq 0$). The $l$th shifted pole of $G_{d,R}^{s, \gamma}$ is located at
\begin{align}
\label{eq:EV_even_general}
  E_{d,l,R}^{ \varepsilon,\eta} &= 
  \frac{1}{R^2}
  \left( \omega_{d,l}^2 
   -\frac{2\omega_{d,l}\prod_{j=0}^{\tfrac{d-4}{2}} 
 \Big(\omega_{d,l}^2-\big(\tfrac12 + j\big)^2\Big)
 }{(4\pi)^{\tfrac{d}{2}} \Gamma\big(\tfrac{d}{2}\big) R^{d-2} \eta\big(\tfrac{\omega_{d,l}^2}{R^2}\big)}
 + \mathcal{O}\Big( \ln(\ee^{\varepsilon} R) R^{-2d+4+4\nu}\Big)\right).
\end{align}
 \item In the special case where $\eta(z)=\eta_0\in\rr\setminus\{0\}$ is  
 constant, we have 
 \begin{align}
  E_{d,l,R}^{ \varepsilon,\eta_0} = 
  \frac{1}{R^2}
  \left( \omega_{d,l}^2 
   -\frac{2\omega_{d,l}\prod_{j=0}^{\tfrac{d-4}{2}} 
 \Big(\omega_{d,l}^2-\big(\tfrac12 + j\big)^2\Big)
 }{(4\pi)^{\tfrac{d}{2}} \Gamma\big(\tfrac{d}{2}\big) R^{d-2} \eta_0}
 + \mathcal{O}\Big( \ln(\ee^{\varepsilon} R) R^{-2(d-2)}\Big)\right).
 \end{align}
 In particular, the first correction to the 
 unperturbed eigenvalues is inversely proportional to the volume 
 of the sphere. 
 \end{enumerate}
\end{theoreme}
\begin{proof}
 The first statement follows from Lemma \ref{lem:poles_refSigma_even}. 
 To show the second statement, we need to consider the equation 
 \begin{align}
 \label{eq:pole_eq_even}
   &\quad \Big(
  \psi\big(\tfrac{d-1}{2}+t\big) 
  +\psi\big(\tfrac{d-1}{2}-t \big) 
  \Big) 
   \prod_{j=0}^{\tfrac{d-4}{2}} 
 \Big(t^2-\big(\tfrac12 + j\big)^2\Big)
 \\ \notag
 &= - (4\pi)^{\tfrac{d}{2}} \Gamma\big(\tfrac{d}{2}\big) R^{d-2} \eta\big(\tfrac{t^2}{R^2}\big)
 +2 \ln(\ee^{\varepsilon} R) \prod_{j=0}^{\tfrac{d-4}{2}} 
 \Big(t^2-\big(\tfrac12 + j\big)^2\Big).
 \end{align}
 If $t=j+\tfrac12$ for some $j=0,\dots,\tfrac{d-4}{2}$, then 
 \eqref{eq:pole_eq_even} becomes 
 \begin{align}
 0= 
 R^{d-2} \eta\Big( R^{-2} \big(j+\tfrac12\big)^2\Big).
 \label{eq:equation_label_679138}
 \end{align}
The right hand side of \eqref{eq:equation_label_679138} is a polynomial 
in $R$. Therefore, equation \eqref{eq:equation_label_679138} is not
satisfied if $R$ is large enough. Hence, 
 $t=j+\tfrac12$ with $j=0,\dots,\tfrac{d-4}{2}$ does not 
 correspond to a pole of Green's function. We may rewrite 
  \eqref{eq:pole_eq_even} as 
  \begin{align}
  \psi\big(\tfrac{d-1}{2}+t\big) 
  +\psi\big(\tfrac{d-1}{2}-t \big) 
 &= - \frac{(4\pi)^{\tfrac{d}{2}} \Gamma\big(\tfrac{d}{2}\big) R^{d-2} \eta\big(\tfrac{t^2}{R^2}\big)}{\prod_{j=0}^{\tfrac{d-4}{2}} 
 \Big(t^2-\big(\tfrac12 + j\big)^2\Big)}
 +2 \ln(\ee^{\varepsilon} R) .
  \end{align} 
 The first term on the right-hand side is of order $R^{d-2-2\nu}$
 and the second term is $\sim\ln R$, so both blow up for $R\to\infty$. 
 Now as before, we denote the right-hand side by $\tfrac{1}{\epsilon}$ 
 and write $t=\omega_{d,l}+\delta$. This gives 
 \eqref{eq:EV_even_general}. The third claim is a special case of the 
 second. 
 \qed
\end{proof}

We remark that for simplicity we did not indicate the dependence of error terms in the in higher dimensions on $\gamma$ resp. $\eta$. Moreover, 
in all results of this section the error bounds are not uniform in $l$.

A more precise analysis of the poles of the perturbed Green's function 
--- including a detailed analysis of the dependence 
of the error terms on the polynomials $\gamma$ resp. $\eta$ and 
estimates that are uniform in $l$ -- is desirable but beyond the 
scope of the current paper.

\appendix

\section{Generalized integrals}
\label{app:genInts}
Generalized integrals go back to ideas of Hadamard
\cite{Hadamard23,Hadamard32} and Riesz \cite{Riesz}. 
In a parallel work \cite{DGR23a}, we revisited this concept in a 
manner that is well-suited for our applications. In the latter reference, 
the proofs for all generalized integrals appearing in Appendices 
\ref{app:bessel} and \ref{app:gegenbauer} are displayed in detail. 

\begin{definition}
\label{def:genInt}
Let $a\in\rr$. We say that a function 
$f$ on $]a,\infty[$ is {\em integrable in the generalized sense} if it is 
integrable on $]a+1,\infty[$ and if there exists a finite set $\Omega\subset\cc$ 
and complex coefficients $(f_k)_{k \in \Omega}$ such that 
\begin{align} 
f-\sum_{k\in\Omega} f_k(r-a)^k
\end{align} 
is integrable on $]a,a+1[$.  We define
\begin{align}
\label{gener}
&\ge\int_a^\infty f(r)\dd r
 := \sum_{k\in\Omega\backslash\{-1\}}\frac{f_k}{k+1}+
\int_a^{a+1}\Big(f(r)-\sum_{k\in\Omega}f_k (r-a)^k\Big)\dd r+\int_{a+1}^\infty
f(r)\dd r.\end{align}
\end{definition}

Note that the set 
$\{k \in \Omega \, | \, \Re k\leq -1\}$ and the corresponding $f_k$ are uniquely determined by $f$. It is convenient to allow $k \in \Omega$ with 
 $\Re k>-1$. The generalized integral of $f$ does not depend on the choice of $\Omega$. 

Clearly 
\begin{align}
\ge\int_a^\infty f(r)\dd r=\int_a^\infty f(r)\dd  r \quad \text{for } f \in L^1[a, \infty[. \end{align}
If $\Phi$ is any other extension of the integration functional from $L^1[a, \infty[$ to the class of all functions integrable in the generalized sense, then $\Phi$ is given by
\begin{equation}
    \Phi(f) = \gen \int_a^\infty f(r) \dd r + \sum_{\substack{k \in \Omega \\  \Re k \leq -1}}  f_{k} \lambda_k
    \label{eq:gen_int_prime}
\end{equation}
for some coefficients $\lambda_k$. Conversely, for any set of $\lambda_k$ one may define an extension $\Phi$ by \eqref{eq:gen_int_prime}. To some extent the definition of $\gen \int_a^\infty$ is arbitrary and one could use some other extension instead. $\gen \int_a^\infty$ has several simple properties which make it a useful reference point. 

The generalized integral is invariant with respect to translations 
and taking power of the integration variable, 
\begin{align}
\label{eq:genInt_transl}
\gen\int_a^\infty f(r) \dd r
&= \gen\int_{a-\alpha}^{\infty} f(u+\alpha) \dd u, &&\alpha\in\rr, 
\\    \label{eq:genInt_powers} 
\gen\int_0^\infty f(r)\dd r &=\gen\int_{0}^\infty f( u^\alpha)\, \alpha u^{\alpha-1} \dd u, &&\alpha>0.
\end{align}
Due to the first property there is no loss in assuming $a=0$.

Generalized integral behaves in an interesting way under coordinate transformations. Let $g : [0, \infty[ \to [0 , \infty[$ be a bijection, smooth down to $0$, such that $g(0)=0$ and $g'(0) \neq 0$. The map $f \mapsto (f \circ g)g'$ preserves the class of functions integrable in the generalized sense and (by~the change of variables formula) the classical integration functional. Hence one may define generalized integration in the changed coordinate system as
\begin{equation}
    \gen_g \int_0^\infty f(r) \dd r = \gen \int_0^\infty f(g(u)) g'(u) \dd u.
\end{equation}
The corresponding coefficients $\lambda_k$ in the comparison formula \eqref{eq:gen_int_prime} have been calculated in \cite{DGR23a}:
\begin{align} \label{eq:change_of_var}
    \gen_g \int_0^\infty f(r) \dd r = & \gen \int_0^\infty f(r) \dd r 
    + f_{-1} \ln \frac{1}{g'(0)} \\
    & + \sum_{\substack{l=2 \\ -l\in\Omega}}^\infty f_{-l} \frac{1}{(l-1)(l-1)!} \frac{\dd^{l-1}}{\dd u^{l-1}} \left. \Big( \frac{u}{g(u)} \Big)^{l-1} \right|_{u=0}. \nonumber
\end{align}
This involves only coefficients $f_{-1}, f_{-2}, \dots$. 
Other $f_k$ appear if one considers more general coordinate transformations. 
For example, if $g^{(n)}(0)=0$ for $n=0,\dots, N-1$ but $g^{(N)}(0)\neq0$, 
then also non-zero coefficients $f_{-\frac{k}{N}}$, $k\in\nn$, will 
cause anomalous behavior. 
Note that the sum on the right-hand side of \eqref{eq:change_of_var} 
is finite and that the number of appearing derivatives of $g$ is governed 
by the scaling behavior of the integrand.

A particularly important change of variables is the scaling:
\begin{align}
    \gen \int_0^\infty f(\alpha u) \alpha \dd  u = \gen \int_0^\infty f(r) \dd r - f_{-1} \ln \alpha,\quad\alpha>0
    . \label{eq:genInt_scaling}
\end{align}
One should carefully distinguish this integral from integration with respect to coordinate $\alpha u$ (which amounts to relabeling $u$ to $\alpha u$):
\begin{equation}
    \gen \int_0^\infty f (\alpha u ) d \alpha u = \gen \int_0^\infty f (u ) d u .
\end{equation}
Combining with \eqref{eq:genInt_scaling} one obtains the formula
\begin{equation}
    \gen \int_0^\infty f(u) d \alpha u = \alpha \left( \gen \int_0^\infty f(u) d  u + f_{-1} \ln\alpha \right).
    \label{eq:genInt_scaling2}
\end{equation}

As seen from \eqref{eq:genInt_scaling} and \eqref{eq:genInt_scaling2}, the generalized integral is only scale invariant on the class of function with $f_{-1}=0$. If $f_{-1} \neq 0$, we say that the integral has a \emph{scaling anomaly}. On the grounds of \eqref{eq:change_of_var}, integrals with $f_k \neq 0$ for any $k=-1,-2,\dots$ were called \emph{anomalous} in \cite{DGR23a}.

In quantum field theory jargon, generalized integrals with a scaling anomaly depend on the choice of a renormalization scale. In Definition \ref{def:genInt} we set this scale for $1$ for mathematical convenience.

The coefficient $f_{-1}$ plays a special role also in computations of generalized integrals by analytic continuation, as in the method of dimensional regularization. More details on this and other properties of generalized integrals can be found in 
the parallel work \cite{DGR23a} and the aforementioned literature \cite{Hadamard23,Hadamard32,Hoermander90,Lesch97,Paycha,Riesz}. We remark that dimensional regularization was used in \cite{DGR23a} to compute generalized integrals in Appendices \ref{app:bessel} and \ref{app:gegenbauer}.

\section{The Bessel equation}
\label{app:bessel} 

The \emph{modified Bessel equation},
\begin{align}
\left(\partial_r^2+\frac{1}{r}\partial_r-\frac{\alpha ^2}{r^2}-1\right)v(r)& = 0,
\label{lap5}\end{align}
has two kinds of standard solutions:
the {\em modified Bessel function}, which can be defined by the power series
\begin{align}
\label{eq:BesselI_series}
I_\alpha (r)=
 \sum_{n=0}^\infty\frac{\left(\frac{r}{2}\right)^{2n+\alpha }}
 {n!\Gamma(\alpha +n+1)},
\end{align}
and at zero behaves as
$\sim\frac{1}{\Gamma(\alpha+1)}\big(\frac{r}{2}\big)^\alpha$, and the Macdonald function, which for $\Re (r)>0$ and all $\alpha $ can be defined by the absolutely convergent integral
 \begin{align}
K_{-\alpha}(r)=K_\alpha (r)&:=\frac12\int_0^\infty\exp
\left(-\frac{r}{2}(s+s^{-1})\right)s^{ \alpha -1}\d  s.\label{basset}
 \end{align}
The Macdonald function can be characterized by its asymptotics at infinity:
 for $|\arg r|<\pi-\epsilon$, $\epsilon>0$,
\begin{align}\lim_{|r|\to\infty}\frac{K_\alpha (r)}
{\frac{\e^{-r}\sqrt{\pi}}{\sqrt{2 r}}}=1.\label{asim}\end{align}
Note the connection formula
\begin{align}
K_\alpha (r)&=\frac{\pi}{2\sin\pi \alpha }(I_{-\alpha }(r)-I_\alpha (r)),
\label{macdo1}
\end{align}
and an asymptotic expansion  for $r\to\infty$ in the sector $|\arg r|<\pi-\epsilon$, $\epsilon>0$ \cite{NIST}:
\beq
K_\alpha (r)\simeq  \sqrt\pi\e^{-r}\sum_{n=0}^\infty
\frac{(\frac12+\alpha -n)_{2n}}{n!(2r)^{n+\frac12}}.\label{asy}\eeq
In the half-integer case we have an expression in terms of elementary
functions \cite{NIST}:
\begin{align} 
    \label{eq:MacDonald_halfinteger}
  K_{\pm(\frac12+k)}(r)&=\Big(\frac\pi2\Big)^{\frac12}r^{\frac12+k}\Big(-\frac1r\partial_r\Big)^k\frac{\e^{- r}}{r}.
\end{align}

We will need the following bilinear integral identities for 
 $\Re b>0$ \cite{GR}:
\begin{align}
\label{int3} \int_0^\infty  K_\alpha (br)^22r\d  r&=\frac{\pi \alpha }
{b^2\sin (\pi \alpha ) },  
\quad \ \alpha \neq0, \ |\Re (\alpha )|<1,\\
\label{int4} \int_0^\infty  K_0(br)^22r\d 
  r&=\frac{1}{b^2}.\end{align}

If we replace the conditions  $|\Re \alpha |<1$,
$\alpha \neq0$ with
$\alpha \not\in\zz$ and in the integrals replace  $\int$ with
$\gen\int$,
then \eqref{int3} remains true. 
\eqref{int3}  can also be generalized to $\alpha\in
\mathbb{Z}$ using anomalous generalized integrals \cite{DGR23a}:   
\begin{align}
\label{int4a.} \gen\int_0^\infty  K_\alpha (br)^22r\d 
r&=\frac{(-1)^\alpha }{b^2} \Big(1+  |\alpha
     |\ln\big(\tfrac{b^2}{4}\big)
     +2|\alpha|\big(1-\psi(1+|\alpha|)\big)\Big).
\end{align}

The \emph{(standard) Bessel equation} is obtained by setting
$r\to\pm\ii r$ in the modified one:
\begin{align}
\left(\partial_r^2+\frac{1}{r}\partial_r-\frac{\alpha ^2}{r^2}+1\right)v(r)& = 0.
\label{lap6}
\end{align}

We have several kinds of standard solutions of (\ref{lap6}). The most important is 
the {\em Bessel function}, defined as
\begin{align}
J_\alpha (r)&=\e^{\pm\ii\pi\frac{ \alpha }{2}}I_\alpha (\mp\ii r).
\end{align}
The two {\em Hankel functions} also solve (\ref{lap6}):
\begin{align}
H_\alpha ^{\pm}(r)&=\frac{2}{\pi}\e^{\mp\ii\frac\pi2(\alpha +1)}
K_\alpha (\mp
\ii r).
\end{align}

\begin{remark}
In the literature the usual  notation for Hankel functions is
\begin{align}
H_\alpha ^{(1)}(r)=H_\alpha ^+(r),\quad
H_\alpha ^{(2)}(r)=H_\alpha ^-(r).\end{align}
\end{remark}

\section{The Gegenbauer equation}
\label{app:gegenbauer}
The Gegenbauer equation is  the special case of the hypergeometric 
equation with the symmetry $w\to-w$ and the singular points put at 
$-1,1,\infty$:
\begin{align}
\Big((1-w^2)\partial_w^2-2(1+\alpha )w\partial_w
+\lambda ^2-\big(\alpha +\tfrac{1}{2}\big)^2\Big)f(w)=0.\label{gege0}
\end{align}
The Gegenbauer equation is closely related to the associated 
Legendre equation, see for example \cite{NIST,GR,WW}. Moreover, there exist various 
conventions for the parameters of the Gegenbauer equation 
(cf. \cite{Durand76,Durand19b}, \cite{De1,De2}). The convention as in 
\eqref{gege0} is the most convenient for our purposes.

One of its standard solutions is the function
the function $\mathbf{S}_{\alpha,\beta}$ characterized by asymptotics $\sim \frac{1}{\Gamma(\alpha+1)}$ at $1$: 
\begin{align}\label{solu1}
{\bf    S}_{\alpha ,\pm \lambda }(w)
    :=&\,{_2{\bf F}_1}\big(\tfrac12+\alpha+\lambda,\tfrac12+\alpha-\lambda;1+\alpha;\tfrac{1-w}{2}\big)
    \\
    =& \left( \tfrac{2}{w+1} \right)^\alpha
    {_2{\bf F}_1}\big(\tfrac12+\lambda,\tfrac12-\lambda;1+\alpha;\tfrac{1-w}{2}\big)
 \label{solu1_form2}
 \\
 =& \left( \tfrac{2}{w+1} \right)^{\tfrac12+\alpha\pm\lambda}
 {_2{\bf F}_1}\big(\tfrac12+\alpha\pm\lambda,\tfrac12\pm\lambda;1+\alpha;\tfrac{w-1}{w+1}\big),
 \label{solu1_form3}
\end{align}
where ${_2{\bf F}_1}(a,b;c;z) := \sum_{j=0}^\infty \frac{(a)_j (b)_j}{\Gamma(c+j)j!}
z^j$ is the Gauß hypergeometric function in Olver's normalization. 
There is also the solution characterized by asymptotics $\sim \frac{1}{w^{\frac12+\alpha+\lambda}\Gamma(\lambda+1)}$ at $\infty$:
\begin{align}\label{solu3}
 \mathbf{Z}_{\alpha ,\lambda }(w)
 &:=
 \frac{\Gamma(1+2\lambda) }{\Gamma(1+\lambda) ( w \pm 1)^{\frac12+\alpha +\lambda }}
 {_2{\bf F}_1}\big(\tfrac12+\lambda,\tfrac12+\lambda+\alpha;1+2\lambda;\tfrac{2}{1\pm w}\big)
\\ \label{solu3_form2}
&=
 \frac{\Gamma(1+2\lambda) }{\Gamma(\lambda+1) ( w \pm 1)^{\frac12 +\lambda }
 ( w \mp 1)^{\alpha}}
 {_2{\bf F}_1}\big(\tfrac12+\lambda,\tfrac12+\lambda-\alpha;1+2\lambda;\tfrac{2}{1\pm w}\big).
\end{align}
The equality of the series representations \eqref{solu1}, \eqref{solu1_form2}, 
\eqref{solu1_form3} respectively \eqref{solu3},\eqref{solu3_form2} follows 
from Kummer's table of hypergeometric functions (see e.g. \cite{NIST}). 
In fact, only \eqref{solu1}, \eqref{solu1_form2} and the expressions 
with the $+$-sign in  \eqref{solu3}, \eqref{solu3_form2} are convergent 
in the whole region of physical interest. 

It is convenient to introduce the notation
\begin{align}
(w^2-1)^\alpha_\bullet := (w-1)^\alpha(w+1)^\alpha,
\end{align} 
where $(w\mp1)^\alpha$ are the usual principal branches with the domains
$\mathbb{C}\setminus]-\infty,\pm 1]$.
We note the identities
\begin{align}
\label{eq:identities_SZ_signs}
  {\bf S}_{\alpha,\lambda}(w)={\bf S}_{\alpha,-\lambda}(w),&\quad
     {\bf Z}_{\alpha,\lambda}(w)=\frac{ {\bf Z}_{-\alpha ,\lambda }(w)}{(w^2-1)^\alpha_\bullet}
  \end{align}
as well as the slightly more subtle {\em Whipple transformations}:
  \begin{align}
 {\bf Z}_{\alpha ,\lambda }(w)&:=  (w^2-1)^{-\frac14-\frac\alpha2 
                                -\frac\lambda2}_\bullet{\bf S}_{\lambda 
                                ,\alpha 
                                }\left(\frac{w}{(w^2-1)^{\frac12}_\bullet}\right), \label{eq:Whipple1} \\
 {\bf S}_{\alpha ,\lambda }(w)&:=  (w^2-1)^{-\frac14-\frac\alpha2 
                                -\frac\lambda2}_\bullet{\bf Z}_{\lambda 
                                ,\alpha 
                                }\left(\frac{w}{(w^2-1)^{\frac12}_\bullet}\right), \qquad \Re w >0. \label{eq:Whipple2}
  \end{align}
  For $n=0,1,\dots$ we define the {\em Gegenbauer polynomials}:
  \beq\label{polynomial}
  C_n^{\alpha+\frac12}(w):=\frac{\Gamma(\alpha+1)(2\alpha+1)_n}{n!}{\bf S}_{\alpha,\frac12+\alpha+n}(w).\eeq
Here are the connection formulas:
\begin{align}\label{formu2}
  {\bf S}_{\alpha,\lambda}(-w) 
&=-\frac{\cos(\pi\lambda)}{\sin(\pi\alpha)}
{\bf S}_{\alpha,\lambda}(w)
+\frac{2^{2\alpha}\pi}{\sin(\pi\alpha) \Gamma(\frac12+\alpha+\lambda)
\Gamma(\frac12+\alpha-\lambda)}
\frac {{\bf S}_{-\alpha,-\lambda}(w)}{(1-w^2)^{\alpha}},\\\label{formu1}
{\bf Z}_{\alpha,\lambda}(w) 
&=-\frac{2^{\lambda-\alpha-\frac12}\sqrt{\pi}}{\sin(\pi\alpha)\Gamma(\frac12-\alpha+\lambda)}
{\bf S}_{\alpha,\lambda}(w) 
+\frac{2^{\lambda+\alpha-\frac12}\sqrt{\pi}}{\sin(\pi\alpha)\Gamma(\frac12+\alpha+\lambda)}
\frac {{\bf S}_{-\alpha,-\lambda}(w)}{(w^2-1)^{\alpha}_\bullet}.\end{align}

For  $\Re \alpha >-1$, we can compute the following generalized 
bilinear integrals of $\mathbf{S}$ functions. 
For $\alpha\not\in\zz$ they are non-anomalous \cite{DGR23a}:
\begin{align}\label{into1}
&\ge\int_{-2}^2{\bf S}_{\alpha,\ii\beta}(w)^2(1-w^2)^\alpha\dd 2w
=
    \frac{ 2^{2\alpha+1} \ii \cosh(\pi\beta) }
    {\beta\sin\pi\alpha \,\Gamma(\tfrac12+\alpha-\ii\beta)
    \Gamma(\tfrac12+\alpha+\ii\beta)}
\\ \notag
    &\hspace{24ex}\times \Big( \psi\big(\tfrac12 +\alpha + \ii \beta \big) 
         - \psi\big(\tfrac12 +\alpha - \ii \beta \big) 
         + \psi\big(\tfrac12 - \ii \beta \big) 
         - \psi\big(\tfrac12 + \ii \beta \big) \Big),\\\label{into2}
 &
\ge\int_{-2}^2{\bf S}_{\alpha,0}(w)^2(1-w^2)^\alpha\dd 2w\ 
                                                     =\ 
                                                     \frac{ 2^{2\alpha+1 } \left( \pi^2-2\psi'\big(\frac12+\alpha\big) \right)}
                                                     {\sin(\pi\alpha)\Gamma(\frac12+\alpha)^2},  \end{align}
For $|\Re \alpha|<1$ the integrals \eqref{into1} and \eqref{into2} 
are standard. For $\alpha\in\nn$ we have anomalous 
generalized integrals:
\begin{align}
  \label{eq:formo_integer1_limit1.}
 & \gen\int_{-2}^2{\bf S}_{\alpha ,\ii\beta}(w)^2(1-w^2)^\alpha \dd 2w 
 \\\notag 
&\qquad=  \frac{(-1)^\alpha  2^{2\alpha +2} \cosh(\pi\beta)}{\pi
\Gamma(\frac12+\alpha+\ii\beta)\Gamma(\frac12+\alpha-\ii\beta)}
      \Bigg( 
       \frac{\ii }{2 \beta } \Big(
       \psi'\big(\tfrac{1}{2}+\alpha +\ii \beta\big) 
    - \psi'\big(\tfrac{1}{2}+\alpha -\ii \beta\big)   
    \Big) \\\notag
   &\qquad-\frac{\ii}{2\beta}\Big(H_{\alpha}\big(\tfrac12+\ii\beta)-H_{\alpha}\big(\tfrac12-\ii\beta)\Big)
   \ln 4
    \\   &\qquad\notag+
  \sum_{k=0}^{\alpha -1} 
 \frac{\psi \big(-\tfrac12-k+\ii \beta \big) 
    +  \psi\big(-\tfrac12-k-\ii \beta \big) 
    -\psi(\alpha -k)-\psi(1+k) 
 }{\big(\tfrac12 +k\big)^2+\beta^2
  }
\Bigg), 
\end{align}
\begin{align}
 \label{eq:formo_integer1_limit2.}
  \gen\int_{-2}^2{\bf S}_{\alpha ,0}(w)^2(1-w^2)^\alpha \dd 2 w\ 
 &=  \frac{2^{2\alpha +2} (-1)^{\alpha }}{\pi \Gamma\big(\tfrac12 +\alpha \big)^2 } 
 \Bigg( - \psi''\big(\tfrac12 +\alpha \big) 
     +H_{\alpha}'\big(\tfrac12\big)\ln 4\\\notag
   &\quad+
   \sum_{k=0}^{\alpha -1} 
 \frac{2\psi \big(-\tfrac12-k \big) 
    -\psi(\alpha -k)-\psi(1+k)
 }{\big(\tfrac12 +k\big)^2
  }
\Bigg) .
\end{align}
If  $\alpha=0$, then \eqref{eq:formo_integer1_limit1.} and 
\eqref{eq:formo_integer1_limit2.} are still true as standard
integrals. Besides, they greatly simplify:
\begin{align}
    \int_{-2}^2{\bf S}_{0,\ii\beta}(w)^2\dd 2w&
  =   \frac{ 2\ii \cosh^2(\pi\beta)
    \Big(\psi'(\frac12+\ii\beta)-\psi'(\frac12-\ii\beta)\Big)}
{\beta\pi^2},\\
\int_{-2}^2{\bf S}_{0,0}(w)^2\dd 2 w&=  -\frac{4  \psi''\big(\frac12\big)}{\pi^2}.
\end{align}

Here are generalized integrals of squares of $\mathbf{Z}$ 
functions for  $\Re\lambda>0$, as computed in \cite{DGR23a}. For 
$\alpha\in\cc\setminus\zz$ we have non-anomalous integrals.
  \begin{align}\label{into3}
\ge\int_2^\infty{\bf Z}_{\alpha,\lambda}(w)^2(w^2-1)^\alpha\d 2 w 
= \frac{2^{2 \lambda } (\psi(\tfrac12 + \alpha + \lambda) 
- \psi(\tfrac12 -\alpha + \lambda)) 
}{\lambda \sin\pi\alpha \Gamma(\tfrac12 -\alpha + \lambda) 
\Gamma(\tfrac12+\alpha+\lambda)}. 
  \end{align}
  For  $|\Re\alpha|<1$, \eqref{into3} is a standard integral.

  For $\alpha\in\zz\setminus\{0\}$ and $\Re\lambda>0$, we have anomalous integrals:
\begin{align}
\label{eq:genintZZ_limit.}
   &\gen\int_2^\infty{\bf Z}_{\alpha,\lambda}(w)^2(w^2-1)^\alpha\d 2w 
   = 
  \frac{(-1)^{\alpha} 2^{2\lambda+1}}{\pi\Gamma\big(\tfrac12-\alpha+\lambda\big)
        \Gamma\big(\tfrac12+\alpha+\lambda\big)}
 \\ \notag 
  &\hspace{12ex}\times \Bigg(
  \frac{\psi'(\tfrac12-\alpha+\lambda)
  +\psi'(\tfrac12+\alpha+\lambda)
    }{2\lambda}+\frac{H_{|\alpha|}(\tfrac12-\lambda)-
    H_{|\alpha|}(\tfrac12+\lambda)}{2\lambda}\ln 4\\
\notag& \hspace{12ex}   + \sum_{k=0}^{|\alpha|-1} 
     \frac
   {\psi(\tfrac32+k+\lambda)
                 +\psi(-\tfrac12-k+\lambda)
   -\psi(|\alpha|-k)
   -\psi(1+k)
     }{\lambda^2-\big(\tfrac12+k\big)^2} 
    \Bigg).
 \end{align}

 If $\alpha=0$, then \eqref{eq:genintZZ_limit.} is still true in the
 sense of standard integrals. Besides, it greatly simplifies:
\begin{align} \int_2^\infty{\bf Z}_{0,\lambda}(w)^2 \dd 2 w = \frac{2^{2 \lambda+1}
\psi'(\tfrac12 + \lambda)}{\pi \lambda \Gamma(\tfrac12 + \lambda)^2}.
    \end{align}

The Gegenbauer functions have the following asymptotics 
\cite{DGR23a} (see also \cite{NIST,Olver}):
\begin{theoreme}
\label{thm:asymptotics}
Let $\alpha \geq -\frac{1}{2}$ and $\pi>\delta>0$ be fixed. Then
we have the following estimates:
\begin{enumerate}
 \item Uniformly for $\theta\in [0,\pi-\delta]$ and  $\beta\to\infty$, 
\begin{align} \label{asym3}
\frac{\pi\e^{-\pi\beta}(\sin\theta)^{\alpha+\frac12}
}{2^\alpha\theta^{\alpha+\frac12}}
  {\bf S}_{\alpha,\pm\ii\beta}(-\cos\theta)
  &
 =(\theta\beta)^{-\alpha} K_\alpha(\beta \theta)\big(1+\mathcal{O}(\beta^{-1})\big).  
\end{align}
\item Uniformly
for $\theta\geq0$ and 
$\lambda\to\infty$,
\begin{align}\label{asym1}
\frac{\sqrt\pi\Gamma(\tfrac12-\alpha+\lambda)(\sinh \theta)^{\alpha+\frac12}}{2^{\lambda+\frac12}  
\theta^{\alpha+\frac12}}
{\bf Z}_{\alpha,\lambda}(\cosh\theta)  
&=
(\lambda\theta)^{-\alpha} K_\alpha(\lambda \theta)\big(1+\mathcal{O}(\lambda^{-1})\big).
\end{align}\end{enumerate}
\end{theoreme}
Correspondingly, the bilinear generalized integrals of Gegenbauer 
functions that are needed to determine the Green's functions of 
point-like perturbations have the expected asymptotics \cite{DGR23a}: 
\begin{theoreme}
 \label{thm:asymp_genintSZ}
 For $\beta,\lambda\to\infty$, we have
 \begin{align}
 \label{eq:Sint_asymp} 
&\frac{\pi^2 {\e}^{-2\pi\beta} \beta^{2\alpha}}{2^{2\alpha}}
  \;\gen\int_{-2}^2{\bf S}_{\alpha,\ii\beta}(w)^2(1-w^2)^\alpha\dd 2 w\\
  &=   \Big(1+\mathcal{O}\big(\tfrac{1}{\beta}\big)\Big)
    \;\gen\int_0^\infty  K_{\alpha}(\beta r)^2 2r \dd r, 
    &&\qquad \Re(\alpha)>-1,
  \notag\\ \label{eq:Zint_asymp}
&\frac{\pi \Gamma\big(\tfrac12+\alpha+\lambda\big)^2
  }{2^{2\lambda+1}\lambda^{2\alpha}}
  \;\gen\int_2^\infty{\bf Z}_{\alpha,\lambda}(w)^2(w^2-1)^\alpha\dd 2 w \\
 &=  \Big(1+\mathcal{O}\big(\tfrac{1}{\lambda}\big)\Big)
    \;\gen\int_0^\infty  K_{\alpha}(\lambda r)^2 2r \dd r, 
    &&\qquad \alpha\in\cc.
\notag  \end{align}
\end{theoreme}

\section{Pochhammer symbols and harmonic numbers}
  \label{poch}
  
The {\em Pochhammer symbol}, as it is usually defined, is a 
generalization of the factorial:
\beq\label{poch1}
(a)_n:=a(a+1)\cdots(a+n-1),\qquad (1)_n=n!.\eeq
We will also use the {\em harmonic numbers}
\beq\label{poch2}
H_n(a):=\frac1a+\frac1{a+1}+\cdots+\frac1{a+n-1},\qquad H_n:=H_n(1).\eeq
Sometimes it is convenient to extend the definitions \eqref{poch1} and
 \eqref{poch2} to  complex parameters $z$:
\begin{align}
  (a)_z :=\frac{\Gamma(a+z)}{\Gamma(a)}, \quad
  H_z(a):=\psi(a+z)-\psi(a),\quad z\in\cc\setminus(-a-\nn_0).
\end{align}
We have
\begin{align}
\partial_a(a)_z=(a)_zH_z(a).
\end{align}
For $n\in\nn$ we have the useful identities
\begin{align}
 (-1)^n (\tfrac12-a)_n(\tfrac12+a)_n&= (a-n+\tfrac12)_{2n}
 =                                                         \prod_{j=0}^{n-1}\Big(a^2-(\tfrac12+j)^2\Big),\\ \label{eq:Poch_id2}
 (-1)^n (\tfrac12-a)_{n+\12}(\tfrac12+a)_{n+\12}&=  \cot(\pi a)(a-n)_{2n+1}=
   \cot(\pi a) a      \prod_{j=1}^{n}\Big(a^2-j^2\Big).
\end{align}

\paragraph{Acknowledgement.}
The research
 was supported by National Science Center of Poland under the
    grant UMO-2019/35/B/ST1/01651. J.D. would like to thank
    Howard Cohl for an inspiring discussion at an early stage of this work.

\end{document}